\documentclass[11pt]{article}

\usepackage{amsfonts}
\usepackage{amssymb}
\usepackage{amsmath}
\usepackage{amsthm}
\usepackage{geometry}
\usepackage{commath}
\usepackage{tabularx,booktabs}
\usepackage{lscape}
\usepackage{color}
\usepackage{colortbl}
\usepackage{xr}

\usepackage[longnamesfirst]{natbib}

\newcommand{\myref}[2]{\hyperref[#1]{#2}}
\numberwithin{equation}{section}

\usepackage{latexsym}
\usepackage{adjustbox}
\usepackage{graphicx} 
\usepackage{enumerate}
\usepackage{setspace}
\usepackage{multirow}

\usepackage[toc,title,titletoc,header]{appendix}
\usepackage[bottom]{footmisc} 
\usepackage[pdfborder={0 0 0}]{hyperref}
\definecolor{fede_color}{cmyk}{0,0.9,0.33,.22}
\usepackage{siunitx}
\usepackage{lineno}
\setpagewiselinenumbers

\newtheorem{theorem}{Theorem}[section]
\newtheorem{lemma}{Lemma}[section]

\theoremstyle{definition}

\theoremstyle{remark}

\newcounter{assumptionM}
\newcounter{assumptionA}
\def\theassumptionM{M.\arabic{assumptionM}}

\newenvironment{runningexample}[1][]{%
    \vspace{0.5\baselineskip} 
    \noindent\textbf{Running Example.}\ #1%
    \begin{itshape}%
}{%
    \end{itshape}%
    \vspace{0.5\baselineskip} 
}

\newenvironment{runningexamplecont}[1][]{%
    \vspace{0.5\baselineskip} 
    \noindent\textbf{Running Example, continued.}\ #1%
    \begin{itshape}%
}{%
    \end{itshape}%
    \vspace{0.5\baselineskip} 
}

\hypersetup{
 colorlinks=true,linkcolor=blue, citecolor=blue, filecolor=blue, 
 linkbordercolor={0 1 1}, citebordercolor={1 0 0},
 pdftitle={Stationarity},
 pdfauthor={Bugni and Wang},
 pdfcreator={\LaTeX\ with package \flqq hyperref\frqq},
 pdfsubject={},
}

\begin{document}
\sloppy
\title{Inference in Auctions with Many Bidders Using Transaction Prices\thanks{
First ArXiv date: November 16, 2023. We thank the Editor, Peter Hull, and three anonymous referees for their valuable comments.
We also thank \'Aureo de Paula, Christian Gourieroux, Joel Horowitz, Yu-Chin Hsu, Tong Li, Yao Luo, Ulrich M\"uller, Nicholas Papageorge, Rob Porter, and Daiqiang Zhang for very helpful comments and suggestions, and Vinci Chow for kindly sharing the Hong Kong car license auction dataset. We also thank the participants at numerous seminars and conferences. All remaining errors are our own.}
}
\author{Federico A. Bugni\\ Department of Economics\\ Northwestern University\\ \url{federico.bugni@northwestern.edu}
\and
Yulong Wang\\ Department of Economics\\ Lehigh University\\ \url{yuw925@lehigh.edu}}
\date{\today}
\maketitle

\begin{abstract}
This paper studies inference in first-price and second-price sealed-bid auctions with many bidders, using an asymptotic framework where the number of bidders increases while the number of auctions remains fixed. Our approach enables asymptotically exact inference on key features, such as the winner’s expected utility, the seller’s expected revenue, and the tail of the valuation distribution, using only transaction price data. Our simulations demonstrate the accuracy of the methods in finite samples. We apply our methods to Hong Kong vehicle license auctions, focusing on high-priced, single-letter plates. Other relevant applications include online and art auctions.

\end{abstract}

\newpage


\setcounter{page}{1} \renewcommand{\thepage}{\arabic{page}} 
\renewcommand{\thefootnote}{\arabic{footnote}} 

\renewcommand{\baselinestretch}{1.2} \small \normalsize

\newpage

\setcounter{page}{1}

\section{Introduction}

This paper considers inference in first-price and second-price sealed-bid auctions in empirical settings where we observe a small number of auctions, each with a large number of bidders. The abundance of bidders in each auction prompts us to consider a novel asymptotic framework in which the number of bidders diverges, while the number of auctions remains small and fixed. This framework differs substantially from the more conventional approach in which the researcher observes multiple bids from a large number of independent and identically distributed (i.i.d.) auctions. See \citet{athey2002,hailetamer2003,athey2007handbook,guerre2000}, among others. Since our asymptotic framework does not require the number of auctions to diverge, our analysis is suitable for applications with substantial heterogeneity across auctions, even when the number of truly homogeneous auctions is limited. 

Within our novel asymptotic framework, we introduce new inference methods for the winner's expected utility, the seller's expected revenue, and the tail behavior of the valuation distribution. We show that the latter can be used to test the regularity conditions commonly assumed in auction literature. Our data requirements are minimal; our methods rely on observing transaction prices from a finite number of auctions. In particular, we do not need to observe multiple bids or the number of participating bidders in these auctions.
Furthermore, our analysis does not attempt to estimate bid distributions and is therefore not designed for commonly adopted two-step procedures \citep[e.g.,][]{guerre2000} that require many auctions and rich bid-level data.

Our methodology is well-suited to auction settings in which the number of participants is known to be large, and we only observe transaction prices for a possibly small number of homogeneous auctions. We illustrate our methodology using data from Hong Kong license plate auctions. Since 1973, the Hong Kong government has held monthly standard oral-ascending price auctions to sell vehicle license plates. Recent interviews with auction participants reveal that these auctions are well attended, with around 100 participants per auction. Only the transaction price from each auction is recorded.  Finally, license plates are inherently heterogeneous goods whose value depends on their configuration of letters and numbers. If we focus on license plates with specific configurations to get a homogeneous sample of auctions, the sample size is usually small. We use the Hong Kong vehicle license plate auction setup as our running example throughout the paper and our empirical application in Section \ref{sec:application}. Beyond this application, our econometric framework is also relevant to many other auctions, including art and online auctions.

In our benchmark model, we consider second-price sealed-bid auctions with a large number of symmetric bidders holding independent private values (IPV). Unlike traditional frameworks that rely on a growing number of auctions, transaction price data alone does not fully identify the valuation distribution. However, as the number of bidders diverges, the transaction price data allows us to characterize the tail properties of the bid distribution. In turn, these can be used to conduct inference on auction fundamentals, such as the winner's expected utility and the seller's expected revenue. Additionally, analyzing tail behavior allows us to assess whether the valuation distribution has bounded support and positive density at its upper limit, which are common assumptions in auction theory \citep[e.g.,][]{MaskinRiley1984,guerre2000,GuerreLuo2022}. The proposed inference methods in our framework control size and demonstrate optimal power properties. Specifically, we use results in \cite{elliott2015} to guarantee that our confidence intervals minimize expected length and our hypothesis tests maximize weighted average power. 

Beyond the benchmark model, we consider two extensions. First, we show that our analysis for second-price auctions extends to first-price auctions. Second, we broaden our analysis to encompass potential bidders, conditional IPV scenarios, and reserve price. The presence of potential bidders often occurs when some buyers' bids remain unsubmitted or unobserved. This scenario complicates standard methods, particularly when the count of potential bidders is unknown. In contrast, our analysis remains unaffected, as the tail of the underlying valuation distribution and the equilibrium bidding function remain unchanged. Regarding the conditional IPV setup, we accommodate auction-level unobserved heterogeneity $A$, conditional on which buyers possess independent values. This scenario, known as conditional IPV, has been explored by \citet{li2000cipv}. With only a limited number of auctions, existing methods struggle to identify the underlying valuation distribution $F_{V|A}$ conditional on this heterogeneity. However, our methodology can be used to discern the tail of $F_{V|A}$, provided we observe at least three bids from a single auction.



The rest of the paper is organized as follows. The following subsection reviews the existing literature and discusses our contribution. Section \ref{sec:setup} establishes the new asymptotic framework with many bidders and discusses its relationship with extreme value (EV) theory. Section \ref{sec:sp} focuses on second-price auctions and introduces our new inference methods. Section \ref{sec:sp-supplus} presents confidence intervals for the winner's expected utility, Section \ref{sec:sp-revenue} introduces confidence intervals for the seller's expected revenue, and Section \ref{sec:sp-index} outlines hypothesis tests for the tail index. Section \ref{sec:extension} extends the analysis to first-price auctions, the conditional IPV model, and the presence of reserve prices. We provide Monte Carlo simulations in Section \ref{sec:simulation}. Section \ref{sec:application} applies our inference methods to data from Hong Kong license plate auctions, and Section \ref{sec:conclusion} concludes. The paper's supplementary appendix provides all proofs, auxiliary results, and computational details.

\subsection{Related literature}\label{sec:literature}
As mentioned earlier, our paper considers auctions in an asymptotic framework where the number of bidders diverges. This framework has been frequently used in the economic theory literature \citep[e.g.,][]{hongshum2004,virag2013,ditillio2021} but less so in econometrics. \citet{krasnokutskaya2022estimating} considers a latent group structure on the set of agents and allows both the number of agents and the number of markets to grow. \citet{menzel2013} shows that the nonparametric estimator of the valuation distribution may become irregular and perform poorly when the number of bidders increases. This issue does not affect our method, which relies on the EV approximation. 

In comparison, much of the existing econometric literature has focused on identifying and estimating the valuation distribution in a finite number of bidders and a diverging number of auctions. The seminal work by \citet{athey2002} derives general results on identifying the valuation distribution from the bid distribution. \citet{hailetamer2003} investigates English auctions and establishes bounds on the valuation distribution and other objects of interest with minimal structural assumptions. \cite{chesher2017wp} extends these bounds to the non-IPV setup. 
\citet{aradillas2013} nonparametrically identifies bounds on seller profit and bidder surplus while accounting for variations in the number of bidders across auctions. \citet{BrendstrupPaarsch2006} and \citet{Komarova2013} derive nonparametric identification of the valuation distribution using transaction prices and the winner's identity.

The existing auction literature has considered different types of data, including submitted bids, the number of actual bidders, and the number of potential bidders. For example, \cite{Li2005} studies first-price auctions with entry and binding reservation prices and estimates the valuation distribution with the observed bids and the number of actual bidders. Without knowing the number of bidders, \cite{an2010} proposes using a proxy of the number of bidders and an instrument variable. \citet{KimLee2014}, \citet{song2015}, \citet{Mbakop2017}, and \citet{freyberger2022} construct identification of the valuation distribution with two or more order statistics of bids. \cite{shneyerov2011} derives nonparametric identification of model primitives based on finitely many groups of bidders. Recently, \citet{LuoXiao2023} derives identification results with two consecutive order statistics and an instrument or three consecutive ones. All these methods of identification, estimation, and inference are based on the traditional asymptotic framework with many auctions and multiple bids per auction. We refer to \citet{hickman2012} and \citet{gentry2018review} for recent surveys.

As already mentioned, our paper provides inference based only on transaction prices. This aspect of our paper resembles the recent work by \cite{GuerreLuo2022}. However, there are considerable differences between this paper and our contribution. In particular, \cite{GuerreLuo2022} considers first-price auctions in which the number of bidders in each auction is random and has finite support. In this context, they establish that the winning bid increases with the number of bidders and that the density of winning bids exhibits discontinuities as the number of bidders changes. Our framework differs significantly from \citet{GuerreLuo2022}. First, their identification strategy relies on the assumption that the number of bidders is finite, whereas we focus on the case where it diverges. Thus, the contributions are designed for different empirical environments. Second, their argument does not apply to second-price auctions, while ours does. Third, their method relies on observing a large number of auctions, whereas ours can handle applications with a finite number of auctions. Finally, their paper delivers identification analysis, while ours focuses on inference.

Finally, our paper is also connected to the literature on testing in auction models. \citet{DonaldPaarsch1996} introduces parametric tests within the context of IPV setups. \citet{HaileHongShum2003} devises nonparametric tests for common values in first-price auctions. \citet{JunPinkseWan2010} develops a nonparametric test for affiliation. \citet{HillShneyerov2013} develops a test for common values in first-price auctions utilizing tail indices. \citet{LiuLuo2017} puts forward a nonparametric test for comparing valuation distributions in first-price auctions. Compared with these studies, our test in Section \ref{sec:sp-index} serves as a model specification test. Specifically, we consider testing whether the support of the valuation distribution is bounded and whether the density is strictly positive near the upper endpoint, both of which are typically imposed as regularity conditions in the existing literature.

\section{Asymptotic framework with many bidders}\label{sec:setup}

We consider inference in sealed-bid auctions for a single object, where the data consist of transaction prices from $n \geq 3$ independent auction realizations, denoted as $\{P_j:j=1,\dots,n\}$. Importantly, our framework does not require $n$ to diverge to infinity.

For each auction $j=1,\dots,n$, the setup is as follows. There is a single object for sale, and $K_j$ potential buyers are bidding for it. These bidders have independent private values (IPV) $\{V_{i,j}:i=1,\dots,K_j\}$ distributed according to a common cumulative distribution function (CDF) $F_V$ with support on $[v_L,v_H]$, with $0\leq v_L<v_H$, where $v_H=\infty$ is allowed. We assume that $F_V$ is strictly increasing on its support and admits a continuous probability density function (PDF) $f_V=F_V'$. Bidders are assumed to be risk-neutral and to maximize expected profits, without liquidity or budget constraints. $F_V$ and $K_j$ are common knowledge to all bidders, but unknown to the researcher. Our inference methods rely on asymptotics as the number of bidders $K_j$ diverges while the number of auctions $n$ remains finite. To this end, we assume that $K \equiv \min\{K_1,\dots,K_n\} \to \infty$ and $K_i/K \to 1$ for each $i=1,\dots,n$. That is, we assume all $n$ auctions have approximately the same large number of bidders.\footnote{Allowing for $K_j/K \not\to 1$ for some $j=1,\dots,n$ significantly complicates our inference problem. In particular, this would require observing more than the transaction price in each auction, which would contradict the core premise of this paper.}

\begin{runningexample}
Throughout this paper, we illustrate our methodology using a specific example from Hong Kong license plate auctions. Specifically, we focus on auctions for the luxurious plates with a single letter and no numbers. Table \ref{tab:plate} presents all historical records of such auctions.  
\begin{table}[ht]
\centering
\renewcommand{\arraystretch}{1.6}
\begin{tabular}{ccccc}
\hline\hline
Plate &  & Auction date &  & Price (million HKD) \\ \hline
\texttt{D} &  & Feb.\ 25, 2024 &  & $20.2$ \\ 
\texttt{R} &  & Feb.\ 12, 2023 &  & $25.5$ \\ 
\texttt{W} &  & Mar.\ 7, 2021 &  & $26.0$ \\ 
\texttt{V} &  & Feb.\ 24, 2019 &  & $13.0$ \\ 
\hline\hline
\end{tabular}
\caption{All recorded auctions with a single letter and no numbers in Hong Kong. The transaction prices are expressed in millions of Hong Kong dollars (approximately equal to \$128,000 USD).}
\label{tab:plate}
\end{table}

\noindent The exceptionally high transaction prices have sparked considerable media attention.\footnote{For example, see https://hobbylistings.com/hong-kong-license-plate and https://www.cnn.com/style/article/\linebreak hong-kong-auctioned-vanity-car-plates-intl-hnk/index.html.} The methods developed in this paper allow us to conduct inference on features of these auctions. 
\end{runningexample}

We make an additional assumption about the valuation distribution $F_V$. We assume that it is in the domain of attraction of the EV distribution $G_\xi$, where $\xi$ denotes the tail index. Formally, this means that there is a sequence of normalizing constants $\{(a_{K},b_{K}) \in \mathbb{R}_{++} \times \mathbb{R} :K \in \mathbb{N}\}$ such that, for all $x$ that is a continuity point of $G_\xi$,
\begin{equation}
\lim_{K \to \infty} ~(F_V(a_{K} x + b_{K}))^{K} ~~=~~ G_\xi(x) .
\label{eq:DomainOfAttraction}
\end{equation}
See \citet[Chapter 1]{dehaan2006book} or \citet[Chapter 10]{david2004book} for recent expositions on this topic, including sufficient conditions for \eqref{eq:DomainOfAttraction}. Under the condition in \eqref{eq:DomainOfAttraction}, standard asymptotic results imply that $G_\xi$ belongs to one of three types: Weibull (if $\xi < 0$), Gumbel (if $\xi = 0$), or Fr\'echet (if $\xi > 0$). We can unify these distributions into the generalized EV distribution, with the following CDF:
\begin{equation}
\label{eq:G}    
G_\xi(x) ~=~\left\{
\begin{array}{ll}
\exp( - (1 + \xi x)^{-1/\xi} ) I(1 + \xi x > 0)&\text{ if }\xi >0,\\
\exp(-\exp(-x)) &\text{ if }\xi =0,\\
\exp( - (1 + \xi x)^{-1/\xi} ) I(1 + \xi x > 0) +  I(1 + \xi x \leq 0) &\text{ if }\xi <0.
\end{array} 
\right.
\end{equation}
Condition \eqref{eq:DomainOfAttraction} is an arguably mild restriction, as is satisfied by most commonly used valuation distributions $F_V$. The case with $\xi >0$ covers distributions with unbounded support (i.e., $v_H=\infty$) and polynomial decaying (i.e., ``heavy'') right tails. In this case, moments of order less than $1/\xi$ exist, and moments of order greater than $1/\xi$ do not (see \citet[page 176]{dehaan2006book}). Then, the restriction to $\xi \leq 1/2$ implies that $F_V$ has finite second moments. Examples include Pareto, Student-t, and F distributions. 
Second, the case with $\xi =0$ encompasses distributions with unbounded support (i.e., $v_H=\infty$) but with exponentially decaying (i.e., ``light'') right tail and bounded moments of any order. Examples include normal, Log-Normal, Weibull, and Gamma distributions. Finally, the case $\xi<0 $ covers distributions with bounded support (i.e., $v_H <\infty$), such as Beta, Uniform, and triangular distributions. In turn, condition \eqref{eq:DomainOfAttraction} fails for any distribution that has a probability mass point at the highest value of its support, such as geometric or Poisson distributions; see \citet[Exercise 1.13]{dehaan2006book}.

Our domain-of-attraction assumption \eqref{eq:DomainOfAttraction} restricts only the upper-tail behavior of the valuation distribution and does not impose a parametric structure on $F_V$ as a whole. It captures the idea that, in auctions with many bidders, transaction prices and revenue are driven primarily by extreme valuations rather than typical ones. In simulations, we show that our inference performs well across a range of empirically relevant tail behaviors.

The significance of \eqref{eq:DomainOfAttraction} in our paper is that it allows us to characterize the joint distribution of the ordered valuations for all auctions as the number of bidders diverges. We now introduce the relevant notation. For each auction $j=1,\dots,n$, let $\{V_{(i),j}:i=1,\dots,K_j\}$ denote the order statistics of $\{V_{i,j}:i=1,\dots,K_j\}$ in decreasing order, i.e., $V_{(1),j} \geq V_{(2),j}\geq \dots \geq V_{(K_j),j}$. Lemma \ref{lem:joint} provides the joint distribution of the extreme order statistics for all auctions.

\begin{lemma}\label{lem:joint}
Assume \eqref{eq:DomainOfAttraction} holds. For any $n \in \mathbb{N}$ and any $d \in \mathbb{N}$, and as $K\to \infty$,
\begin{align}
&\Big\{\Big( \frac{V_{(1),j} -b_{K}}{a_{K}}, ~\frac{V_{(2),j} -b_{K}}{a_{K}},~\dots~,~ \frac{V_{(d),j} -b_{K}}{a_{K}}\Big):j=1,\dots,n\Big\}~\overset{d}{\to}~\notag\\
&\Big\{\Big( H_{\xi}(E_{1,j}),~H_{\xi}(E_{1,j}+E_{2,j}),~\dots~,~H_{\xi}\big(\sum\nolimits_{s=1}^{d}E_{s,j}\big)\Big):j=1,\dots,n\Big\}, \label{eq:EVT}
\end{align}
where $\{(a_{K},b_{K}) \in \mathbb{R}_{++} \times \mathbb{R} :K \in \mathbb{N}\}$ are the normalizing constants in \eqref{eq:DomainOfAttraction},  $\{E_{s,j}:s=1,\dots,d,~j=1,\dots,n\}$ are i.i.d.\ standard exponential random variables, and
\begin{equation}
    H_{\xi}(x) ~\equiv~  \bigg\{ 
    \begin{array}{cc}
         ({x}^{-\xi} -1)/{\xi}&  \text{ if }\xi\neq 0,\\
         - \ln(x) &  \text{ if }\xi= 0.
    \end{array} 
    \label{eq:H_function}
\end{equation}
\end{lemma}

Lemma \ref{lem:joint} characterizes the asymptotic distribution of the largest order statistics. In this paper, we consider first-price and second-price auction formats, which involve only $V_{(1),j}$ and $V_{(2),j}$, respectively. We assume symmetric equilibrium bidding, enabling us to establish a relationship between private valuations, equilibrium bids, and transaction prices ${\bf P}\equiv\{P_{j}:j =1,\dots,n\}$. As a corollary, we can completely describe the asymptotic distribution of transaction prices in terms of the tail index $\xi$. This, in turn, allows us to perform inference of several objects of economic interest solely based on the transaction prices.

\begin{runningexamplecont}
Recall our dataset of $n=4$ license plate auctions with a single letter. According to Table \ref{tab:plate}, the transaction prices expressed in million HKD are $\mathbf{P} = \{P_1,P_2,P_3,P_4\} = \{20.2,~ 25.5,~ 26.0,~ 13.0\}$. The number of bidders, $K_j$, is unobserved but is known anecdotally to be large for all auctions $j=1,2,3,4$. This empirical setup aligns well with our asymptotic framework, which considers a diverging $K = \min\{K_1,K_2,K_3,K_4\}$ and a fixed, small $n$. Beyond the IPV assumption, we only assume that the underlying value distribution $F_V$ satisfies \eqref{eq:DomainOfAttraction}, which, as noted earlier, encompasses a broad class of valuation distributions.
\end{runningexamplecont}

\section{Second-price auctions}\label{sec:sp}

We begin our analysis with second-price sealed-bid auctions, in which the highest bidder wins the object and pays the second-highest bid. Since we consider a private-value framework, second-price auctions are equivalent in a weak sense to open ascending-price (or English) auctions (see \citet[page 4]{krishna2009book}). By standard arguments (e.g., \citet[Proposition 2.1]{krishna2009book}), the weakly dominant strategy for a bidder with valuation $v$ in auction $j=1,\dots,n$ is
\begin{equation}
\beta_j(v)~=~v.  \label{eq:bidSecond}
\end{equation}
Thus, the observed transaction price in auction $j$ equals the second-highest bid, i.e., 
\begin{equation}
P_{j}~=~V_{( 2) ,j}.  \label{eq:secondprice}
\end{equation}
By Lemma \ref{lem:joint} and \eqref{eq:secondprice}, we conclude that as $ K\to \infty $, 
\begin{equation}
\Big\{ \frac{P_{j}-b_{K}}{a_{K}}:j=1,\dots,n\Big\} ~\overset{d}{ \to }~ \{ Z_{j}:j=1,\dots,n\} , 
\label{eq:EVT2}
\end{equation}
where $\{ Z_{j}:j=1,\dots,n\}$ is i.i.d.\ with $Z_{j}\equiv H_{\xi }(E_{1,j}+E_{2,j})$ for each $j=1,\dots,n$, and $ \{(a_{K},b_{K})\in \mathbb{R}_{++}\times \mathbb{R}:K\in \mathbb{N}\}$, $\{(E_{1,j},E_{2,j}):~j=1,\dots,n\}$, and $H_{\xi }$ are as in Lemma \ref{lem:joint}.

If the constants $\{(a_{K},b_{K})\in \mathbb{R}_{++}\times \mathbb{R}:K \in \mathbb{N}\}$ were known, we could use \eqref{eq:EVT2} to perform inference on functions of the EV index $\xi$. Unfortunately, these constants are unknown and depend implicitly on the underlying distribution of valuations. Specifically, $a_K$ characterizes the scale of the largest valuations, i.e., how fast it approaches $v_H$ as $K\to\infty$, and $b_K$ characterizes the location, i.e., a constant shift. When $\xi>0$, $v_H$ is infinite, and $P_j$ diverges almost surely. The scale $a_K$ is typically proportional to $K^{\xi}\to\infty$ as $K\to\infty$, and $b_K$ can be zero. When $\xi<0$, $v_H$ is finite, and $v_H - P_j$ converges to zero almost surely. The scale $a_K$ is again proportional to $K^{\xi}$, which now shrinks to zero as $K\to\infty$. The location $b_K$ is then the constant $v_H$. 

\begin{runningexamplecont}
Compared with the price of ordinary plates, the exceptionally high transaction prices of these luxurious plates in Table \ref{tab:plate} suggest that $v_H$ could be infinite. This would imply that $\xi > 0$ and rule out the uniform distribution. Consequently, the winner’s expected utility and the seller’s expected revenue may be substantial, and are better modeled as diverging as $K \to \infty$. In contrast, existing nonparametric methods typically assume a bounded support with $v_H < \infty$, where $P_j$ approaches the upper bound $v_H$ for all $j$. Under this assumption, the winner's expected utility shrinks to zero, and the seller's expected revenue converges to $v_H$ as $K \to \infty$. We introduce a formal test to distinguish between these scenarios in Section \ref{sec:sp-index}. 
\end{runningexamplecont}

To sidestep the unknown $a_K$ and $b_K$, we sort the transaction prices across auctions (i.e., $P_{(1)}\leq P_{(2)}\leq \dots \leq P_{(n)}$), and consider the sorted and self-normalized prices: for $j=1,\dots, N \equiv n-2 \geq 1$, 
\begin{equation}
\tilde{P}_{j}~\equiv ~\Bigg\{
\begin{array}{cc}
     \frac{P_{(j+1)}-P_{(1)}}{P_{(n)}-P_{(1)}}&\text{ if }P_{(n)} > P_{(1)},  \\
    0&\text{ if }P_{(n)} = P_{(1)} , 
\end{array}
    \label{eq:P*}
\end{equation}
and let $\mathbf{\tilde{P}}=\{ \tilde{P}_{j}:j=1,\dots,N\} \in \Sigma \equiv \{h \in [0,1]^N : 0\leq h_1 \leq \dots \leq h_N \leq 1\}$. The following result characterizes the asymptotic distribution of $\mathbf{\tilde{P}}$ as $K\to \infty$. 
\begin{lemma}\label{lem:densitySecond}
Assume \eqref{eq:DomainOfAttraction} holds. For any $N\in \mathbb{N}$, and as $K \to \infty$,
\begin{equation}
\mathbf{\tilde{P}}=\{ \tilde{P}_{j}:j=1,\dots,N\} ~\overset{d}{\to }~\mathbf{\tilde{Z}}=\{ \tilde{Z}_{j}:j=1,\dots,N\} ,
\label{eq:P*joint}
\end{equation}
where the joint density of $\mathbf{\tilde{Z}}$ is
\begin{align}
&f_{\mathbf{\tilde{Z}}|\xi}( z_{1},\dots,z_{N} ) ~\equiv~ 1[ 0\leq z_{1}\leq \dots \leq z_{N}\leq 1] ~ (N+2)!~\Gamma ( 2( N+2) ) \times \notag \\
&~\left\{ 
\begin{array}{cc}
\int_{0}^{-1/\xi }s^{N}\exp \bigg(
\begin{array}{c}
-2( N+2) \ln ( \sum_{j=1}^{N}( 1+z_{j}\xi s) ^{-1/\xi }+(1+\xi s)^{-1/\xi }) \\
-( 1+2/\xi ) ( \sum_{j=1}^{N}\ln ( 1+z_{j}\xi s) +\ln (1+\xi s))
\end{array}
\bigg) ds & \text{if }\xi <0, \\
\int_{0}^{\infty }s^{N}\exp \bigg(
\begin{array}{c}
-2( N+2) \ln ( \sum_{j=1}^{N}\exp ( -z_{j}s) +\exp ( -s) ) \\
-2s( \sum_{j=1}^{N}z_{j}+1)
\end{array}
\bigg) ds & \text{if }\xi =0,\\
\int_{0}^{\infty }s^{N}\exp \bigg(
\begin{array}{c}
-2( N+2) \ln ( \sum_{j=1}^{N}( 1+z_{j}\xi s) ^{-1/\xi }+(1+\xi s)^{-1/\xi }) \\
-( 1+2/\xi ) ( \sum_{j=1}^{N}\ln ( 1+z_{j}\xi s) +\ln (1+\xi s))
\end{array}
\bigg) ds & \text{if }\xi >0,
\end{array}
\right. \label{eq:densitySecond}
\end{align}
and $\Gamma $ is the standard Gamma function.
\end{lemma}

Lemma \ref{lem:densitySecond} reveals that the asymptotic distribution of $\mathbf{\tilde{P}}$ is informative about the tail index $\xi $. In the next subsections, we show how to use this information to conduct asymptotically valid inference on the tail index $\xi $ and several other important features of these auctions, such as the winner's expected utility and the seller's expected revenue. Other tail features, such as the extreme quantile $F_{V}^{-1}(\tau)$ for some $\tau \approx 1$, can be similarly analyzed.  

\begin{runningexamplecont}
Provided that the valuation distribution $F_V$ satisfies Assumption \eqref{eq:DomainOfAttraction}, and as the number of bidders grows, the observed transaction prices satisfy
\begin{equation*}
    \frac{\mathbf{P}-b_K}{a_K}~=~ \left\{\frac{P_1-b_K}{a_K},\frac{P_2-b_K}{a_K},\frac{P_3-b_K}{a_K},\frac{P_4-b_K}{a_K}\right\} ~\overset{d}{\to }~ \left\{Z_1,Z_2,Z_3,Z_4 \right\},
\end{equation*}
where $\{Z_1,Z_2,Z_3,Z_4 \}$ are i.i.d.\ with distribution $H_{\xi}(E_{1,j}+E_{2,j})$ as in Lemma \ref{lem:joint}. If we sort $\{P_1,P_2,P_3,P_4\}$ in ascending order, we get $\{P_{(1)}\leq P_{(2)}\leq P_{(3)} \leq P_{(4)}\}= \{13.0,~ 20.2,~ 25.5,~ 26.0\}$.\footnote{We adjust these prices by inflation in Section \ref{sec:application}, but we ignore the inflation here for a simple illustration.}
The self-normalized vector $\tilde{\mathbf{P}}$ satisfies
\begin{equation*}
\tilde{\mathbf{P}} ~=~ \left\{\frac{P_{(2)}-P_{(1)}}{P_{(4)}-P_{(1)}},\frac{P_{(3)}-P_{(1)}}{P_{(4)}-P_{(1)}}\right\}~\overset{d}{\to }~\left\{\frac{Z_{(2)}-Z_{(1)}}{Z_{(4)}-Z_{(1)}},\frac{Z_{(3)}-Z_{(1)}}{Z_{(4)}-Z_{(1)}}\right\} ~=~ \tilde{\mathbf{Z}},
\end{equation*}
where the distribution of $\tilde{\mathbf{Z}}$ is uniquely characterized by the tail index $\xi$ as in Lemma \ref{lem:densitySecond} with $N=n-2=2$. 
Plugging in the observed transaction prices, we get
\begin{equation*}
\tilde{\mathbf{P}} ~=~ \left\{\frac{20.2~-~13.0}{26.0-~13.0},\frac{25.5~-~13.0}{26.0-~13.0}\right\}.
\end{equation*}
We treat these as random draws from $\tilde{\mathbf{Z}}$ and construct inference for $\xi$ and other auction fundamentals as functionals of $\xi$. Self-normalization allows us to bypass the unknown constants $a_K$ and $b_K$.
\end{runningexamplecont}

\subsection{Inference about the winner's expected utility}
\label{sec:sp-supplus}

Our objective is to conduct inference on the average of the winner's expected utility based on the transaction prices. Since each bidder bids their own valuation, the auction is won by the highest bidder, and the transaction price is equal to the second-highest bid, we conclude that the winner's expected utility in auction $j=1,\dots,n$ is $E[V_{(1),j}-P_{j}]=E[V_{(1),j}-V_{(2),j}] $, whose average is
\begin{equation}
\mu _{K}~=~ \frac{1}{n}\sum_{j=1}^{n}E[V_{(1),j}-V_{(2),j}].
\label{eq:winner1}
\end{equation}

\begin{runningexamplecont}
When the number of auctions is large, we can consistently estimate $E[V_{(1),j}-V_{(2),j}]$ using existing methods in the auction literature. However, this approach is not applicable in our current setting, where we have only four observations of single-letter license plates. On the other hand, this setup can be well represented by our asymptotic framework, which considers a diverging $K = \min\{K_1,K_2,K_3,K_4\}$ and fixed number $n$ of homogeneous auctions.
\end{runningexamplecont}

Given transaction prices $\mathbf{P}$, we consider a confidence interval (CI) for $\mu _{K}$ given by
\begin{equation}
U(\mathbf{P})~= ~( P_{(n)}-P_{(1)}) \times \tilde{U}(\mathbf{\tilde{P }}), \label{eq:CI_muK_defn}
\end{equation}
where $\mathbf{\tilde{P}}\in \Sigma $ are the sorted and self-normalized transaction prices in \eqref{eq:P*}, and $\tilde{U}:\Sigma\to \mathcal{P}(\mathbb{R})$ is a CI defined on $\mathbf{\tilde{P}}$. By \eqref{eq:CI_muK_defn}, the CI $U(\mathbf{P})$ is invariant to the sorting and translation of $\mathbf{P}$, and equivariant to their scale.

The remainder of the section denotes $ Z_{1,j}=H_{\xi }(E_{1,j})$, $Z_{2,j}=H_{\xi }(E_{1,j}+E_{2,j})$, $ \{(E_{1,j},E_{2,j}):j=1,\dots ,n\}$ are i.i.d.\ standard exponential random variables and $H_{\xi }$ is as in \eqref{eq:H_function}. We also use $ \mathbf{Z}=\{Z_{2,j}:j=1,\dots ,n\}$, $Z_{(n)}=\max \{Z_{2,j}:j=1,\dots ,n\}$, $Z_{(1)}=\min \{Z_{2,j}:j=1,\dots ,n\}$, for $j=1,\ldots ,N=n-2$,
\begin{equation}
\tilde{Z}_{j}~=~\bigg\{
\begin{tabular}{ll}
$\frac{Z_{( j+1) }-Z_{(1)}}{Z_{(n)}-Z_{(1)}}$ & if $ Z_{(n)}>Z_{(1)},$ \\
$0$ & if $Z_{(n)}=Z_{(1)},$
\end{tabular}
\label{eq:ZsortedScaled}
\end{equation}
and $\mathbf{\tilde{Z}}=\{\tilde{Z}_{j}:j=1,\dots ,N\} \in \Sigma$. Finally, let $Y_{\mu }\equiv E[Z_{1,1}-Z_{2,1}]/(Z_{(n)}-Z_{(1)})$ and $\kappa _{\xi }(\mathbf{\tilde{Z}} )=E[Z_{(n)}-Z_{(1)}|\mathbf{\tilde{Z}}]$. The distributions of these random variables are fully characterized by the tail index $\xi $. 
In particular, Lemma \ref{lem:sp-surplus} shows $ E[Z_{1,1}-Z_{2,1}] = \Gamma(1-\xi)$, where $\Gamma$ is the standard Gamma function. 
For the remainder of this section, we will use $P_{\xi}$ and $E_{\xi}$ to refer to the probability and expectation associated with this distribution. The following result describes the asymptotic properties of the CI in \eqref{eq:CI_muK_defn} as $K\to \infty$. 

\begin{theorem} \label{thm:CI_K} 
Assume \eqref{eq:DomainOfAttraction} holds, and that for some $\varepsilon >0$ with $(1+\varepsilon)\xi<1$, $E[|V_{i,j}|^{1+\varepsilon}]<\infty $ for all $i=1\dots,K_j$ in auction $j=1,\dots,N$. Finally, assume that the CI for $\mu _{K}$, $U(\mathbf{P})$, is as in \eqref{eq:CI_muK_defn} with $\tilde{U}:\Sigma\to \mathcal{P}(\mathbb{R}) $ that satisfies the following conditions:
\begin{enumerate}[(a)]
\item $P_{\xi}(\{Y_{\mu },\mathbf{\tilde{Z}}\}\in \partial \{(y,h)\in \mathbb{R} \times \Sigma:y\in \tilde{U}(h)\})=0$, where $\partial A$ denotes the boundary of $A$.
\item $\lg (\tilde{U}(h))<\infty $ for any $h\in \Sigma$, where $ \lg (A)$ denotes the length of $A$ (i.e., $\lg (A)\equiv \int \mathbf{1} [y\in A]dy$).
\item For any sequence $\{h_{\ell}\in \Sigma\}_{\ell\in \mathbb{N}}$ with $ h_{\ell}\to h\in \Sigma$, $\lg (\tilde{U}(h_{\ell}))\to \lg (\tilde{U}(h))$.
\end{enumerate}
Then, as $K\to \infty $,
\begin{enumerate}
\item $P(\mu _{K}\in U(\mathbf{P}))~\to ~P_{\xi }(Y_{\mu }\in \tilde{ U}(\mathbf{\tilde{Z}}))$,
\item $E[\lg (U(\mathbf{P}))]/a_{K}~\to ~E_{\xi }[\kappa _{\xi }( \mathbf{\tilde{Z}})\lg (\tilde{U}(\mathbf{\tilde{Z}}))]$.
\end{enumerate}
\end{theorem}

\vspace{0.5em}
Suppose that we consider the class of CIs for $\mu_K$ given by \eqref{eq:CI_muK_defn} and conditions (a)-(c) in Theorem \ref{thm:CI_K}. The finite sample properties of these CIs are unknown. If the number of bidders is large, it is natural to rely on the asymptotic behavior derived in Theorem \ref{thm:CI_K} to choose our CI. First, we can guarantee asymptotic validity by imposing
\begin{equation}
P_{\xi }( Y_{\mu }\in \tilde{ U}(\mathbf{\tilde{Z}})) ~\geq~ 1-\alpha ~~~\text{ for all }\xi \in \Xi .  \label{eq:asy_cov}
\end{equation}
Second, we can seek improvements in statistical power by choosing a CI that has a small asymptotic expected length (scaled by $a_{K}>0$). Since the tail index $\xi$ is unknown, we focus on the CI's asymptotic weighted length, given by
\begin{equation}
\int_{\xi \in \Xi} E_{\xi }[ \kappa _{\xi }(\tilde{\mathbf{{Z}}})\lg (\tilde U(\tilde{\mathbf{{Z}}})))] dW(\xi ),  \label{eq:asy_length}
\end{equation}
where $W$ is a user-defined weight function. We can combine both objectives by choosing the CI for $\mu_K$ that minimizes the asymptotic weighted length in \eqref{eq:asy_length} subject to the asymptotic validity condition in \eqref{eq:asy_cov}. 

Formally, let $\mathbb{U}$ denote the collection of CIs that satisfy conditions (a)-(c) in Theorem \ref{thm:CI_K}. Then, we propose choosing $\tilde{U}$ in \eqref{eq:CI_muK_defn} as the solution to the following problem:
\begin{equation}
\underset{\tilde{U} \in \mathbb{U}}{\arg\min} \int_{\xi \in \Xi} E_{\xi }[ \kappa _{\xi }(\tilde{\mathbf{{Z}}})\lg (\tilde{U}(\tilde{\mathbf{{Z}}})))] dW(\xi )~~\text{ s.t.}~~
P_{\xi }( Y_{\mu }\in \tilde{ U}(\mathbf{\tilde{Z}})) ~\geq~ 1-\alpha~~\text{ for all }\xi \in \Xi .
\label{eq:program}
\end{equation}
Following \citet{mullerwang2017}, we write down \eqref{eq:program} in its Lagrangian form:
\begin{equation}
\min_{\tilde{U} \in \mathbb{U}}\int_{\xi \in \Xi} E_{\xi }[ \kappa _{\xi }(\tilde{\mathbf{{Z}}})\lg (\tilde{U}(\tilde{\mathbf{{Z}}})))] dW(\xi )+\int_{\xi \in \Xi}P_{\xi }( Y_{\mu }\in \tilde{ U}(\mathbf{\tilde{Z}})) d\Lambda (\xi ),
\label{eq:program2}
\end{equation}
where the non-negative measure $\Lambda$ denotes the Lagrangian weights chosen to guarantee the asymptotic size constraint in \eqref{eq:program}.\footnote{\citet{mullerwang2017} proposes the inference method for extreme quantile and tail conditional expectation. Despite its similar structure, our problem here is fundamentally different from that in \citet{mullerwang2017}. First, our asymptotic distribution is built upon $n$ independent draws from $Z_{2,j}$, while that in \citet{mullerwang2017} is built upon a vector of order statistics, which are dependent and jointly EV distributed. Second, the objects of interest in the limit experiment are substantially different. Third, our Theorem \ref{thm:CI_K} formally establishes asymptotic optimality, while the analysis in \citet{mullerwang2017} does not.}
If we ignore the constraints in $\mathbb{U}$, the solution to \eqref{eq:program2} is given by the following set: for every $h \in \Sigma$,
\begin{equation}
\tilde{U}(h) ~=~ \bigg\{ y \in \mathbb{R} : \int_{\xi \in \Xi} \kappa_{\xi}(h)f_{\tilde{\mathbf{Z}}|\xi}(h) dW(\xi) \leq \int_{\xi \in \Xi} f_{(Y_{\mu},\tilde{\mathbf{Z}})|\xi}(y,h) d\Lambda(\xi) \bigg\}, \label{eq:sol_program2}
\end{equation}
where for every $(y,h) \in \mathbb{R} \times \Sigma$, $\kappa_{\xi}(h)f_{\tilde{\mathbf{Z}}|\xi}(h)$ and $f_{(Y_{\mu},\tilde{\mathbf{Z}})|\xi}(y,h)$ are given in Section \ref{sec:appendix-computation}, and the integrals in \eqref{eq:sol_program2} can be numerically calculated by Gaussian quadrature. We can verify numerically that $\tilde{U}(h) $ in \eqref{eq:sol_program2} takes the form of an interval whose length varies continuously with $h \in \Sigma$. 
Under these conditions, Lemma \ref{lem:interval1} shows that $\tilde{U}(h)$ in \eqref{eq:sol_program2} belongs to $\mathbb{U}$ and, therefore, $\tilde{U}(h)$ in \eqref{eq:sol_program2} solves \eqref{eq:program2}.

By combining \eqref{eq:CI_muK_defn}, \eqref{eq:sol_program2}, and the arguments in the previous paragraph, we propose:
\begin{equation}
U(\mathbf{P}) ~=~ ( P_{(n)}-P_{(1)}) \times \bigg\{ y \in \mathbb{R} : \int_{\xi \in \Xi} \kappa_{\xi}(\mathbf{\tilde{P}}) f_{\tilde{\mathbf{Z}}|\xi}(\mathbf{\tilde{P}}) dW(\xi) \leq \int_{\xi \in \Xi} f_{(Y_{\mu},\tilde{\mathbf{Z}})|\xi}(y,\mathbf{\tilde{P}}) d\Lambda(\xi) \bigg\}. \label{eq:CI_muK_defn2}
\end{equation}
Under the conditions in the previous paragraph, Theorem \ref{thm:CI_K} implies that \eqref{eq:CI_muK_defn2} belongs to the class of asymptotically valid CIs and minimizes the asymptotic expected length (scaled by $a_K>0$) within this class.

To calculate \eqref{eq:CI_muK_defn2}, the only remaining challenge is to find appropriate Lagrangian weights $\Lambda$ that ensure asymptotic validity in the limiting problem, as described in \eqref{eq:asy_cov}. We tackle this challenge using a numerical approach developed by \citet{elliott2015}. It is relevant to emphasize that these Lagrangian weights depend solely on the value of $n$ and, as a result, they only need to be computed once. For more information on calculating these weights, please refer to Section \ref{sec:appendix-computation}.

\begin{runningexamplecont}
    The winner's expected utility $\mu_K$ is substantial in single-letter license plate auctions. Based on $\mathbf{P}=\{20.2, 25.5, 26.0, 13.0\}$, we find in Section \ref{sec:application} that $U(\mathbf{P})=[2.19,~58.89]$, expressed in million HKD. According to our theoretical results, this CI covers $\mu_K$ with a probability of 95\% and minimizes the weighted expected length defined in \eqref{eq:asy_length}.
\end{runningexamplecont}

\subsection{Inference about the seller's expected revenue}
\label{sec:sp-revenue}

We now conduct inference on the average of the seller's expected revenue based on transaction prices. Since bidders bid their valuation and the transaction price is the second-highest bid, we get that the seller's expected revenue in auction $j=1,\dots,n$ is $E[P_{j}]=E[V_{(2),j}]$, whose average is
\begin{equation}
\pi _{K}~= 
~\frac{1}{n}\sum_{j=1}^{n}E[V_{(2),j}].  \label{eq:seller1}
\end{equation}

Given transaction prices $\mathbf{P}$, we consider a CI for $\pi _{K}$ given by
\begin{equation}
U(\mathbf{P})~\equiv ~( P_{(n)}-P_{(1)}) \times \tilde{U}(\mathbf{\tilde{P }}) + P_{(1)}, \label{eq:CI_piK_defn}
\end{equation}
where $\mathbf{\tilde{P}}\in \Sigma $ are the sorted and self-normalized transaction prices in \eqref{eq:P*}, and $\tilde{U}:\Sigma\to \mathcal{P}(\mathbb{R})$ is a CI defined on $\mathbf{\tilde{P}}$. By \eqref{eq:CI_piK_defn}, the CI $U(\mathbf{P})$ is invariant to the sorting of $\mathbf{P}$, and equivariant to their location and scale.

The remainder of this section denotes $ Z_{2,j}=H_{\xi }(E_{1,j}+E_{2,j})$, $ \{(E_{1,j},E_{2,j}):j=1,\dots ,n\}$ are i.i.d.\ standard exponential random variables and $H_{\xi }$ is as in \eqref{eq:H_function}. We also use $ \mathbf{Z}=\{Z_{2,j}:j=1,\dots ,n\}$, $Z_{(n)}=\max \{Z_{2,j}:j=1,\dots ,n\}$, $Z_{(1)}=\min \{Z_{2,j}:j=1,\dots ,n\}$, for $j=1,\ldots ,N=n-2$,
\begin{equation*}
\tilde{Z}_{j}=\bigg\{
\begin{tabular}{ll}
$\frac{Z_{( j+1) }-Z_{(1)}}{Z_{(n)}-Z_{(1)}}$ & if $ Z_{(n)}>Z_{(1)},$ \\
$0$ & if $Z_{(n)}=Z_{(1)},$
\end{tabular}
\end{equation*}
and $\mathbf{\tilde{Z}}=\{\tilde{Z}_{j}:j=1,\dots ,N\} \in \Sigma $. Finally, let $Y_{\pi }\equiv (E[Z_{2,1}] - Z_{(1)}) /(Z_{(n)}-Z_{(1)})$ and $\kappa _{\xi }(\mathbf{\tilde{Z}} )=E[Z_{(n)}-Z_{(1)}|\mathbf{\tilde{Z}}]$. As in the previous section, the distribution of these random variables is fully characterized by its tail index $\xi $. In particular, Lemma \ref{lem:sp-revenue} shows that $E[Z_{2,1}] = (\Gamma(2-\xi)-1)/\xi$ if $\xi \neq 0$, and $-1+\bar{\gamma}$ if $\xi =0$, where $\Gamma$ is the standard Gamma function and $\bar{\gamma}\approx 0.577$ is the Euler's constant. For the remainder of this section, we will use $P_{\xi}$ and $E_{\xi}$ to refer to the probability and expectation associated with this distribution.
The next result provides the asymptotic properties of the CI in \eqref{eq:CI_piK_defn} as $K\to \infty$. 
\begin{theorem}\label{thm:CI_pi_K} 
Assume \eqref{eq:DomainOfAttraction} holds, and that for some $\varepsilon >0$ with $(1+\varepsilon)\xi<1$, $E[|V_{i,j}|^{1+\varepsilon}]<\infty $ for all $i=1\dots,K_j$ in auction $j=1,\dots,N$. Finally, assume that the CI for $\pi _{K}$, $U(\mathbf{P})$, is as in \eqref{eq:CI_piK_defn} with $\tilde{U}:\Sigma\to \mathcal{P}(\mathbb{R}) $ that satisfies the following conditions:
\begin{enumerate}[(a)]
\item $P_{\xi}(\{Y_{\pi},\mathbf{\tilde{Z}}\}\in \partial \{(y,h)\in \mathbb{R} \times \Sigma:y\in \tilde{U}(h)\})=0$, where $\partial A$ denotes the boundary of $A$.
\item $\lg (\tilde{U}(h))<\infty $ for any $h\in \Sigma$, where $\lg (A)$ denotes the length of $A$.
\item For any sequence $\{h_{\ell}\in \Sigma\}_{\ell\in \mathbb{N}}$ with $ h_{\ell}\to h\in \Sigma$, $\lg (\tilde{U} (h_{\ell}))\to \lg (\tilde{U}(h))$.
\end{enumerate}
Then, as $K\to \infty $,
\begin{enumerate}
\item $P(\pi _{K}\in U(\mathbf{P}))~\to ~P_{\xi }(Y_{\pi }\in \tilde{ U}(\mathbf{\tilde{Z}}))$,
\item $E[\lg (U(\mathbf{P}))]/a_{K}~\to ~E_{\xi }[\kappa _{\xi }( \mathbf{\tilde{Z}})\lg (\tilde{U}(\mathbf{\tilde{Z}}))]$.
\end{enumerate}
\end{theorem}
\vspace{0.5em}

Following the ideas in Section \ref{sec:sp-revenue}, we choose our CI based on the asymptotic behavior in Theorem \ref{thm:CI_pi_K}. In particular, we propose to choose $\tilde{U}$ in \eqref{eq:CI_piK_defn} to minimize the asymptotic weighted length of the CI subject to asymptotic validity, given by
\begin{equation}
\underset{\tilde{U} \in \mathbb{U}}{\arg\min} \int_{\xi \in \Xi} E_{\xi }[ \kappa _{\xi }(\tilde{\mathbf{{Z}}})\lg (\tilde{U}(\tilde{\mathbf{{Z}}})))] dW(\xi )~~~\text{ s.t.}~~~
P_{\xi }( Y_{\pi }\in \tilde{ U}(\mathbf{\tilde{Z}})) ~\geq~ 1-\alpha~~~\text{ for all }\xi \in \Xi ,
\label{eq:program-rp}
\end{equation}
where $\mathbb{U}$ denotes the collection of CIs that satisfy the conditions in Theorem \ref{thm:CI_pi_K} and $W$ is a user-defined weight function. 
The solution to \eqref{eq:program-rp} is as follows: for every $h \in \Sigma$,
\begin{equation}
\tilde{U}(h) ~=~ \bigg\{ y \in \mathbb{R} : \int_{\xi \in \Xi} \kappa_{\xi}(h)f_{\tilde{\mathbf{Z}}|\xi}(h) dW(\xi) \leq \int_{\xi \in \Xi} f_{(Y_{\pi},\tilde{\mathbf{Z}})|\xi}(y,h) d\Lambda(\xi) \bigg\}, \label{eq:Lambda-rp}
\end{equation}
where, for every $(y,h) \in \mathbb{R} \times \Sigma$, $\kappa_{\xi}(h)f_{\tilde{\mathbf{Z}}|\xi}(h)$ and $f_{(Y_{\pi},\tilde{\mathbf{Z}})|\xi}(y,h)$ are given in Section \ref{sec:appendix-computation}.

By combining \eqref{eq:CI_muK_defn}, \eqref{eq:sol_program2}, and the arguments in the previous paragraph, we propose:
\begin{align}
U(\mathbf{P})~=~& P_{(1)}+( P_{(n)}-P_{(1)}) \notag \\ 
\times& \bigg\{ y \in \mathbb{R} : \int_{\xi \in \Xi} \kappa_{\xi}(\mathbf{\tilde{P}}) f_{\tilde{\mathbf{Z}}|\xi}(\mathbf{\tilde{P}}) dW(\xi) \leq \int_{\xi \in \Xi} f_{(Y_{\pi},\tilde{\mathbf{Z}})|\xi}(y,\mathbf{\tilde{P}}) d\Lambda(\xi) \bigg\} , \label{eq:CI_piK_defn2}
\end{align}
Our derivations establish that \eqref{eq:CI_piK_defn2} belongs to the class of asymptotically valid CIs and minimizes the asymptotic expected length (scaled by $a_K>0$) within this class.

\begin{runningexamplecont}
    The government's expected revenue $\pi_K$ is also substantial in single-letter license plate auctions. Using $\mathbf{P}=\{20.2, 25.5, 26.0, 13.0\}$, we find in Section \ref{sec:application} that $U(\mathbf{P})=[8.70,~37.65]$, expressed in million HKD. As the number of participants diverges, this CI covers $\pi_K$ with a minimum probability of 95\% and minimizes the weighted expected length.
\end{runningexamplecont}

\subsection{Inference about the tail index}\label{sec:sp-index}

A key parameter in our asymptotic framework is the tail index $\xi \in \Xi$.
The goal of this section is to conduct inference on this parameter. That is, we are interested in the following hypothesis test:
\begin{equation}
H_{0}:\xi \in \xi _{0}~~~\text{v.s.}~~~H_{1}:\xi \in \Xi _{1},
\label{eq:HT0}
\end{equation}
where $\xi _{0}$ is a fixed parameter value in $\Xi $ and $\Xi _{1}=\Xi \backslash \{ \xi _{0}\}$.

We divide this section into three subsections. Section \ref{sec:HypSimple} considers the situation where the alternative hypothesis in \eqref{eq:HT0} is simple, while Section \ref{sec:HypComp} explores the case where this alternative hypothesis is composite. Finally, Section \ref{sec:HypApplication} applies the methods in previous sections to test the standard regularity conditions in auction models. 

\subsubsection{Simple alternative hypothesis}\label{sec:HypSimple}

This section considers the following inference problem:
\begin{equation}
H_{0}:\xi =\xi _{0}~~~\text{v.s.}~~~H_{1}:\xi =\xi _{1},  \label{eq:HT1}
\end{equation}
where $\xi _{0}$ and $\xi _{1}$ are distinct parameter values. 

By the Neyman-Pearson Lemma, a natural starting point is the likelihood ratio test for the sorted and self-normalized version of our data. While the likelihood ratio test is unknown in finite samples, its limiting distribution is provided in Lemma \ref{lem:densitySecond}. This idea yields the following testing procedure:
\begin{equation}
{\varphi}^{\ast }( \mathbf{{Z}}) 
~\equiv ~1\big[ {f_{\mathbf{\tilde{Z}}|\xi _{1}}( \mathbf{\tilde{Z}}) }/{f_{\mathbf{\tilde{Z}}|\xi _{0}}( \mathbf{\tilde{Z}}) }>q( \xi_{0},\xi _{1},\alpha ) \big] ,  \label{eq:TailTest1}
\end{equation}
where the critical value is given by
\begin{equation*}
q( \xi _{0},\xi _{1},\alpha ) ~\equiv ~( 1-\alpha )\text{-quantile of }{f_{\mathbf{\tilde{Z}}|\xi _{1}}( \mathbf{\tilde{Z}}_{0}) }/{f_{\mathbf{\tilde{Z}}|\xi _{0}}( \mathbf{\tilde{Z}}_{0}) }
\end{equation*}
and $\mathbf{\tilde{Z}}_{0}$ is distributed according to $f_{\mathbf{\tilde{Z}}|\xi _{0}}$. The Neyman-Pearson Lemma implies that \eqref{eq:TailTest1} is the most powerful level-$\alpha $ test in the limiting problem. 

Following the guidance of the asymptotic analysis, our candidate for optimal test follows from replacing in \eqref{eq:TailTest1} the limiting random variable $\mathbf{\tilde{Z}}$ with its data analog $\mathbf{\tilde{P}}$, i.e., 
\begin{equation}
\varphi _{K}^{\ast }( \mathbf{P}) ~\equiv
~1\big[ {f_{\mathbf{\tilde{Z}}|\xi _{1}}( \mathbf{\tilde{P}}) }/{f_{\mathbf{\tilde{Z}}|\xi _{0}}( \mathbf{\tilde{P}}) }>q( \xi _{0},\xi_{1},\alpha ) \big] . \label{eq:simpleTest}
\end{equation}
By Lemma \ref{lem:densitySecond} and standard convergence arguments, \eqref{eq:simpleTest} is asymptotically valid, i.e.,
\begin{equation}
\lim_{K\to \infty }E_{\xi _{0}}[ \varphi _{K}^{\ast }( \mathbf{P}) ] ~\leq ~\alpha ,  \label{eq:sizeAlpha}
\end{equation}
where $E_{\xi } $ denotes the expectation with respect to distribution with tail index $\xi $. In fact, \eqref{eq:simpleTest} is asymptotically level $\alpha $, i.e., \eqref{eq:sizeAlpha} holds with equality. More interestingly, we can leverage \citet[Theorem 1]{muller:2011} to establish that \eqref{eq:simpleTest} is efficient in the sense of being the asymptotically most powerful test in the class of asymptotically valid and equivariant tests. Formally, for any asymptotically valid test $\varphi _{K}( \mathbf{P}) $ (i.e., \eqref{eq:sizeAlpha} holds with $\varphi _{K}^{\ast }( \mathbf{P}) $ replaced by $\varphi _{K}( \mathbf{P}) $),
$
\underset{K\to \infty }{\lim \sup }~E_{\xi _{1}}[ \varphi _{K}( \mathbf{P}) ] ~\leq ~\underset{K\to \infty }{\lim \sup }~E_{\xi _{1}}[ \varphi _{K}^{\ast }( \mathbf{{{P}}}) ] .
$
We record these conclusions in the following result.

\begin{theorem}\label{thm:simpleTest}
Assume \eqref{eq:DomainOfAttraction} holds. In the hypothesis testing problem in \eqref{eq:HT1}, the test defined by \eqref{eq:simpleTest} satisfies the following properties:
\begin{enumerate}
\item It is asymptotically valid and level $\alpha$, i.e., \eqref{eq:sizeAlpha} holds with equality.
\item It is asymptotically efficient.
\end{enumerate}
\end{theorem}

\subsubsection{Composite alternative hypothesis}\label{sec:HypComp}

We now turn our attention to the case in which the alternative hypothesis in \eqref{eq:HT0} is composite. Following the ideas used in Section \ref{sec:HypSimple}, we consider the feasible version of the efficient test in the limiting problem. Unfortunately, the limiting problem does not lend itself to the usual tools to develop uniformly most powerful tests.\footnote{In particular, the likelihood ratio statistic is not monotonic, rendering the results in \citet[Section 3.4]{Lehman/Romano:2005} inapplicable.} For this reason, we consider tests that maximize the weighted average power criterion (WAP), following the approach of \cite{wald:1943,andrews/ploberger:1994}. To this end, let $W$ denote a user-defined weight function on $\Xi _{1}$, which the researcher chooses to reflect the importance attached to the various alternative hypotheses. The weighting function $W$ effectively transforms the composite alternative into a simple one, allowing us to focus on the weighted average power:
\begin{equation*}
\mathrm{WAP}_{K}( \varphi _{K}) ~\equiv ~\int_{\Xi _{1}}E_{\xi } [ \varphi _{K}( \mathbf{\tilde{P}}) ] dW( \xi ) .
\end{equation*}

As in Section \ref{sec:HypSimple}, we use the asymptotic behavior to guide the construction of our hypothesis test. The likelihood ratio test in the limiting problem applied to its data $\mathbf{\tilde{P}}$ is given by
\begin{equation}
\varphi _{K}^{\ast }( \mathbf{P}) 
~\equiv ~1\bigg[ {\int_{\Xi _{1}}f_{\mathbf{\tilde{Z}}|\xi }( \mathbf{\tilde{P}}) dW( \xi ) }/{f_{\mathbf{\tilde{Z}}|\xi _{0}}( \mathbf{\tilde{P}}) }>q( \xi _{0},W,\alpha ) \bigg] ,  \label{eq:compTest}
\end{equation}
where the critical value is
\begin{equation*}
q( \xi _{0},W,\alpha ) ~\equiv ~( 1-\alpha ) \text{-quantile of }{\int_{\Xi _{1}}f_{\mathbf{\tilde{Z}}|\xi _{0}}(  \mathbf{\tilde{Z}_0}) dW( \xi ) }/{f_{\mathbf{\tilde{Z}}|\xi _{0}}( \mathbf{\tilde{Z}_0}) },
\end{equation*}
and $\mathbf{\tilde{Z}}_{0}$ is distributed according to $f_{\mathbf{\tilde{Z}}|\xi _{0}}$. By standard asymptotic arguments, we can establish that \eqref{eq:compTest} is asymptotically valid and level $\alpha $, i.e.,
\begin{equation}
\lim_{K\to \infty }E_{\xi_0 }[ \varphi _{K}^{\ast }( \mathbf{P}) ] ~= ~\alpha .  \label{eq:sizeAlpha2}
\end{equation}
Moreover, \eqref{eq:compTest} is efficient in the sense of maximizing the asymptotic weighted average power criterion in the class of asymptotically valid and equivariant tests. Formally, for any test $\varphi _{K}( \mathbf{P}) $ that is asymptotically valid (i.e., \eqref{eq:sizeAlpha2} holds with $\varphi _{K}^{\ast }$ replaced by $\varphi _{K}$), then
$
\underset{K\to \infty }{\lim \sup }~\mathrm{WAP}_{K}( \varphi_{K}) ~\leq ~\underset{K\to \infty }{\lim \sup }~\mathrm{WAP}_{K}( \varphi _{K}^{\ast }) .
$
The next result records these conclusions.

\begin{theorem}\label{thm:compTest}
Assume \eqref{eq:DomainOfAttraction} holds. In the hypothesis testing problem in \eqref{eq:HT0}, the test defined by \eqref{eq:compTest} satisfies the following properties:
\begin{enumerate}
\item It is asymptotically valid and level $\alpha $, i.e., \eqref{eq:sizeAlpha2} holds.
\item It is asymptotically efficient.
\end{enumerate}
\end{theorem}

\subsubsection{Testing the regularity conditions in the auction literature}\label{sec:HypApplication}

The auction literature routinely assumes the regularity condition that $f_{V}$ is continuous, and that there is a finite maximum valuation $v_H<\infty $ with $f_{V}(v_H)>0$. For examples of this, see \cite{MaskinRiley1984}, \cite{guerre2000}, and \cite{GuerreLuo2022}. Within our asymptotic framework, the next result shows that these regularity conditions imply \eqref{eq:DomainOfAttraction} holds with $\xi =-1$.

\begin{lemma}\label{lem:psiOne} 
Assume that $f_{V}( v) \to f_{V}( v_H) >0 $ as $v\uparrow v_H<\infty $. Then, \eqref{eq:DomainOfAttraction} holds with $\xi =-1$.
\end{lemma}

In light of Lemma \ref{lem:psiOne}, we can test the regularity conditions used in the auction literature via the following hypothesis test:
\begin{equation}
H_{0}:\xi =-1~~\text{ v.s. }~~H_{1}:\xi \in \Xi_1,\label{eq:HT3}
\end{equation}
where $\Xi_1 = \Xi \backslash \{-1\}$. In the remainder of this section, we will argue that $\Xi =[-1,0.5]$ is a suitable choice for the parameter space for $\xi$. Since \eqref{eq:HT3} is a special case of \eqref{eq:HT0} with $\Xi _{1} =(-1,0.5]$, we can implement this test using the procedure in \eqref{eq:compTest}. For concreteness, we consider uniform weight $W( \xi ) =1[ \xi \in (-1,0.5]] $. By Theorem \ref{thm:compTest}, our proposed test is asymptotically valid, level-$\alpha$, and efficient.

We now justify choosing $\Xi =[-1,0.5]$ as the parameter space for the test. We rely on the so-called von Mises' condition to interpret the different values of $\xi$. \citet[Theorem 1.1.8]{dehaan2006book} states this condition and shows that it is a sufficient condition for \eqref{eq:DomainOfAttraction}. Under the von Mises' condition and that $f_{V}'$ is bounded, we have three possible cases:
\begin{itemize}
\item Case 1: $f_{V}( v) \to f_{V}( v_H)>0$ as $v \uparrow  v_H<\infty $ implies that $\xi =-1$.
\item Case 2: $f_{V}(v)\to f_{V}( v_H) =0$ as $v \uparrow  v_H<\infty $ implies that $\xi \in (-1,0]$.
\item Case 3: $f_{V}(v)\to f_{V}( v_H) =0$ as $v \uparrow  v_H=\infty $ implies that $\xi >0$.
\end{itemize}
A few remarks are in order. First, as expected, Case 1 aligns with the findings of Lemma \ref{lem:psiOne}. Second, we note that \citet[Theorem 2.1.2]{dehaan2006book} shows that $v_H<\infty $ if and only if $\xi \leq 0$. Cases 2 and 3 then follow from derivations in \citet[page 18]{dehaan2006book} and \citet[Theorem 2.1.2]{falk/husler/reiss:1994}. Since these three cases are exhaustive, we deduce that $\xi \geq -1$. Third, if we also impose that $V$ has second moments, \citet[page 176]{dehaan2006book} implies that $\xi \leq 1/2$. Combining these restrictions, we conclude that $\Xi= [-1,0.5]$ is a suitable parameter space for $\xi$. Finally, in first-price auctions, our test also serves as a specification test for the existence of a unique Bayesian Nash equilibrium. In particular, the first-price auction may not have a unique Bayesian Nash Equilibrium if the private density is not bounded away from zero on its compact support. We provide such a test in Appendix \ref{sec:appendix-fp}.

Figure \ref{fig:power} presents the asymptotic rejection probabilities of \eqref{eq:compTest} for the hypothesis testing problem in \eqref{eq:HT3}. The proposed test is asymptotically valid under $H_{0}:\xi =-1$ and has nontrivial power properties under $H_{1}$, with asymptotic rejection rates that increase when either $ \xi $ or $n$ increases. It is worth noting that Figure \ref{fig:power} only shows the asymptotic properties (as $K\to\infty$) of our test. We explore the finite-sample properties of our methodology via simulations in Section \ref{sec:simulation}.
\begin{figure}[htb]
\centering
\includegraphics[width=0.6\textwidth]{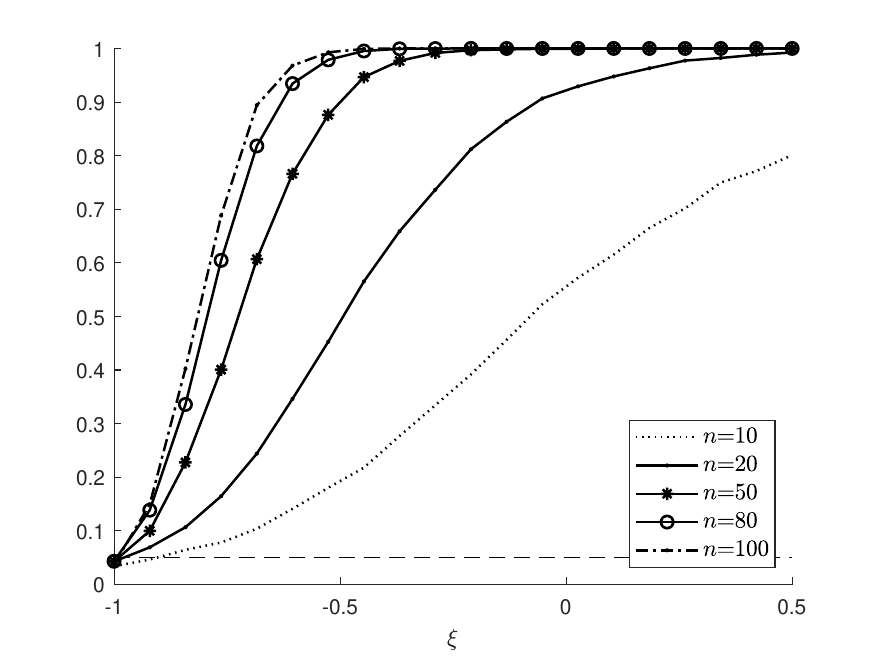}
\caption{Asymptotic rejection probabilities of the hypothesis testing procedure in \eqref{eq:compTest} with $\alpha = 5\%$ for the hypothesis testing problem in \eqref{eq:HT3}.}
\label{fig:power}
\end{figure}

\begin{runningexamplecont}
Using $\mathbf{P} = \{20.2, 25.5, 26.0, 13.0\}$, we implement a hypothesis test of \eqref{eq:HT3} in Section \ref{sec:application} and obtain a $p$-value of 0.64. Thus, the standard regularity conditions in the auction literature are not rejected in these data. Nonetheless, the methods in this literature require a diverging number of auctions, which is inapplicable when there are only four transaction prices.
\end{runningexamplecont}

\section{Extensions}\label{sec:extension}

Our analysis thus far considers second-price auctions in an IPV setup and does not allow for a reserve price. 
This section first extends our analysis to first-price auctions. It then briefly discusses how the framework can be extended beyond the IPV setting and to accommodate a binding reserve price.

\subsection{First-price auctions}\label{sec:fp}

We now consider first-price sealed-bid actions, in which the highest bidder gets the object and pays the highest bid. This type of auction is strategically equivalent to an open descending price (or Dutch) auction (see \citet[page 4]{krishna2009book}). By standard arguments (e.g., \citet[Proposition 2.2]{krishna2009book}), the symmetric equilibrium strategy for a bidder with valuation $v$ in an auction with $K_{j}$ participants is to bid $\beta_{j}(v) \equiv E[\check{V}_{(1),j}|\check{V}_{(1),j}<v]$, where $\check{V}_{(1),j}$ denotes the highest bid among the remaining $(K_{j}-1)$ participants. Note that
\begin{align}
\beta_{j}(v) ~\overset{(1)}{=}~ \frac{K_j-1}{F_V(v)^{K_j-1}}\int_{v_{L}}^{v} {u F_V(u)^{K_j-2}f_V(u)} du~\overset{(2)}{=}~
v - \frac{\int_{v_{L}}^{v} {F_V(u)^{K_j-1}} du}{F_V(v)^{K_j-1}},
\label{eq:bid_first}
\end{align}
where (1) holds by computing $E[\check{V}_{(1),j}|\check{V}_{(1),j}<v]$ using the fact that the remaining $(K_{j}-1)$ participants have a common valuation distribution with PDF $f_V$, and (2) by integration by parts. 
Since $\beta (v)$ is increasing, the auction is won by the highest-valuation bidder, with a transaction price equal to
\begin{equation}
P_{j}~=~V_{(1),j}-\frac{\int_{v_{L}}^{V_{(1),j}}{F_{V}(u)^{K_{j}-1}}du}{F_{V}(V_{(1),j})^{K_{j}-1}}.  \label{eq:price_first}
\end{equation}

Lemma \ref{lem:joint_first} uses \eqref{eq:price_first} to deduce that 
\begin{equation}
\Big\{ \frac{P_{j}-b_{K}}{a_{K}}:j=1,\dots ,n\Big\} ~\overset{d}{\to }~\{ X_{j}:j=1,\dots ,n\} , 
\label{eq:aux_first}
\end{equation}
where, for each $j=1,\dots ,n$, 
\begin{equation}
X_{j}~\equiv ~H_{\xi }( E_{1,j}) -\frac{\int_{-\infty }^{H_{\xi }( E_{1,j}) }G_{\xi }( h) dh}{G_{\xi }( H_{\xi }( E_{1,j}) ) },  \label{eq:aux_RV}
\end{equation}
with $G_{\xi }$ is as in \eqref{eq:G}, and $\{(a_{K},b_{K})\in \mathbb{R}_{++}\times \mathbb{R}:K\in \mathbb{N}\}$, $\{E_{1,j}:~j=1,\dots ,n\}$, and $H_{\xi }$ are as specified in Lemma \ref{lem:joint}.

As in Section \ref{sec:sp}, the statement in \eqref{eq:aux_first} cannot be directly used for inference as it requires the unknown normalizing constants $\{(a_{K},b_{K})\in \mathbb{R}_{++}\times \mathbb{R}: K \in \mathbb{N}\}$. Nevertheless, we can reiterate the idea of considering sorted and self-normalized prices in \eqref{eq:P*}; for $j=1,\dots,N = n-2\geq 1$, let
\begin{equation}
\tilde{P}_{j}~\equiv ~\bigg\{
\begin{array}{cc}
     \frac{P_{(j+1)}-P_{(1)}}{P_{(n)}-P_{(1)}}&\text{ if }P_{(n)} > P_{(1)},  \\
    0&\text{ if }P_{(n)} = P_{(1)} , 
\end{array}
    \label{eq:P*fp}
\end{equation}
and let $\mathbf{\tilde{P}}=\{ \tilde{P}_{j}:j=1,\dots,N\} \in \Sigma$. We now characterize the asymptotic distribution of $\mathbf{\tilde{P}}$ as $K \to \infty$.

\begin{lemma}\label{lem:densityFirst}
Assume \eqref{eq:DomainOfAttraction} holds. For any $ N\in \mathbb{N}$, and as $K\to \infty $, 
\begin{equation}
\mathbf{\tilde{P}}=\{\tilde{P}_{j}:j=1,\dots ,N\}~\overset{d}{\to }~\mathbf{\tilde{X}}=\{\tilde{X}_{j}:j=1,\dots ,N\},  \label{eq:P*joint_fp}
\end{equation}
where $\mathbf{\tilde{X}}=\{\tilde{X}_{j}:j=1,\dots ,N\}\in \Sigma $ is obtained as follows: for $j=1,\dots ,N = n-2\geq 1$, 
\begin{equation}
\tilde{X}_{j}~\equiv ~\bigg\{ 
\begin{array}{cc}
\frac{X_{(j+1)}-X_{(1)}}{X_{(n)}-X_{(1)}} & \text{ if }X_{(n)}>X_{(1)}, \\ 
0 & \text{ if }X_{(n)}=X_{(1)},
\end{array}
 \label{eq:Xtilde}
\end{equation}
with $\{ X_{j}:j=1,\dots ,n\}$ is i.i.d.\ according to \eqref{eq:aux_RV}.
\end{lemma}

Given that the random variable in \eqref{eq:aux_RV} is informative about the tail EV index $\xi $, Lemma \ref{lem:densityFirst} implicitly reveals that the asymptotic distribution of $\mathbf{\tilde{P}}$ can be used to conduct inference on functions of $\xi $. From this point onward, the remainder of this section is analogous to that of Section \ref{sec:sp}. The main difference is that the explicit PDF of the asymptotic distribution of $\mathbf{\tilde{P}}$ in Lemma \ref{lem:densitySecond} is replaced by the implicit distribution in Lemma \ref{lem:densityFirst}. 

Analogous to second-price auctions, we can construct confidence intervals for the winner's expected utility and seller's expected revenue and test for $\xi$. 
See Appendix \ref{sec:appendix-fp} for details.

\subsection{Beyond the IPV setup}

According to Section \ref{sec:setup}, the $K_{j}$ bidders in auction $j=1,\dots,n$ have IPV with a common CDF $F_{V}$. The auctions may differ in the number of potential bidders, but these are assumed to coincide asymptotically in the sense that  $K_{j}/K\to 1$ with $K\equiv \min \{ K_{1},\dots,K_{n}\} $. 

The IPV assumption may be restrictive in certain empirical settings. For example, consider the case where the $K_{j}$ bids in auction $j$ depend on an auction-specific feature $A_{j}$. In this context, it is plausible that auctions are independent and, conditional on $A_{j}$, the $K_{j}$ bids in auction $j$ are independent and distributed according to a common CDF $ F_{V|A_{j}}$. This is known as the conditional IPV model, and it is an extension of the standard IPV setup that allows for heterogeneity across auctions and (unconditional) dependence among private values within each auction. See \citet{li2000cipv} for a discussion of the conditional IPV model. As we now explain, one can adapt our methodology to the conditional IPV setup.

First, consider the case in which the auction-specific features $\{A_j: j=1, \dots, n\}$ are observed. In this case, one can always apply our analysis to collections of auctions that have the same (or similar) feature value, which we refer to as clusters. One can apply the analysis to each cluster with more than three auctions without any modification. This extension illustrates one of the advantages of our methodology, as it does not require the number of auctions to diverge. 
Because our approach exploits the extreme-value behavior of transaction prices as the number of bidders grows, it does not rely on the joint structure of multiple order statistics and is not intended to identify separable or nonseparable forms of unobserved heterogeneity.

Second, consider the case in which $\{ A_{j}:j=1,\ldots ,n\} $ are unobserved. 
The existing literature proposes to (i) exploit the joint distribution of three bids from the same auction and (ii) require the number of auctions to diverge \citep[e.g.,][]{li/vuong:1998,hu2013identification,luo2024order}. 
Rather than requiring a large number of auctions, we can still perform inference using our asymptotic framework, provided we observe at least three bids for each action. The basic insight is to treat each auction as its own cluster (in the sense of the previous paragraph), which naturally satisfies the IPV model in Section \ref{sec:setup}. For illustration, we consider the second-price auctions, in which bidders declare their true valuations. Suppose that for auction $j=1,\ldots,n$ we observe $m\geq 3$ bids $ \{P_{(k_{1}),j},P_{(k_{2}),j},\ldots,P_{(k_{m}),j}\}$ for known indices $ (k_{1},k_{2},\ldots,k_{m})\in \mathbb{N}$, which are in increasing order without loss of generality (i.e., $k_{1}<k_{2}<\ldots <k_{m}$). These bids need not be consecutive (i.e., $k_{u+1}$ need not equal $k_{u}+1$ for $u=1,\ldots,m-1$) or include the maximum in the auction (i.e., $k_{m}$ need not equal $m$). Given that the bids belong to the same auction, the unobserved auction-specific feature $A_{j}$ is conditioned upon in these data. By repeating the arguments in Lemma \ref{lem:joint}, it follows that, as $ K\rightarrow \infty ,$
\begin{equation}
\Big( \frac{P_{(k_{1}),j}-b_{K,j} }{a_{K,j} },\cdots ,\frac{P_{(k_{m}),j}-b_{K,j} }{a_{K,j}}\Big) ~\overset{d}{\to }~\Bigg( H_{\xi_j }\Bigg( \sum_{s=1}^{k_{1}}E_{s,j}\Bigg) ,\cdots ,H_{\xi_j }\Bigg( \sum_{s=1}^{k_{m}}E_{s,j}\Bigg) \Bigg) ,
\label{eq:EVT_extension}
\end{equation}
where $\{(a_{K,j} ,b_{K,j} )\in \mathbb{R} _{++}\times \mathbb{R}:K\in \mathbb{N}\}$ are auction-specific normalizing constants, $H_{\xi_j }\left( \cdot \right) $ is defined in \eqref{eq:H_function} with auction-specific EV index $\xi_j \equiv \xi ( A_{j}) $, and $\{E_{s,j}:s=1,\dots ,k_{m}\}$ are i.i.d.\ standard exponential random variables. Given these bids, we can construct the auction-specific self-normalized statistics: for $u=1,\ldots ,M\equiv m-2\geq 1,$
\begin{equation*}
\tilde{P}_{u,j}~\equiv ~\Bigg\{
\begin{array}{cc}
\frac{P_{(k_{u+1}),j}-P_{(k_{1}),j}}{P_{(k_{m}),j}-P_{(k_{1}),j}} & \text{ if }P_{(k_{m}),j}>P_{(k_{1}),j} ,\\
0 & \text{ if }P_{(k_{m}),j}=P_{(k_{1}),j},
\end{array}
\end{equation*}
and let $\mathbf{\tilde{P}}_{j}=\{\tilde{P}_{u,j}:u=1,\dots ,M\}\in \Sigma \equiv \{h\in [0,1]^{M}:0\leq h_{1}\leq \dots \leq h_{M}\leq 1\}$. We can then derive the asymptotic distribution of $\mathbf{\tilde{P}}_{j}$ as $ K\rightarrow \infty $ from \eqref{eq:EVT_extension}. If we observe multiple bids from several independent auctions, we can combine the self-normalized statistics $ \mathbf{\tilde{P}}_{j}$ for $j=1,\dots, J$ for further analysis. Since auctions are independent, the limiting distribution of the combined self-normalized statistics is the product of their limiting distributions. We can then test hypotheses about these auctions similar to those in previous section. In addition, we could test the homogeneity of the $J$ auctions (i.e., $\xi ( A_{j}) =\xi $ for all $j=1,\dots, J$) by using a generalized likelihood ratio test.

Third, our approach relies on extreme value theory to characterize the behavior of transaction prices as the number of bidders grows. While much of the classical extreme value literature is developed under i.i.d.\ assumptions, the validity of extreme value limits does not hinge on full homogeneity across bidders. In particular, our framework naturally accommodates environments in which bidders’ valuation distributions share a common positive tail index but differ through bidder-specific location shifts that are uniformly bounded. Such shifts do not affect tail behavior and therefore leave the extreme-value approximation unchanged. As a result, the distribution of the maximum valuation and hence of transaction prices in second-price auctions continues to be governed by the same tail index.
This form of asymmetry is economically natural in auction settings with heterogeneous bidder types or observable characteristics that shift valuations by bounded amounts, while preserving common tail behavior. By contrast, our method is not designed for settings in which bidders differ in tail heaviness, or in which a small subset of bidders has systematically heavier tails that dominate the maximum. In those cases, the effective tail behavior would be determined by the most extreme bidder type, and separate modeling of bidder asymmetry would be required. We emphasize that our results exploit the robustness of extreme-value approximation to bounded location heterogeneity, rather than modeling general bidder-specific asymmetries.

Fourth, our analysis can be extended to account for uncertainty in the number of bidders $K$. In second-price auctions with IPVs and risk-neutral bidders, truthful bidding is a weakly dominant strategy and does not depend on bidders’ knowledge of the number of competitors. Consequently, uncertainty about the number of bidders does not affect equilibrium bidding behavior or the distribution of transaction prices in this format. Transaction prices are equal to the second-highest valuation, and their tail behavior is governed directly by the extreme order statistics of valuations.
In first-price auctions, by contrast, equilibrium bidding strategies depend on the number of bidders. When bidders are uncertain about the number of competitors and share a common prior, \citet[][Chapter 2]{krishna2009book} shows that the symmetric equilibrium bidding function can be expressed as a convex combination of the equilibrium bidding functions corresponding to each possible number of bidders, with weights given by bidders’ beliefs. Specifically, let $p_K$ denote the probability that each bidder assigns to the event that he is facing $K$ other bidders. The equilibrium bidding strategy for the bidder with valuation $v$ in the $j$-th auction becomes
\begin{equation*}
    \beta_j(v) ~=~ \sum_{k=0}^{K-1}\frac{p_kF_V(v)^k}{\sum_{s=0}^{K-1}p_sF_V(v)^s}\beta^{(k)}_j(v),
\end{equation*}
where $\beta^{(k)}_j(v)$ is the equilibrium bidding function as in \eqref{eq:bid_first} when the bidder knows there are $k$ additional competitors. 
In this scenario, our proposed method still applies if the weight is asymptotically degenerate on the largest component $p_{K-1}$, that is, the bidder expects that there will be $K-1$ competitors almost surely. Extending our results to first-price auctions beyond this case would require additional restrictions on the probabilities $p_k$, which we leave for future research.

Finally, under suitable conditions, our analysis can be further generalized to allow for dependence among valuations. In particular, the EV convergence in Lemma \ref{lem:joint} can be extended to settings in which the valuations $V_{i,j}$ for $i=1,\dots,K_j$ are dependent and satisfy the so-called $D$-condition (e.g., \cite{o1974limit}, \cite{leadbetter1982extremes}). This condition restricts the dependence structure, ensuring that extreme values do not occur in clusters. 
This aspect is left for future research.



\subsection{Reserve price}

Our analysis can be adapted to allow for the presence of a reserve price $r$ set by the seller. For concreteness, we assume that  $r\in (v_{L},v_{H}) $. By setting a reserve price, the seller reserves the right not to sell the object if the price determined in the auction is below this price. \citet[Section 2.5]{krishna2009book} analyses the effect of the reserve price on the equilibrium bidding behavior. In second-price auctions, bidders still have a weakly dominant strategy to bid their valuation, i.e., \eqref{eq:bidSecond} holds. In first-price auctions, the symmetric equilibrium strategy for a bidder with valuation $v$ in an auction with $K_{j}$ participants and reserve price $r$ is to bid $\beta _{j}(v)\equiv E[\max \{\check{V}_{(1),j},r\}|\check{V}_{(1),j}<v]$, where $\check{V}_{(1),j}$ denotes the highest bid among the remaining $(K_{j}-1)$ participants.
By repeating arguments used in Section \ref{sec:fp}, we can show that the resulting formula for $\beta _{j}(v)$ is as in \eqref{eq:bid_first} but with $v_{L}$ replaced by $r$. 

These results imply that reserve prices do not affect the asymptotic distribution of the transaction prices in our asymptotic framework. That is, reserve prices have no effect on Lemmas \ref{lem:densitySecond} and \ref{lem:densityFirst}, or any of the subsequent asymptotic results. Since $v_{L}$ does not appear on our asymptotic distributions, these remain unaltered when $v_{L}$ is replaced by $r$. Intuitively, the asymptotic behavior of transaction prices with $K\rightarrow \infty $ is naturally focused on the upper tail of the valuation distribution (i.e., valuations in the neighborhood of $v_{H}$), which remains unaffected by the introduction of a reservation price $r<v_{H}$.

The fact that reserve prices do not affect the asymptotic analysis is a by-product of our limiting framework in which $K \to \infty$ and $r<v_H$. If one were interested in making reserve prices a salient feature in the asymptotic analysis, one would need to consider an asymptotic framework in which the reserve prices approach $v_H$ as $K\to \infty$. (As with any other asymptotic analysis, this is not meant to be taken literally, but rather as an approximation to a finite-sample situation in which $K$ is large and $r$ is relatively close to $v_H$). 
We present a formal analysis of this aspect in Appendix \ref{sec:appendi-reserve}.

\section{Monte Carlo simulations \label{sec:simulation}}

This section investigates the finite-sample performance of our proposed inference methods via Monte Carlo simulations. We also compare their performance with other inference methods that rely on the standard asymptotic framework with a divergent number of auctions. For brevity, we concentrate on simulations for second-price auctions. The corresponding results for first-price auctions are provided in the supplementary material.

We first consider the problem of inference about the winner's expected utility from transaction prices alone. Each simulated dataset consists of $n$ independent second-price auctions, where auction $j=1,\dots,n$ has $K$ bidders with independent valuations drawn from $F_V$. As explained in Section \ref{sec:sp-supplus}, our parameter of interest is the average winner's expected utility across auctions, given by $\mu _{K}= \frac{1}{n}\sum_{j=1}^{n}E[V_{(1),j}-V_{(2),j}]$. We conduct simulations with four distributions for $F_V$: (a) the standard uniform distribution over $[0,3]$, (b) the absolute value of standard Normal, (c) the absolute value of Student-t(20), and (d) the Pareto distribution with exponent 0.25. These distributions satisfy condition \eqref{eq:DomainOfAttraction} with $\xi = -1$, $0$, $0.05$, and $0.25$, respectively.

In our simulations, we set the number of auctions to $n=10$ or $n=100$. For each one of these $n$ auctions, we set the number of bidders to either $K=10$, $K=100$, or to a random number to a uniformly distributed integer in $\{90,91,\dots,110\}$. This generates six possible designs for $(n, K)$, and includes combinations that are better approximated by our asymptotic framework with a growing number of bidders (e.g., $n=10$ and $K=100$ or $K \in\{90,91,\dots,110\}$), but also others that are better aligned to the traditional asymptotic framework with an increasing number of auctions (e.g., $n=100$ and $K=10$). Our results are based on $S=500$ independent datasets. Finally, we set the significance level to $\alpha = 5\%$ throughout this section.

We now describe the three CIs that we consider for $\mu _{K}$:
\begin{enumerate}[(i)]
    \item This is our proposed CI in Section \ref{sec:sp-supplus}. The method is implemented by constructing $U( \mathbf{P})$ in \eqref{eq:CI_muK_defn2} with $\Xi =[ -1,0.5] $ and $W$ equal to the uniform distribution over this interval. As already explained, the validity of this CI is based on asymptotics as $K\to\infty$.
    \item A CI based on observing the two highest bidders in each auction. This approach is infeasible in our data setting (we only observe transaction prices), but we take it as a benchmark for any method that relies on the traditional asymptotics with $n\to \infty$. 
    These data allow us to compute $D_{j}\equiv B_{(1),j}-B_{(2),j}$ for each auction $j=1,\dots,n$. In turn, this enables us to test $H_0:\mu_{K} = b$ using the t-statistic $\sqrt{n}(\bar{D}_{n} - b)/s_n$ with $\bar{D}_{n}=\sum_{j=1}^n D_{j}/n$ and $s_n^2 = {\sum^{n}_{j=1} ( D_{j}-\bar{D}_{n}) ^{2}/(n-1)}$. One can then construct a CI using standard asymptotic approximations as $n\to\infty$. 
    \item A bootstrap-based CI. To implement this idea, we first need to express $\mu_K$ as a function of the distribution of transaction prices (further details are provided in Section \ref{sec:appendix-MC}). By replacing this distribution with a suitable estimator, one can consistently estimate $\mu_K$. In addition, one can repeat this process using bootstrap samples to construct a CI for $\mu _{K}$. The method is implemented with 500 bootstrap samples. The validity of this method relies on standard asymptotic arguments as $n\to\infty$. Unlike our methodology, implementing this CI requires knowledge of $K$. For references using this type of CI, see \cite{BajariHortacsu2003} and \cite{hailetamer2003}.
\end{enumerate}

Table \ref{tbl:sp-CI} presents the average coverage probability and average length of the three confidence intervals across simulation designs. We summarize the main findings as follows. First, our proposed confidence interval achieves coverage close to the nominal level across all designs, even with a small number of auctions. Coverage improves monotonically with either the number of bidders or the number of auctions, reflecting the accuracy of the extreme-value approximation as the number of bidders increases. This robustness comes at the cost of wider intervals in small samples, which is unavoidable when inference must remain valid uniformly over a broad class of tail behaviors.

Second, confidence intervals based on the two highest bids perform substantially worse in small samples, exhibiting systematic undercoverage when the number of auctions is limited. This method relies on asymptotic approximations that require a large number of auctions and implicitly treats the relevant order statistics as if their finite-sample distribution were well approximated by their limiting counterpart. When the number of auctions is small, this approximation appears to break down, leading to overly optimistic inference despite shorter interval length.

Third, the bootstrap-based confidence intervals perform particularly poorly in our simulations, especially when valuations have unbounded support. In these cases, extreme realizations dominate the sampling distribution, invalidating resampling-based approximations and resulting in severe undercoverage. This finding is consistent with the theoretical requirements underlying the influential estimator by \citet{guerre2000} and its subsequent inference literature: valid bootstrap approximations rely on the bid density being bounded away from zero throughout the support and on the valuation distribution having compact support, conditions under which kernel-based inversion is well behaved and uniform convergence over the support can be established. When support is unbounded, the tails of the bid distribution may approach zero, causing the pseudo-value transformation to become numerically unstable and undermining the regularity conditions that justify the validity of the bootstrap. While the bootstrap can yield relatively short confidence intervals, this gain in precision comes at the expense of reliability and is only observed in designs with bounded support, a density bounded away from zero, (implying $\xi=-1$ by our Lemma \ref{lem:psiOne}) and sufficiently large numbers of auctions.\footnote{In particular, we find in unreported simulations that the coverage of the bootstrap-based confidence intervals varies substantially across data generating processes. For uniformly distributed valuations, coverage is close to the nominal level even for relatively small sample sizes, such as $n = 60$. By contrast, for other distributions, the bootstrap exhibits considerable undercoverage even for much larger sample sizes, such as $n = 500$.}

Taken together, these results highlight the central trade-off underlying inference in auctions with many bidders and few auctions. Methods based on classical large-$n$ asymptotics or resampling techniques can yield shorter confidence intervals but fail to control size in empirically relevant settings. By contrast, our proposed approach prioritizes uniform validity under a many-bidders asymptotic framework, delivering reliable inference in environments where alternative methods are infeasible or unreliable.

\begin{table}[ht]
\centering
\renewcommand{\arraystretch}{1.4}
\begin{tabular}{lllllllllllll}
\hline\hline
\# Bidders & \multicolumn{4}{c}{$K=10$} & \multicolumn{4}{c}{$K=100$} & \multicolumn{4}{c}{$K \sim U\{90,91,\dots,110\}$} \\ 
\# Auctions & \multicolumn{2}{c}{$n=10$} & \multicolumn{2}{c}{$n=100$} & \multicolumn{2}{c}{$n=10$} & \multicolumn{2}{c}{$n=100$} & \multicolumn{2}{c}{$n=10$} & \multicolumn{2}{c}{$n=100$} \\ \hline
Dist.\ & Cov & Lgth & Cov & Lgth & Cov & Lgth & Cov & Lgth & Cov & Lgth & Cov & Lgth \\ \hline
\multicolumn{13}{l}{~~~~~~~~~~~~~~~~~~~Method (i): Our proposed CI, asy.\ with $K\to\infty$} \\ 
$U(0,3)$               & 0.98 & 0.72 & 0.96 & 0.09 & 0.98 & 0.30 & 0.99 & 0.13 & 0.97 & 0.30 & 0.98	& 0.13 \\
{$\vert N(0,1) \vert$} & 0.95 & 1.54 & 0.97 & 0.29 & 0.94 & 1.20 & 0.98 & 0.23 & 0.92 & 1.21 & 0.97 & 0.23 \\ 
{$\vert t(20) \vert$}  & 0.93 & 1.74 & 0.97 & 0.37 & 0.94 & 1.52 & 0.98 & 0.35 & 0.93 & 1.49 & 0.99 & 0.35 \\ 
$Pa(0.25) $            & 0.96 & 1.45 & 0.99 & 0.62 & 0.95 & 2.62 & 0.98 & 1.11 & 0.95 & 2.62 & 0.98 & 1.10 \\ 
 \hline
\multicolumn{13}{l}{~~~~~~~~~~~~~~~~~~~Method (ii): CI based on two highest bids, asy.\ with $n\to\infty$} \\  
$U(0,3)$               & 0.86 & 0.28 & 0.94 & 0.10 & 0.83 & 0.03 & 0.93 & 0.01 & 0.88 & 0.03 & 0.88 & 0.01 \\
{$\vert N(0,1) \vert$} & 0.89 & 0.46 & 0.94 & 0.16 & 0.88 & 0.35 & 0.96 & 0.12 & 0.88 & 0.35 & 0.95 & 0.12 \\ 
{$\vert t(20) \vert$}  & 0.89 & 0.55 & 0.92 & 0.19 & 0.86 & 0.48 & 0.92 & 0.17 & 0.87 & 0.47 & 0.93 & 0.17 \\ 
$Pa(0.25) $            & 0.81 & 0.76 & 0.90 & 0.30 & 0.80 & 1.34 & 0.90 & 0.51 & 0.78 & 1.35 & 0.88 & 0.51 \\ 
 \hline
\multicolumn{13}{l}{~~~~~~~~~~~~~~~~~~~Method (iii): CI based on bootstrap, asy.\ with $n\to\infty$} \\  
$U(0,3)$               & 0.66 & 0.22 & 0.93 & 0.08 & 0.67 & 0.03 & 0.95 & 0.01 & 0.68 & 0.03 & 0.94 & 0.01 \\
{$\vert N(0,1) \vert$} & 0.04 & 0.22 & 0.23 & 0.11 & 0.05 & 0.15 & 0.20 & 0.08 & 0.05 & 0.15 & 0.18 & 0.08 \\ 
{$\vert t(20) \vert$}  & 0.02 & 0.24 & 0.15 & 0.13 & 0.04 & 0.19 & 0.13 & 0.11 & 0.03 & 0.19 & 0.14 & 0.11 \\ 
$Pa(0.25) $            & 0.00 & 0.17 & 0.06 & 0.15 & 0.01 & 0.31 & 0.05 & 0.25 & 0.00 & 0.30 & 0.05	& 0.25 \\ 
\hline\hline
\end{tabular}
\caption{Empirical coverage frequency (Cov) and length (Lgth) of various CIs for the winner's expected utility $\mu_K$ in second-price auctions. The results are the average of 500 simulation draws and a nominal coverage level of 95\%.}\label{tbl:sp-CI}
\end{table}

Next, we consider the problem of hypothesis testing for the tail index based on transaction prices. Motivated by Section \ref{sec:HypApplication}, we focus on $H_0:\xi = -1$ vs.\ $H_1:\xi \in (-1,0.5]$, with $W$ equal to the uniform distribution over this interval. As explained, the validity of this method is based on asymptotics as $K\to\infty$. We consider the same four distributions for $F_V$: (a) the standard uniform distribution over $[0,3]$, (b) the absolute value of standard Normal, (c) the absolute value of Student-t(20), and (d) the Pareto distribution with exponent 0.25. The first distribution satisfies $H_0:\xi = -1$, while the other three belong to $H_1$ with $\xi = 0$, $0.05$, and $0.25$, respectively.

Table \ref{tbl:sp-test} presents the empirical rejection rate of the test proposed in Section \ref{sec:HypApplication} over 500 simulations using a nominal size of $5\%$. Under $H_0$ (i.e., when valuations are $U(0,3)$), our proposed test controls size as long as the number of bidders is not too small relative to the sample size. Under $H_1$ (i.e., when valuations are not $U(0,3)$), our methodology provides reasonable power.

\begin{table}[ht]
\centering
\renewcommand{\arraystretch}{1.5}
\begin{tabular}{lcccccccc}
\hline\hline
\# Bidders & \multicolumn{2}{c}{$K=10$}   & \multicolumn{2}{c}{$K=20$}   & \multicolumn{2}{c}{$K=100$} & \multicolumn{2}{c}{ $K\sim U\{90,91,\dots,110 \}$  } \\ 
\# Auctions & $n=10$ & $n  = 100$   & $n=10$ & $n  = 100$   & $n=10$ & $n  = 100$  & $n=10$ & $n  = 100$\\ \hline
$U(0,3)$               & 0.06 & 0.35 &   0.05 & 0.14 &  0.04 & 0.05 & 0.03 & 0.04 \\ 
{$\vert N(0,1) \vert$} & 0.42 & 1.00 &   0.46 & 1.00 &  0.48 & 1.00 & 0.45 & 1.00 \\ 
{$\vert t(20) \vert$}  & 0.46 & 1.00 &   0.48 & 1.00 &  0.54 & 1.00 & 0.51 & 1.00 \\ 
$Pa(0.25) $            & 0.73 & 1.00 &   0.71 & 1.00 &  0.68 & 1.00 & 0.66 & 1.00 \\ 
\hline\hline
\end{tabular}
\caption{Empirical rejection rate of the test proposed in Section \ref{sec:HypApplication} for $H_0:\xi = -1$ in second-price auctions. The results are the average of 500 simulation draws and a nominal size of 5\%.}\label{tbl:sp-test}
\end{table}


\section{Application to Hong Kong license plate auctions}
\label{sec:application}

In this section, we apply our proposed inference methods to Hong Kong vehicle license plate auctions. We begin with some administrative details about these auctions. Since 1973, the Hong Kong government has used standard oral ascending auctions to sell license plates. As mentioned in Section \ref{sec:sp}, this auction format is weakly equivalent to a second-price auction. There are two types of license plates: personalized and traditional. Personalized plates feature a combination of letters and numbers (subject to certain rules) proposed by the individuals. In contrast, traditional plates have a random combination of letters and numbers assigned by the government. Both types of plates are sold via public auctions, and there is no resale market. The government holds monthly public auctions and posts the available license plates a few days before each auction. Anyone interested in participating must pay a fully refundable deposit. The reserve price is set at 1,000 HKD, which is small relative to the average annual salary (above 435,000 HKD in 2024). The revenue from all auctions will go to the Lotteries Fund for charitable purposes.

Personalized license plates typically feature more attractive letter/number combinations, resulting in higher transaction prices. The government introduced personalized license plates in 2006 and has maintained records of all auction results since then.\footnote{Data are publicly available at https://www.td.gov.hk/en/public\_services/vehicle\_registration\_mark.html} However, aside from the letter/number combinations and transaction prices, the government does not record other auction details, such as the number of bidders. Having said this, we interviewed recent auction attendees to gauge the typical number of participants in these auctions. Our anecdotal evidence suggests that the number of participants in a typical auction is around 100. The government also holds a Lunar New Year special auction every February (or March), during which some very popular plates are sold. In these special auctions, the number of participants can be significantly higher, around 300.

We now examine three cases in detail: (i) luxurious license plates with a single letter, (ii) luxurious license plates with two letters and two numbers, and (iii) ordinary license plates with two letters and four numbers. 
The first two cases contain only personalized plates, while the third contains both personalized and traditional plates.

Before discussing the empirical results, it is useful to clarify how to interpret the widths of the reported confidence intervals. Our confidence intervals are shown to be asymptotically optimal in the sense that they minimize length subject to uniform size control over the admissible class of valuation distributions. Consequently, the interval lengths reported below are driven by the limited information in the data rather than by inefficiencies of the inference method. In particular, when inference relies solely on transaction prices from a small number of auctions and does not impose strong distributional assumptions, any uniformly valid method must exhibit comparable uncertainty. Narrower confidence intervals can be obtained only by imposing additional structure, such as restrictions on tail behavior, which we explore in our analysis.

\subsection{Luxurious license plates with a single letter}

Our first analysis investigates super-luxury license plates that contain only a single letter, as in our running example. Only four such plates have ever been auctioned, and the transaction prices were exorbitant. For example, the plate with a single letter ``\texttt{D}'' was sold on February 25, 2024, for 20.2 million HKD. The total revenue from all auctions that day was 24.5 million HKD, of which 82\% was due to this single luxurious plate. 

Table \ref{tab:case1} presents the letters, auction dates, and transaction prices for all four auctions of single-letter plates. The government's expected revenue and the winner's expected utility are important features of these auctions. To our knowledge, no existing method can perform inference on these quantities using transaction prices from four auctions. Assuming that the value distribution remains unchanged and the number of participants is approximately the same across these four auctions, we construct the 95\% CIs for the expected revenue and the winner's expected utility. Our analysis considers several restrictions on the tail index $\xi$. Our benchmark specification consists of distributions with finite second moments (i.e., $\xi \in [-1,0.5]$), but we complement these with the analogous results for distributions with bounded support or exponentially decaying tails (i.e., $\xi \in [-1,0]$), and bounded-support distributions with positive density at the end of the support (i.e., $\xi=-1$).
The transaction prices are adjusted for inflation\footnote{Annual inflation in Hong Kong was 2.88\%, 0.25\%, 1.57\%, 1.88\%, and 2.10\% in 2019 to 2023, respectively. Data source: https://www.statista.com/statistics/1365695/inflation-rate-in-hong-kong/}, and the CIs are measured in million HKD in 2024. These values are substantial, as shown in Table \ref{tab:case1}.

As explained in the previous paragraph, Table \ref{tab:case1} reports confidence intervals under several specifications about the tail index $\xi$. Allowing a wider range of $\xi$ guarantees uniform validity across a broad class of valuation distributions but necessarily leads to longer confidence intervals, reflecting greater uncertainty about tail behavior. In contrast, imposing tighter restrictions on $\xi$ reduces this uncertainty and results in substantially shorter intervals. The empirical results show that this gain in precision can be sizable, especially in small samples, while maintaining validity under the restricted model.

\begin{table}[ht]
\centering
\renewcommand{\arraystretch}{1.5}
\begin{tabular}{ccccc}
\hline\hline
Plate &  & Auction Date &  & Price (million HKD) \\ \hline
\texttt{D} &  & Feb. 25 2024 &  & $20.2$ \\ 
\texttt{R} &  & Feb. 12 2023 &  & $25.5$ \\ 
\texttt{W} &  & Mar. 07 2021 &  & $26.0$ \\ 
\texttt{V} &  & Feb. 24 2019 &  & $13.0$ \\ \hline
\multicolumn{5}{c}{95\% CI for winner's utility $\mu _{K}$} \\ \hline
\multicolumn{4}{l}{Assuming $\xi\in[-1,0.5]$} & $[2.19$,~ $58.89]$ \\ 
\multicolumn{4}{l}{Assuming $\xi\in[-1,0]$}   & $[2.36$,~ $22.09]$ \\ 
\multicolumn{4}{l}{Assuming $\xi=-1$}   & $[1.53$,~ $13.87]$ \\    \hline            
\multicolumn{5}{c}{95\% CI for seller's revenue $\pi _{K}$} \\ \hline
\multicolumn{4}{l}{Assuming $\xi\in[-1,0.5]$}   & $[8.70$,~ $37.65]$ \\ 
\multicolumn{4}{l}{Assuming $\xi\in[-1,0]$}     & $[8.52$,~ $31.62]$ \\ 
\multicolumn{4}{l}{Assuming $\xi=-1$}           & $[8.77$,~ $28.99]$ \\  \hline
\multicolumn{3}{l}{$p$-value for $H_{0}:\xi =-1$} &  & $0.64$ \\ 
\hline\hline
\end{tabular}
\caption{95\% CIs for the winner's expected utility and the government's expected revenue (in million HKD, adjusted by inflation) and $p$-values for the test in \eqref{eq:compTest} for $H_0:\xi = -1$ in Hong Kong car license plate auctions, using four license plates with a single letter.}
\label{tab:case1}
\end{table}

\subsection{Luxurious plates with two letters and two numbers}

Our second analysis focuses on a single auction held on February 25, 2024, the most recent Lunar New Year special auction. Based on interviews with participants, the number of bidders in this auction was very large. On that day, only five personalized license plates featuring two letters and two numbers were sold. Table \ref{tab:case21} shows that their transaction prices ranged from 23,000 HKD to 185,000 HKD. Given the similarity in the combinations of letters and numbers, it is reasonable to assume that their underlying valuation distributions are the same, which is our primary assumption. Using our proposed method, we construct the 95\% CI for the government's expected revenue and the winner's expected utility. Additionally, our five observations allow us to reject the hypothesis that $\xi = -1$, suggesting that this dataset does not satisfy the standard regularity conditions in the auction literature.

\begin{table}
\renewcommand{\arraystretch}{1.5}
\centering
\begin{tabular}{ccccc}
\hline\hline
Plate &  & Auction Date &  & Price (1,000 HKD) \\ \hline
\texttt{BE22} &  & Feb. 25 2024 &  & $40.0$ \\ 
\texttt{HE68} &  & Feb. 25 2024 &  & $23.0$ \\ 
\texttt{WE88} &  & Feb. 25 2024 &  & $185.0$ \\ 
\texttt{LY38} &  & Feb. 25 2024 &  & $40.0$ \\ 
\texttt{ME33} &  & Feb. 25 2024 &  & $105.0$ \\ \hline
\multicolumn{5}{c}{95\% CI for winner's utility $\mu _{K}$} \\ \hline
\multicolumn{4}{l}{Assuming $\xi\in[-1,0.5]$} & $[31.58$,~ $523.98]$ \\ 
\multicolumn{4}{l}{Assuming $\xi\in[-1,0]$}   & $[29.78$,~ $226.91]$ \\ 
\multicolumn{4}{l}{Assuming $\xi=-1$}         & $[31.41$,~ $183.59]$ \\    \hline            
\multicolumn{5}{c}{95\% CI for seller's revenue $\pi _{K}$} \\ \hline
\multicolumn{4}{l}{Assuming $\xi\in[-1,0.5]$}   & $[2.84$,~ $211.80]$ \\ 
\multicolumn{4}{l}{Assuming $\xi\in[-1,0]$}     & $[1.00$,~ $174.80]$ \\ 
\multicolumn{4}{l}{Assuming $\xi=-1$}           & $[1.00$,~ $170.70]$ \\  \hline
\multicolumn{3}{l}{$p$-value for $H_{0}:\xi =-1$} &  & $0.03$ \\ 
\hline\hline
\end{tabular}
\caption{95\% CIs for the winner's expected utility and the government's expected revenue (in 1,000 HKD) and $p$-values for the test in \eqref{eq:compTest} for $H_0:\xi = -1$ in Hong Kong car license plate auctions, using five license plates with two letters and two numbers.}\label{tab:case21}
\end{table}

We notice that three of these five plates share the same two numbers, i.e., \texttt{BE22}, \texttt{WE88}, and \texttt{ME33}. We now repeat the analysis with these three plates. The left panel of Table \ref{tab:case22} again presents the 95\% CIs for the government's expected revenue and the winner's expected utility. Compared with the previous table, using only three observations results in much wider confidence intervals, as expected. The lower bound of the expected revenue reaches the reserve price of 1,000 HKD. Note that having three auctions is the minimum amount we require. In the right panel of Table \ref{tab:case22}, we collect three plates with letters \texttt{UU} and numbers 28, 38, and 68, respectively. We adjusted the inflation and reported the CIs in 1,000 HKD in 2021. The results are comparable to those in the left panel. 

\begin{table}
\renewcommand{\arraystretch}{1.5}
\centering
\begin{adjustbox}{max width=\textwidth}
\begin{tabular}{ccccccccc}
\hline\hline
Plate &  & Auction Date & Price (1,000 HKD) &  & Plate &  & Auction Date & 
Price (1,000 HKD) \\ \hline
\texttt{BE22} &  & Feb. 25 2024 & $40.0$  &  & \texttt{UU28} &  & Mar. 07 2021 & $170.0$ \\ 
\texttt{WE88} &  & Feb. 25 2024 & $185.0$ &  & \texttt{UU38} &  & Feb. 24 2019 & $59.0$ \\ 
\texttt{ME33} &  & Feb. 25 2024 & $105.0$ &  & \texttt{UU68} &  & Feb. 24 2019 & $95.0$ \\ \hline
\multicolumn{4}{c}{95\% CI for winner's utility $\mu _{K}$} & & \multicolumn{4}{c}{95\% CI for winner's utility $\mu _{K}$}\\
\multicolumn{3}{l}{Assuming $\xi\in[-1,0.5]$} & $[31.66$,~ $1122.33]$ & & \multicolumn{3}{l}{Assuming $\xi\in[-1,0.5]$} & $[21.61$,~ $848.34]$\\ 
\multicolumn{3}{l}{Assuming $\xi\in[-1,0]$}   & $[27.39$,~ $437.86]$ & & \multicolumn{3}{l}{Assuming $\xi\in[-1,0]$} & $[19.87$,~ $339.34]$\\ 
\multicolumn{3}{l}{Assuming $\xi=-1$}   & $[17.38$,~ $304.68]$ & &  \multicolumn{3}{l}{Assuming $\xi=-1$} & $[13.97$,~ $243.76]$\\    \hline  
\multicolumn{4}{c}{95\% CI for seller's revenue $\pi_{K}$} & & \multicolumn{4}{c}{95\% CI for seller's revenue $\pi_{K}$}\\
\multicolumn{3}{l}{Assuming $\xi\in[-1,0.5]$} & $[1.00$,~ $402.79]$ & & \multicolumn{3}{l}{Assuming $\xi\in[-1,0.5]$} & $[1.00$,~ $331.48]$\\ 
\multicolumn{3}{l}{Assuming $\xi\in[-1,0]$}   & $[1.00$,~ $294.71]$ & & \multicolumn{3}{l}{Assuming $\xi\in[-1,0]$} & $[1.00$,~ $253.05]$\\ 
\multicolumn{3}{l}{Assuming $\xi=-1$}   & $[1.00$,~ $267.20]$ & &  \multicolumn{3}{l}{Assuming $\xi=-1$} & $[1.00$,~ $235.34]$\\    \hline  
\multicolumn{3}{l}{$p$-value for $H_{0}:\xi =-1$} & $0.33$ &  & 
\multicolumn{3}{l}{$p$-value for $H_{0}:\xi =-1$} & $0.24$ \\ 
\hline\hline
\end{tabular}
\end{adjustbox}
\caption{95\% CIs for the winner's expected utility and the government's expected revenue (in 1,000 HKD) and $p$-values for the test in \eqref{eq:compTest} for $H_0:\xi = -1$ in Hong Kong car license plate auctions, using three license plates with two letters and the same two numbers (left) and the same two letters \texttt{UU} and two numbers (right).}\label{tab:case22}
\end{table}

\subsection{Ordinary plates with two letters and four numbers}

Our third analysis studies ordinary plates. Due to its large volume, the government only releases the results of the three most recent auctions on its website. We use the dataset obtained from the Hong Kong Transport Department by \cite{ng/chong/du:2010}. The data include a detailed description of 292 license plate auctions conducted between 1997 and 2008, resulting in the sale of 40,000 license plates. On average, each auction sells more than a hundred plates in sequence. In our asymptotic framework, each license plate is treated as an individual instance of a second-price auction with a large number of bidders. Although the dataset includes numerous individual license plate auction instances, significant heterogeneity exists among them. In fact, the main finding in \cite{ng/chong/du:2010} is that certain plates are more valuable due to superstition. To account for this observed heterogeneity, we focus on ordinary license plates that meet specific criteria: the letters are not \texttt{HK}, the letters are not the same (e.g., not \texttt{AA}, \texttt{BB}, \texttt{CC}, etc.), the numbers on the plate are not in order (e.g., not 1369), or in reverse order (e.g., not 9631), the number part of the plate has 4 digits, none of these digits is an 8 or a 4 (considered fortunate or unfortunate), and the transaction price exceeds the reserve price. After imposing these restrictions, we are left with 318 license plates sold between 1997 and 2008. These data are further divided by year, allowing valuation distributions to vary over time. This results in 12 separate datasets, one for each year from 1997 to 2008, with an average of 26.5 auctions per year. We assume that the selected auctions within each year are homogeneous, with a large and relatively constant number of bidders. Under these conditions, we can implement our methods for each one of these years.

Table \ref{tab:case3} displays the confidence intervals for the expected utility of the winner and the $p$-values of the test \eqref{eq:compTest} with the null hypothesis $H_{0}:\xi=-1$ using data from each year. For brevity, the confidence intervals are restricted to the benchmark case with $\xi \in [-1,0.5]$. Our analysis reveals several interesting empirical observations. First, the expected utility of the winner is economically substantial, with the midpoints of the 95\% confidence intervals ranging from approximately 1,500 to 7,000 HKDs (equivalent to 192 to 895 USD at the current exchange rate). Second, prior to 2006, the average confidence interval had a midpoint of approximately 3,500 HKD and a width of approximately 4,500 HKD. After 2006, these figures decreased to 1,600 HKD and 2,300 HKD, respectively. We speculate that these differences could be attributed to the introduction of special license plates in these auctions in 2006. Third, compared to the winner's utility, the CI for the government's expected revenue is much narrower. There also seems to be a decrease in the midpoint after 2006, which could also be explained by the introduction of special license plates. Finally, the hypothesis of $\xi =-1$ is strongly rejected for all years. This suggests that the standard regularity conditions imposed by traditional inference methods in auction literature do not hold in this dataset.

\begin{table}[ht]
\centering
\renewcommand{\arraystretch}{1.5}
 \begin{adjustbox}{max width=\textwidth}
\begin{tabular}{cccccccccccccccc}
\hline\hline
year & $n$ & 95\% CI & 95\% CI & $p$-value for &  & year & $n$ & 95\% CI & 
95\% CI & $p$-value for \\ 
&  & for $\mu _{K}$ & for $\pi _{K}$ & $H_{0}$:$\xi =-1$ &  &  &  & for $\mu_{K}$ & for $\pi _{K}$ & $H_{0}$:$\xi =-1$ \\ \hline
1997 & 26 & $[2.23,9.23]$ & $[3.20,5.61]$ & $0.00$ &  & 2003 & 46 & $[0.82,2.98]$ & $[2.25,3.00]$ & $0.00$ \\ 
1998 & 27 & $[1.71,7.29]$ & $[2.97,4.89]$ & $0.00$ &  & 2004 & 12 & $[0.69,5.59]$ & $[2.44,4.00]$ & $0.00$ \\ 
1999 & 32 & $[1.22,5.72]$ & $[3.58,4.63]$ & $0.00$ &  & 2005 & 22 & $[0.61,3.83]$ & $[2.18,3.23]$ & $0.00$ \\ 
2000 & 29 & $[1.58,6.63]$ & $[2.92,4.62]$ & $0.00$ &  & 2006 & 17 & $[0.12,1.69]$ & $[2.03,2.52]$ & $0.00$ \\ 
2001 & 34 & $[1.23,4.71]$ & $[2.74,3.82]$ & $0.00$ &  & 2007 & 31 & $[0.69,3.07]$ & $[2.40,3.10]$ & $0.00$ \\ 
2002 & 18 & $[1.02,6.30]$ & $[2.62,4.23]$ & $0.00$ &  & 2008 & 24 & $[0.58,3.54]$ & $[2.16,3.16]$ & $0.00$ \\ \hline\hline
\end{tabular}
\end{adjustbox}
\caption{95\% CIs for the winner's expected utility $\mu_K$ and the government's expected revenue $\pi_K$ (in 1,000 HKD and assuming $\xi \in [-1,0.5]$) and $p$-values for the test in \eqref{eq:compTest} for $H_0:\xi = -1$ in Hong Kong car license plate auctions, using ordinary plates with two letters and four different numbers.}\label{tab:case3}
\end{table}


\section{Conclusions}\label{sec:conclusion}

This paper considers statistical inference problems in auctions with many bidders with symmetric, independent private values. In this context, we present a novel asymptotic framework in which the number of bidders grows while the number of auctions remains small and fixed. This approach differs from the more conventional approach, which uses a large number of auctions while keeping the number of bidders fixed. We argue that our results provide an accurate approximation in auction settings where the number of bidders is large relative to the number of auctions. This framework is especially well-suited for applications with substantial heterogeneity across auctions, where only a few are truly homogeneous.

Within our novel asymptotic framework, we introduce new inference methods for the expected utility of the auction winner, the expected revenue for the seller, and the tail behavior of the valuation distribution. We show that the latter can serve as a means of testing the validity of the regularity conditions typically imposed in the auction literature (e.g., bounded valuation support and nonzero density at the upper end of the support). Our data requirements are minimal; our tests require only observing transaction prices from a fixed, finite set of auctions. That is, we do not require observing multiple bids from these auctions or the number of bidders.

Our methodology relies on the fact that, as the number of bidders grows, the transaction prices reveal the tail properties of the valuation distribution. This information is sufficient to provide asymptotically valid inference for the above-mentioned objects of economic interest. Within our asymptotic framework, our inference methods are shown to control size and have desirable power properties. We demonstrate through Monte Carlo simulations that our methodology provides an accurate approximation in finite samples. Finally, we apply our methods to Hong Kong license plate auction data.

There are several related avenues for future research. In particular, an important related question concerns asymptotic environments in which both the number of bidders and the number of auctions grow large. In such settings, inference may rely on two distinct approximations: an extreme value approximation governing transaction prices as $K \to \infty$, and a sampling approximation as $n \to \infty$. In principle, one could identify the valuation distribution in such a regime if the approximation error induced by the extreme value limit as $K \to \infty$ vanishes sufficiently quickly relative to the sampling uncertainty as $n \to \infty$. That said, the development of such results is fundamentally different from the analysis carried out in this paper, both conceptually and in spirit. Accordingly, we leave the study of such joint asymptotic regimes for future work.

Conversely, we conjecture that our results may continue to apply in joint asymptotic regimes in which $K \to \infty$ sufficiently faster than $n \to \infty$, so that the approximation induced by the extreme value limit remains dominant relative to the sampling uncertainty arising from $n \to \infty$.\footnote{The convergence rate of the EVT approximation varies substantially across underlying distributions. For example, the convergence rate for the exponential distribution is of order $n^{-1}$, whereas for the normal distribution it is $(\log n)^{-1}$. See \citet[][Chapter 5.3]{dehaan2006book} for details.} Even in settings with a large number of auctions, our framework remains attractive when the number of bidders is unknown, valuations may have unbounded support, or only transaction prices are observed. In these cases, many existing methods are either infeasible or rely on strong parametric assumptions, whereas our approach continues to deliver valid inference without requiring knowledge of the number of bidders or a full structural specification.

\newpage
\appendix
\section{Appendix}\label{sec:appendix}
The appendix is organized as follows. 
Section \ref{sec:appendi-reserve} extends our previous analysis for the optimal reserve price. 
Section \ref{sec:appendix-additional} reviews commonly used parametric distributions in the existing auction literature and presents simulation studies using these distributions. 
Section \ref{sec:appendix-fp} provides additional analysis of the first-price auctions. Section \ref{sec:appendix-proof} collects all of the proofs omitted from the main text. Section \ref{sec:appendix-aux} gives intermediate results. Section \ref{sec:appendix-computation} provides computational details omitted from the main text. Section \ref{sec:simulation-fp} presents additional Monte Carlo simulations, focusing on first-price auctions. Finally, Section \ref{sec:appendix-MC} provides auxiliary derivations related to the Monte Carlo simulations.
\subsection{Extension to optimal reserve price}\label{sec:appendi-reserve}
When a seller sets a sufficiently high reserve price before the auction begins, it effectively screens out bidders whose valuations fall below this threshold. The imposition of a binding reserve price does not alter the Bayesian Nash equilibrium strategy, thereby leaving the bidding function unchanged. In second-price auctions, bidders continue to bid their valuations, whereas in first-price auctions, the bidding function remains consistent with the prior formulation \eqref{eq:bid_first}. However, the introduction of a binding reservation price poses a novel challenge for existing methodologies that rely on observing bidder counts. Specifically, these methods may necessitate information on both the potential and actual numbers of bidders. For instance, \citet{LiZheng2009} defines the number of potential bidders as the number of plan holders, whereas \citet{RobertsSweeting2013} considers it in the context of related auctions. The number of actual bidders is typically inferred from the number of submitted bids.

In our many-bidder framework, let $K_{j}$ represent the number of \textit{potential} bidders in the $j$-th auction. Given that the bidding function remains unaffected, the EV convergence as per Lemma \ref{lem:joint} still holds as $K_{j}\rightarrow \infty $. Consequently, as long as the winning bid surpasses the reserve price, our preceding analysis relying solely on the transaction price remains applicable. This insight proves invaluable in empirical settings where the observed number of bidders may be limited despite a potentially large pool of prospective participants, such as in
online auctions where the number of potential buyers is effectively infinite \citep[e.g.,][p.128]{gentry2018review}.

Another pertinent example arises in the domain of art painting auctions. As discussed by \citet{BeggsGladdy2009}, art is typically auctioned in an English or ascending format, with the transaction price equal to the valuation of the second-highest bidder. Given the diverse characteristics of paintings, including dimensions, titles, sizes, and artist provenance, the recurrence of identical auctions for the same artwork is limited. However, the potential pool of bidders remains substantial and unobservable, prompting auction houses to strategically set aggressive reserve prices to
selectively screen participants.

Furthermore, the imposition of a binding reserve price complicates the identification of the entire valuation distribution. Denoting $v_{0}$ as the reserve price, \citet[][Section 4]{guerre2000} establishes the identification of $F_{V}( v) $ for $v\geq v_{0}$. Nevertheless, our focus on the tail characteristics of the distribution remains unaffected provided that $v_{0}$ lies strictly below the upper endpoint $v_H$.

Actually, our new framework leads to an alternative optimal reserve price that maximizes the seller's expected revenue. In the classic setup where the seller's own value of the good, denoted by $v_{0K}$, is constant, \citet[][Proposition 3]{rileysamuelson1981} show that the optimal reserve price is also constant regardless of the number of bidders. As the number of bidders increases, the transaction price will almost surely exceed the reserve price \citep[cf.][]{wilson1977}. As a consequence, the effect of the reserve price on the seller's revenue becomes asymptotically negligible.

In principle, one could achieve a non-negligible asymptotic reserve price by considering an alternative asymptotic framework in which the seller's value of the good is of the same order of magnitude as the transaction price. To develop this result, consider the seller's expected revenue as a function of the reserve price:
\begin{equation*}
\pi _{K}( \gamma _{K}) ~=~E\Big[ v_{0K} I\Big[V_{(1),j} \leq \gamma_K \Big]+ \gamma_{K}I\Big[ V_{(2),j }\leq \gamma _{K}\leq V_{(1),j }\Big] +V_{( 2),j }I\Big[ \gamma _{K}\leq V_{(2),j}\Big]  \Big].
\end{equation*}
The right-hand side has three terms. The first term reflects the case in which the highest valuation (and bid) does not exceed the reservation price, so the good remains unsold, and the seller obtains his own valuation. The second term corresponds to the case in which the reservation price lies between the highest and second-highest valuations (and bids), so the good is sold at the reservation price. Finally, the third term reflects the case in which the second-highest valuation (and bid) exceeds the reservation price, so the good is sold at the second-highest bid.

To analyze the problem as $K\to \infty$, it is convenient to normalize by $b_K$ and scale by $a_K$, yielding
\begin{align*}
& \tfrac{\pi _{K}( \gamma _{K})  -b_{K}}{a_{K}}~=~\left\{
\begin{array}{c}
 \tfrac{v_{0K }-b_{K} }{a_{K}}P\Big[ \tfrac{V_{(1),j }-b_{K}}{a_{K}} \leq \tfrac{\gamma _{K}-b_{K}}{a_{K}} \Big]+ \tfrac{\gamma _{K}-b_{K} }{a_{K}}P\Big[ \tfrac{V_{( 2),j }-b_{K}}{a_{K}}< \tfrac{\gamma _{K}-b_{K}}{a_{K}}\leq \tfrac{V_{( 1),j }-b_{K}}{a_{K}}\Big] \\
 + E\Big[\tfrac{V_{( 2),j }-b_{K} }{a_{K}}I\Big[ \tfrac{\gamma _{K}-b_{K}}{a_{K}}\leq \tfrac{V_{(2),j }-b_{K}}{a_{K}}\Big]\Big] 
\end{array}\right\}.
\end{align*}
To make the reservation prices salient in our asymptotic framework, we take $(v_{0K }-b_{K})/a_{K} \to v_0$ and $( \gamma_{K}-b_{K}) /a_{K}\to \gamma$ as $K\to \infty$. Under these conditions and \eqref{eq:DomainOfAttraction}, the arguments used in Lemma \ref{lem:joint} imply
\begin{align*}
\tfrac{\pi _{K}( \gamma _{K})  -b_{K}}{a_{K}}~\to ~\left\{
 \begin{array}{c}
 v_0 P\left( H_{\xi}(E_{1,j}) \leq \gamma \right)  + 
\gamma P\left( H_{\xi}(E_{1,j}+E_{2,j})\leq \gamma\leq H_{\xi}(E_{1,j})\right)\\
+E\left[ H_{\xi}(E_{1,j}+E_{2,j})I\left[ \gamma\leq H_{\xi}(E_{1,j}+E_{2,j})\right] \right]
\end{array}\right\} ~=:~\pi( \gamma ) .
\end{align*}
where $(H_{\xi}(E_{1,j}),H_{\xi}(E_{1,j}+E_{2,j}))$ are distributed as in \eqref{eq:EVT}, with joint density
\begin{align}
&f_{H_{\xi}(E_{1,j}),H_{\xi}(E_{1,j}+E_{2,j})}(x_{1},x_{2})\notag\\
&=~I[x_2 \leq x_1]\left\{ 
\begin{array}{lc}
( 1+\xi x_{1}) ^{-\frac{1+\xi }{\xi }}( 1+\xi x_{2})^{-\frac{1+\xi }{\xi }}\exp\big( -( 1+\xi x_{2}) ^{-1/\xi}\big)  &   ~\text{if }\xi \neq 0 \\ 
\exp (-x_{1})\exp (-x_{2})\exp ( -\exp (-x_{2}))  &   ~\text{if }
\xi =0
\end{array}\right.   \label{eq:pdf_e1e2}
\end{align}
The following result derives the asymptotic expression of the optimal reserve price under a mild tail condition of $\xi<1$.
\begin{lemma}\label{lem:sp-rp}
Under \eqref{eq:pdf_e1e2} with $\xi<1$, $\gamma^{\ast }\equiv \arg \max_{\gamma}
\pi\left( \gamma\right) =(1+v_0)/(1-\xi )$.
\end{lemma}

Lemma \ref{lem:sp-rp} expresses the optimal reserve price $\gamma^{\ast }$ as a function of the limiting value of the scaled and normalized seller's valuation $v_0$ and the tail index $\xi$.
See Section \ref{sec:appendix-aux} for its proof. 
Given a known (or hypothesized) value of $v_0$ (e.g., $v_0=0$), we can repeat the arguments used in Section \ref{sec:sp-revenue} to construct a CI $U({\bf P})$ for $\gamma^{\ast}_K$. The main difference between the two CIs is that the limiting distribution of $( ({\pi_K - P_{(1)} })/({P_{(n)}-P_{(1)}}), \tilde{\bf P} )$ (denoted by $f_{(Y_{\pi},\tilde{\mathbf{Z}})|\xi}(y,h)$ in Section \ref{sec:sp-revenue}) is replaced by the limiting distribution of $( ({\gamma^{\ast}_K - P_{(1)} })/({P_{(n)}-P_{(1)}}), \tilde{\bf P} )$, where $\gamma^{\ast}_K \equiv \arg \max_{\gamma_K} \pi_K\left( \gamma_K\right)$.
\subsection{Commonly used parametric distributions and additional simulations}\label{sec:appendix-additional}
\subsubsection{Review of some parametric distributions}

Structural auction studies frequently adopt parametric assumptions about bidders’ valuations, costs, or information signals to obtain tractable equilibrium conditions and estimable likelihood functions. The most commonly used distributions in this literature, summarized in Table~\ref{tab:parametric_families_evt}, include the Weibull, Log-Normal, Gamma, Pareto, and Normal families. These arise in influential applications such as timber procurement auctions \citep{paarsch1992,paarsch1997,flambardperrigne2006}, comparisons of ascending and sealed-bid formats \citep{athey2011}, structural models of entry and competition \citep{lizheng2012}, online auctions \citep{BajariHortacsu2003}, and affiliated-signal environments \citep{hendrickspinkseporter2003}.

The use of the Weibull distribution in timber auctions is notable for its flexibility: with shape parameter $k>1$, it is light-tailed and belongs to the Gumbel domain of attraction, while for $k<1$ it becomes heavy-tailed and enters the Fr\'echet domain. Log-Normal specifications, common in settings with strictly positive and right-skewed valuations, also fall in the Gumbel domain, despite their slower tail decay. Gamma distributions, sometimes used as alternatives to Log-Normal forms, likewise exhibit exponential tails and belong to the Gumbel class. Pareto specifications, frequently employed when modeling environments with large underlying uncertainty, such as mineral-rights or exploration auctions \citep{paarsch1992}, feature regularly varying tails and lie in the Fr\'echet domain. Normal signals, used in affiliated-value and information-based models \citep{hendrickspinkseporter2003}, are extremely light-tailed yet fall within the Gumbel domain.

A central implication of Table~\ref{tab:parametric_families_evt} is that \emph{all} of the major parametric families used in empirical auction research fall within the domain of attraction of an extreme-value distribution -- either Gumbel or Fr\'echet. This means that the extreme-value regularity conditions invoked in theoretical and semi-parametric auction analyses are broadly satisfied in practice. The EVT perspective, therefore, provides a unified treatment of how the tail behavior of valuations affects bidding, revenue predictions, and the sensitivity of counterfactual results to distributional assumptions.

\begin{table}[htbp]
\centering
\begin{tabularx}{\textwidth}{l X c c}
\toprule
\textbf{Distribution} & \textbf{Representative Papers} & \textbf{DoA} & \textbf{Tail Index $\xi$} \\
\midrule
Weibull     & \citet{paarsch1992,paarsch1997,flambardperrigne2006} & Yes & $k\cdot1[k<1]$ \\
Pareto      & \citet{paarsch1992} & Yes & $1/\alpha$ \\
Log-Normal   & \citet{athey2011,lizheng2012,roberts2016bailouts} & Yes & $0$ \\
Gamma       & \citet{athey2011} & Yes & $0$ \\
Normal      & \citet{hendrickspinkseporter2003} & Yes & $0$ \\
\bottomrule
\end{tabularx}
\caption{Parametric families used in structural auction models and their tail properties}
\label{tab:parametric_families_evt}
\end{table}

\subsubsection{Additions simulations}
We now examine the finite-sample performance of our proposed CI and three competing ones with these previously adopted parametric distributions. 
In particular, we consider the Log-Normal distribution with $\mu=\sigma=1$, the Gamma distribution with $\theta=\lambda=1$, and the Weibull distribution with $k=0.5$ and $\lambda=1$. All three distributions satisfy the domain-of-attraction assumption with $\xi=0$ for Log-Normal and Gamma distributions, and $\xi=k$ for the Weibull distribution. 
Note that location and scale shift do not affect the coverage and length properties due to our equivariance/invariance property. 
Table \ref{tbl:sp-CI-additional} presents the results, using the same setup as in Table \ref{tbl:sp-CI}. 
Across all distributions, our proposed confidence interval, derived under the many-bidders asymptotic framework, achieves coverage probabilities close to the nominal 95\% level even with a small number of auctions. 
Coverage improves further as either the number of bidders or the number of auctions increases. 
In contrast, confidence intervals based on asymptotics with a diverging number of auctions tend to undercover when the number of auctions is small, whereas bootstrap-based intervals perform particularly poorly, exhibiting severe undercoverage regardless of the underlying distribution. 
In terms of length, our proposed intervals are generally wider than those of competing methods when sample sizes are small, reflecting the cost of uniform validity, but their lengths shrink rapidly as the number of auctions grows.

Table \ref{tbl:sp-CI-restrict} investigates the role of shape restrictions on the tail index in determining the length of confidence intervals. 
Imposing restrictions on the parameter space, such as ruling out bounded-support distributions (assuming $\xi\leq0$) or restricting attention to light-tailed classes (assuming $\xi=-1$), can substantially reduce interval length without compromising coverage when the restrictions are valid. The gains from shape restrictions are most pronounced in small-sample settings, where unrestricted inference must remain robust to a wide range of tail behaviors. 
These results highlight the practical value of incorporating economically meaningful shape information when available, while also illustrating the robustness of our proposed method in the absence of such restrictions.

\begin{table}[ht]
\centering
\renewcommand{\arraystretch}{1.4}%
\begin{tabular}{lllllllll}
\hline\hline
\# Bidders & \multicolumn{4}{c}{$K=10$} & \multicolumn{4}{c}{$K=100$} \\ 
\# Auctions & \multicolumn{2}{c}{$n=10$} & \multicolumn{2}{c}{$n=100$} & \multicolumn{2}{c}{$n=10$} & \multicolumn{2}{c}{$n=100$} \\ \hline
Dist.\ & Cov & Lgth & Cov & Lgth & Cov & Lgth & Cov & Lgth\\ \hline
\multicolumn{9}{l}{~~~~~~~~~Method (i): Our proposed CI, asy.\ with $K\to\infty$} \\ 
{$\log N(0,1)$}        & 0.96 & 6.37 & 1.00 & 2.85 & 0.96 & 13.08 & 0.99 & 5.50 \\ 
{$\Gamma(1,1)$}        & 0.97 & 3.26 & 0.99 & 1.06 & 0.94 & 3.39 & 0.99 & 0.87 \\ 
$W(1,0.5)$             & 0.97 & 13.70 & 1.00 & 5.33 & 0.94 & 30.05 & 0.99 & 12.70 \\ 
 \hline
\multicolumn{9}{l}{~~~~~~~~Method (ii): CI based on two highest bids, asy.\ with $n\to\infty$} \\  
{$\log N(0,1)$}        & 0.78 & 3.57 & 0.92 & 1.41 & 0.82 & 6.29 & 0.92 & 2.35 \\ 
{$\Gamma(1,1)$}        & 0.85 & 1.08 & 0.91 & 0.39 & 0.87 & 1.10 & 0.92 & 0.39 \\ 
$W(1,0.5)$             & 0.81 & 8.33 & 0.91 & 3.19 & 0.82 & 13.20 & 0.92 & 4.80 \\ 
 \hline
\multicolumn{9}{l}{~~~~~~~~Method (iii): CI based on bootstrap, asy.\ with $n\to\infty$} \\  
{$\log N(0,1)$}        & 0.00 & 0.78 & 0.06 & 0.74 & 0.00 & 1.48 & 0.07 & 1.25 \\ 
{$\Gamma(1,1)$}        & 0.02 & 0.40 & 0.15 & 0.26 & 0.03 & 0.42 & 0.11 & 0.25 \\ 
$W(1,0.5)$             & 0.00 & 1.63 & 0.06 & 1.73 & 0.01 & 3.56 & 0.07 & 2.76 \\ 
\hline\hline
\end{tabular}
\caption{Empirical coverage frequency (Cov) and length (Lgth) of various CIs for the winner's expected utility $\mu_K$ in second-price auctions. The results are the average of 500 simulation draws and a nominal coverage level of 95\%. $\log N$, $\Gamma$, and $W$ stand for Log-Normal, Gamma, and Weibull distributions. }\label{tbl:sp-CI-additional}
\end{table}

\begin{table}[ht]
\centering
\renewcommand{\arraystretch}{1.4}%
\begin{tabular}{lllllllll}
\hline\hline
\# Bidders & \multicolumn{4}{c}{$K=10$} & \multicolumn{4}{c}{$K=100$}  \\ 
\# Auctions & \multicolumn{2}{c}{$n=10$} & \multicolumn{2}{c}{$n=100$} & \multicolumn{2}{c}{$n=10$} & \multicolumn{2}{c}{$n=100$} \\ \hline
$U(0,3)$   & Cov & Lgth & Cov & Lgth & Cov & Lgth & Cov & Lgth  \\ \hline
Assuming $\xi\in[-1,0.5]$       & 0.98 & 0.72 & 0.96 & 0.09 & 0.98 & 0.30 & 0.99 & 0.13  \\
Assuming $\xi\in[-1,0]$         & 0.97 & 0.44 & 0.92 & 0.08 & 0.98 & 0.14 & 0.99 & 0.01  \\
Assuming $\xi=-1$               & 0.97 & 0.31 & 0.96 & 0.08 & 0.94 & 0.05 & 0.99 & 0.01  \\
\hline\hline
\end{tabular}
\caption{Empirical coverage frequency (Cov) and length (Lgth) of various CIs for the winner's expected utility $\mu_K$ in second-price auctions. The results are the average of 500 simulation draws and a nominal coverage level of 95\%. Data are generated from uniform distribution $U[0,3]$. }\label{tbl:sp-CI-restrict}
\end{table}

\subsection{Additional details about first-price auctions}\label{sec:appendix-fp}
Section \ref{sec:fp} derives the limiting behavior of the transaction price $P_j$. Analogously to the second-price auctions, we now construct inference methods about the winner's expected utility, seller's expected revenue, and the tail index itself in the subsequent subsections.

\subsubsection{Inference about the winner's expected utility}\label{sec:fp-supplus}

Our goal is to conduct inference on the average of the winner's expected utility using the transaction prices. Since bidders act according to \eqref{eq:bid_first}, the auction is won by the bidder with the highest valuation. From here, we conclude that the winner's expected utility in auction $j=1,\dots,n$ is $E[V_{(1),j}-P_{j}]$, whose average is
\begin{equation}
\mu _{K}~{=}~\frac{1}{n}\sum_{j=1}^{n}E\Big[\frac{ \int_{-\infty}^{V_{(1),j}}{F_{V}(u)^{K_{j}-1}}du}{F_{V}(V_{(1),j})^{K_{j}-1}}\Big] .
\label{eq:fp_muK_as_fn_ofV}
\end{equation}
By the Revenue Equivalence Theorem (e.g., \citet[Section 3]{krishna2009book}), $\mu _{K}$ coincides with the average of the winner's expected utility in second-price auctions.

The construction of the CI for $\mu _{K}$ closely follows the arguments and derivations presented in Section \ref{sec:sp-supplus}. In particular, we propose a CI for $\mu _{K}$ given by
\begin{equation}
U(\mathbf{P})~=~(P_{(n)}-P_{(1)})\times \bigg\{y\in \mathbb{R}:\int_{\xi \in \Xi }\kappa _{\xi} (\mathbf{\tilde{P}})f_{\tilde{\mathbf{X}}|\xi }(\mathbf{ \tilde{P}})dW(\xi )\leq \int_{\xi \in \Xi }f_{(Y_{\mu},\tilde{\mathbf{X}})|\xi }(y,\mathbf{\tilde{P}})d\Lambda (\xi )\bigg\},  \label{eq:CI_muK_defn2_fp}
\end{equation}
where $\mathbf{\tilde{P}}=\{\tilde{P}_{j}:j=1,\dots ,N\}$ with $\tilde{P}_{j}$ as in \eqref{eq:P*fp}, $\mathbf{\tilde{X}}=\{\tilde{X}_{j}:j=1,\dots ,N\}$ as in \eqref{eq:Xtilde}, $\kappa _{\xi }(h)=E[ X_{( n) }-X_{( 1) }|\mathbf{\tilde{X}}=h] $, and $Y_{\mu }=E[ \int\nolimits_{-\infty }^{X_{1,j}}{G}_{\xi }{(u)}du/G_{\xi }( X_{1,j}) ]/(X_{( n) }-X_{( 1) }) $. By repeating arguments in Section \ref{sec:sp-revenue}, \eqref{eq:CI_muK_defn2_fp} belongs to the class of asymptotically-valid CIs and minimizes the asymptotic expected length within this class. See Theorem \ref{thm:CI_K_fp} in the appendix for a statement of this result.

Unfortunately, implementing the CI in \eqref{eq:CI_muK_defn2_fp} is considerably more challenging than in the case of second-price auctions. The main reason is that there is no closed-form expression for the PDF of $\mathbf{\tilde{X}}$ available for first-price auctions. 
Without these, we cannot easily compute $\kappa _{\xi} (h)f_{\tilde{\mathbf{X}}|\xi }(h)$ and $f_{(Y_{\mu},\tilde{\mathbf{X}})|\xi }(y,h)$ for $(y,h) \in \mathbb{R} \times \Sigma$. To sidestep this issue, we use a numerical approximation to the problem based on a Taylor series expansion. See Section \ref{sec:appendix-computation} for a detailed explanation.

\subsubsection{Inference about the seller's expected revenue}

We now conduct inference on the average of the seller's expected revenue using transaction prices. Since bidders act according to \eqref{eq:bid_first} and the transaction price equals the highest bid, the average of the seller's expected revenue is given by
\begin{equation}
\pi _{K}~= 
~\frac{1}{n}\sum_{j=1}^{n}E\Big[V_{(1),j}-\frac{\int_{-\infty }^{V_{(1),j}}{F_{V}(u)^{K_{j}-1}}du}{F_{V}(V_{(1),j})^{K_{j}-1}}\Big].  \label{eq:seller2}
\end{equation}

Reiterating arguments in Section \ref{sec:sp-revenue}, we propose a CI for $\mu _{K}$ given by
\begin{align}
U(\mathbf{P})~=~& P_{(1)} + ( P_{(n)}-P_{(1)}) \notag \\ 
\times& \bigg\{ y \in \mathbb{R} : \int_{\xi \in \Xi} \kappa_{\xi}(\mathbf{\tilde{P}}) f_{\tilde{\mathbf{X}}|\xi}(\mathbf{\tilde{P}}) dW(\xi) \leq \int_{\xi \in \Xi} f_{(Y_{\pi},\tilde{\mathbf{X}})|\xi}(y,\mathbf{\tilde{P}}) d\Lambda(\xi) \bigg\} , \label{eq:CI_piK_defn2_fp}
\end{align}
where $Y_{\pi }=E[ X_{1,j} -{\int_{-\infty }^{X_{1,j} }G_{\xi }( h) dh}/{G_{\xi }( X_{1,j} ) }]/( X_{( n) }-X_{( 1) }) $ and the rest of the objects are as in Section \ref{sec:fp-supplus}.
By the same arguments in Section \ref{sec:sp-revenue}, \eqref{eq:CI_muK_defn2_fp} belongs to the class of asymptotically valid CIs and minimizes the asymptotic expected length within this class. See Theorem \ref{thm:CI_pi_K_fp} in the appendix for a statement of this result.

Implementing the CI in \eqref{eq:CI_piK_defn2_fp} suffers from the computational issues elaborated in Section \ref{sec:fp-supplus}, which we avoid via numerical approximations. Once again, see Section \ref{sec:appendix-computation} for details.

\subsubsection{Inference about the tail index}\label{sec:HypApplication-fp}

Finally, we consider the problem of inference about the tail index from transaction prices in first-price auctions. As in Section \ref{sec:sp-index}, we are interested in the following hypothesis test:
\begin{equation}
H_{0}:\xi \in \xi _{0}~~~\text{v.s.}~~~H_{1}:\xi \in \Xi _{1},
\label{eq:HT0_fp}
\end{equation}
where $\xi _{0}$ is a fixed parameter value in $\Xi $ and $\Xi _{1}=\Xi \backslash \{ \xi _{0}\}$. For brevity, we focus only on the case where the alternative hypothesis in \eqref{eq:HT0_fp} is composite. By our previous arguments, we propose using the feasible version of the generalized likelihood ratio test in the limiting problem, i.e.,
\begin{equation}
\varphi _{K}^{\ast }( \mathbf{P}) ~\equiv ~1\bigg[ {\int_{\Xi _{1}}f_{\mathbf{\tilde{X}}|\xi }( \mathbf{\tilde{P}}) dW( \xi ) }/{f_{\mathbf{\tilde{X}}|\xi _{0}}( \mathbf{\tilde{P}}) }>q( \xi _{0},W,\alpha ) \bigg] ,  \label{eq:compTest_fp}
\end{equation}
where $\mathbf{\tilde{P}}$, $\mathbf{\tilde{X}}$, and $\kappa _{\xi }(h)$ are as in Section \ref{sec:fp-supplus},
$W$ is the user-defined weight function on $\Xi _{1}$, and $q( \xi _{0},W,\alpha ) $ is the critical value of the likelihood ratio test in \eqref{eq:HT0_fp}, i.e.,
\begin{equation*}
q( \xi _{0},W,\alpha ) ~\equiv ~( 1-\alpha ) \text{-quantile of }{\int_{\Xi _{1}}f_{\mathbf{\tilde{X}}|\xi _{0}}(  \mathbf{\tilde{X}}_{0}) dW( \xi ) }/{f_{\mathbf{\tilde{X}}|\xi _{0}}( \mathbf{\tilde{X}}_{0}) }.
\end{equation*}
and $\mathbf{\tilde{X}}_{0}$ is a random vector with PDF $f_{\mathbf{\tilde{X}}|\xi _{0}}$. By our previous arguments, it follows that \eqref{eq:compTest_fp} is asymptotically valid and asymptotically level $\alpha $, and efficient in the sense of maximizing the asymptotic weighted average power criterion in the class of asymptotically valid and equivariant tests. See Theorem \ref{thm:compTest_fp} in the appendix for a statement of this result. The implementation of the hypothesis test in \eqref{eq:compTest_fp} suffers from the aforementioned computational issues, and is addressed by the approximation described in Section \ref{sec:appendix-computation}. 

To conclude, we note that the arguments presented in Section \ref{sec:HypApplication} indicate that the hypothesis test in \eqref{eq:HT0_fp} can be applied to test the standard regularity conditions in the auctions literature. This can be achieved by conducting the hypothesis test in \eqref{eq:HT0_fp} with $\xi _{0}=-1$ and $\Xi _{1} = (-1,0.5]$.


\subsection{Proofs of results in the main text}\label{sec:appendix-proof}

\begin{proof}[Proof of Lemma \ref{lem:joint}]
By independence across auctions, it suffices to prove that for any auction $j=1,\ldots n$ and as $K\to\infty$,
\begin{align}
 \big( \tfrac{V_{(1),j}-b_{K}}{a_{K}},\dots ,~\tfrac{V_{(d),j}-b_{K}}{a_{K}}\big) ~\overset{d}{\to }~  \big( H_{\xi }(E_{1,j}),\dots ,~H_{\xi }\big( \sum\nolimits_{s=1}^{d}E_{s,j}\big) \big),\label{eq:joint_1}
\end{align}
where $\{E_{s,j}:s=1,\dots ,d\}$ are i.i.d.\ standard exponential random variables. 

Let $x$ be any continuity point of $G_{\xi }$. Since $K_{j}\to \infty $ as $K\to \infty $, \eqref{eq:DomainOfAttraction} implies that there is a sequence of constants $\{(a_{K_{j}},b_{K_{j}})\in \mathbb{R} _{++}\times \mathbb{R}:K_{j}\in \mathbb{N}\}$ such that
\begin{equation}
\lim_{K \to \infty} F( a_{K_{j}}x+b_{K_{j}}) ^{K_{j}}~=~ G_{\xi }( x)~\text{ as }K\to \infty. \label{eq:joint_2}
\end{equation}
Under \eqref{eq:joint_2}, \citet[Theorem 2.1.1]{dehaan2006book}, 
implies that as $K\to\infty$,
\begin{align}
\Psi_{j} \equiv\big( \tfrac{V_{(1),j}-b_{K_{j}}}{a_{K_{j}}},\dots ,~\tfrac{V_{(d),j}-b_{K_{j}}}{a_{K_{j}}}\big) ~\overset{d }{\to }~ \big( H_{\xi }(E_{1,j}),\dots ,~H_{\xi }\big( \sum\nolimits_{s=1}^{d}E_{s,j}\big) \big). \label{eq:joint_3} 
\end{align}
Note that
\begin{align}
\big( \tfrac{V_{(1),j}-b_{K}}{a_{K}},\dots ,\tfrac{V_{(d),j}-b_{K}}{a_{K}}\big)
&~=~ 
\Psi_{j} \tfrac{ a_{K_{j}}}{a_{K}} ~+~\big( \tfrac{b_{K}-b_{K_{j}}}{a_{K}},\dots ,\tfrac{b_{K}-b_{K_{j}}}{a_{K}}\big). \label{eq:joint_4}
\end{align}
By \eqref{eq:joint_3} and \eqref{eq:joint_4}, \eqref{eq:joint_1} follows from showing that
\begin{equation}
\lim_{K \to \infty} \big( \tfrac{a_{K_{j}}}{a_{K}},~\tfrac{b_{K_{j}}-b_{K}}{a_{K}}\big) ~=~ ( 1,~0) . \label{eq:joint_5}
\end{equation}
We devote the remainder of this proof to establishing \eqref{eq:joint_5}.

Let $x$ be any continuity point of $G_{\xi }$. By \eqref{eq:joint_2} and $K_{j}/K\to 1$ as $K\to \infty$,
\begin{equation}
\lim_{K \to \infty} F( a_{K_{j}}x+b_{K_{j}}) ^{K}~=~ G_{\xi }( x) ~\text{ as }K\to \infty . \label{eq:joint_6}
\end{equation}
Under \eqref{eq:DomainOfAttraction} and \eqref{eq:joint_6}, \citet[Lemma 1]{dehaan1976} implies that for all $ y\in \mathbb{R}$,
\begin{equation}
G_{\xi }\big( \lim_{K\to \infty }\big( \tfrac{a_{K_{j}}}{a_{K}}y+ \tfrac{b_{K_{j}}-b_{K}}{a_{K}}\big) \big) ~=~G_{\xi }( y) . \label{eq:joint_7}
\end{equation}
Regardless of $\xi$, there is a continuum of values for which $G_{\xi }$ is strictly increasing and, thus, invertible. For any $y$ in this continuum, \eqref{eq:joint_7} implies that
\begin{equation}
\lim_{K \to \infty}  \tfrac{a_{K_{j}}}{a_{K}}y+\tfrac{ b_{K_{j}}-b_{K}}{a_{K}} ~=~y . \label{eq:joint_8}
\end{equation}
From this observation, \eqref{eq:joint_5} follows.
\end{proof}

\begin{proof}[Proof of Lemma \ref{lem:densitySecond}]
This result follows from Lemma \ref{lem:joint}, Lemma \ref{lem:cts}, and the continuous mapping theorem.
\end{proof}


\begin{proof}[Proof of Theorem \ref{thm:CI_K}]
\underline{Part 1}. By \eqref{eq:secondprice} and Lemma \ref{lem:joint}, as $K\to \infty$,
\begin{equation}
\big\{\big(\tfrac{V_{(1),j}-P_{j}}{a_{K}},\tfrac{P_{j}-b_{K}}{a_{K}}\big) :j=1,\ldots ,N\big\}~\overset{d}{\to } ~\{(Z_{1,j}-Z_{2,j},Z_{2,j}):j=1,\ldots ,N\} . \label{eq:CI_muK2}
\end{equation}
Let $\delta >0$ and $\bar{K} \in \mathbb{N}$ be as in Lemma \ref{lem:finite_muK2}. Then, for all $j =1,\dots,N$ and $K>\bar{K}$,
\begin{equation}
E\big[ | \tfrac{V_{(1),j}-P_{j}}{a_{K}}| ^{1+\delta }\big] ~\overset{(1)}{=}~
E\big[ |\tfrac{V_{( 1) ,j}-V_{(2) ,j}}{a_{K}}| ^{1+\delta }\big]~\leq~ 2^{\delta }\big( E\big[ | \tfrac{V_{( 1) ,j}-b_{K}}{a_{K}}| ^{1+\delta } +  | \tfrac{V_{( 2) ,j}-b_{K}}{a_{K}}|^{1+\delta }\big] \big)~\overset{(2)}{<}~\infty ,  \label{eq:CI_muK2p5}
\end{equation}
where (1) holds by \eqref{eq:bidSecond} and \eqref{eq:secondprice}, and (2) by Lemma \ref{lem:finite_muK2}. By \eqref{eq:winner1}, \eqref{eq:CI_muK2}, and \eqref{eq:CI_muK2p5}, as $K\to \infty$,
\begin{equation}
\big\{ \tfrac{\mu _{K}}{a_{K}},\big\{\big( \tfrac{P_{j}-b_{K}}{a_{K}} \big) :j=1,\ldots ,N \big\} \big\}~
\overset{d}{\to }~ \{ E [Z_{1,1} - Z_{2,1}], \{ Z_{2,j} : j=1,\ldots ,N \} \},\label{eq:CI_muK3}
\end{equation}
where we have used that $Z_{1,1}-Z_{2,1}\overset{d}{=}Z_{1,j}-Z_{2,j}$ for all $j=1,\ldots ,n$. Define
\begin{equation}
R_{K}~\equiv ~\bigg\{ 
\begin{array}{cc}
\tfrac{\mu _{K}}{P_{(n)}-P_{(1)}} & \text{if }P_{(n)}\not=P_{(1)}, \\ 
0 & \text{if }P_{(n)}=P_{(1)}.
\end{array}
  \label{eq:CI_muK3b}
\end{equation}
By \eqref{eq:CI_muK3}, \eqref{eq:CI_muK3b}, Lemma \ref{lem:cts}, and the continuous mapping theorem, we have that, as $K\to \infty $,
\begin{equation}
\{R_{K},\mathbf{\tilde{P}}\}~\overset{d}{\to }~\{Y_{\mu },\mathbf{ \tilde{Z}}\}. \label{eq:CI_muK4}
\end{equation}

Next, consider the following derivation as $K\to \infty $.
\begin{align*}
P(\mu _{K}\in U(\mathbf{P}))
& ~=~P(\{ \mu _{K}\in U(\mathbf{P})\} \cap \{ P_{(n)}\not=P_{(1)}\} )+P(\{ \mu _{K}\in U(\mathbf{P} )\} \cap \{ P_{(n)}=P_{(1)}\} ) \\
& ~\overset{(1)}{=}~P(\{ \mu _{K}\in ( P_{(n)}-P_{(1)}) \times \tilde{U}(\mathbf{\tilde{P}})\} \cap \{ P_{(n)}\not=P_{(1)}\} )+o(1) \\
& ~\overset{(2)}{=}~P(\{R_{K}\in \tilde{U}(\mathbf{\tilde{P}})\}\cap \{P_{(n)}\not=P_{(1)}\})+o(1) \\
& ~\overset{(3)}{=}~P(R_{K}\in \tilde{U}(\mathbf{\tilde{P}}))+o(1) \\
& ~=~P(\{R_{K},\mathbf{\tilde{P}}\}\in \{(y,h)\in \mathbb{R}\times \Sigma:y\in \tilde{U}(h)\})+o(1) \\
& ~\overset{(4)}{\to }~P(\{Y_{\mu },\mathbf{\tilde{Z}}\}\in \{(y,h)\in \mathbb{R}\times \Sigma:y\in \tilde{U}(h)\}) \\
& ~=~P(Y_{\mu }\in \tilde{U}(\mathbf{\tilde{Z}})),
\end{align*}
as desired, where (1) holds by \eqref{eq:CI_muK_defn} and Lemma \ref{lem:cts}, (2) holds because \eqref{eq:CI_muK3b} implies that $\{ \mu _{K}\in ( P_{(n)}-P_{(1)}) \times \tilde{U}(\mathbf{\tilde{P}})\} =\{R_{K}\in \tilde{U}(\mathbf{\tilde{P}})\}$ under $P_{(n)}\not=P_{(1)}$, (3) by Lemma \ref{lem:cts}, and (4) by \eqref{eq:CI_muK4}, Assumption (a), and the Portmanteau Theorem.

\underline{Part 2}. By similar arguments as those used in the proof of Lemma \ref{lem:densitySecond}, we have that, as $K\to \infty $,
\begin{equation}
\{{(P_{(n)}-P_{(1)})}/{a_{K}},\mathbf{\tilde{P}}\}~\overset{d}{ \to }~\{Z_{(n)}-Z_{(1)},\mathbf{\tilde{Z}}\}. \label{eq:CI_muK7}
\end{equation}
By \eqref{eq:CI_muK7}, Assumption (c), and the continuous mapping theorem, we have that, as $K\to \infty $, 
\begin{equation}
{((P_{(n)}-P_{(1)})}/{a_{K}}) \times \lg (\tilde{U}(\mathbf{\tilde{P}}))~\overset{d}{ \to }~(Z_{(n)}-Z_{(1)})\times\lg (\tilde{U}(\mathbf{\tilde{Z}})).  \label{eq:CI_muK8}
\end{equation}

Let $\delta >0$ and $\bar{K}\in \mathbb{N}$ be as in Lemma \ref{lem:finite_muK2}. Then,
\begin{align}
|P_{(n)}-P_{(1)}|^{1+\delta }~&\leq~(\sum_{l,m=1}^{n}|P_{l}-P_{m}|)^{1+\delta }\notag\\
~&\overset{(1)}{=}~(\sum_{l,m=1}^{n}|V_{(2),l}-V_{(2),m}|)^{1+\delta }\notag\\
~&{\leq }~(2n)^{1+\delta}\sum_{l=1}^{n}|V_{(2),l}-b_{K}|^{1+\delta },  \label{eq:CI_muK9}
\end{align}
where (1) holds by \eqref{eq:secondprice}. Then, for all $K>\bar{K}$, 
\begin{align}
E[ |\tfrac{(P_{(n)}-P_{(1)})}{a_{K}} \lg (\tilde{U}(\mathbf{\tilde{P}}))|^{1+\delta }] ~\overset{(1)}{\leq }~
(2n)^{1+\delta }\sum_{l=1}^{n}E\big[|\tfrac{V_{(2),l}-b_{K}}{a_{K}}|^{1+\delta } |\lg (\tilde{U}(\mathbf{\tilde{P}}))|^{1+\delta }\big] 
~\overset{(2)}{<}~\infty ,\label{eq:CI_muK10} 
\end{align}%
where (1) holds by \eqref{eq:CI_muK9}, and (2) by Assumption (b), $\lg (\tilde{U}(h))\geq 0$ for all $h \in \Sigma$, and Lemma \ref{lem:finite_muK2}.

Then, consider the following derivation as $K\to \infty $.
\begin{align*}
E[\lg (U(\mathbf{P}))/a_{K}]~&\overset{(1)}{=}~
E[(( P_{(n)}-P_{(1)}) /a_{K})\times \lg ( \tilde{U}(\mathbf{\tilde{P}}))] \\
~&\overset{(2)}{\to }~E[(Z_{(n)}-Z_{(1)})\times\lg (\tilde{U}(\mathbf{\tilde{Z}} ))]\\
~&=~E_{\xi }[\kappa _{\xi }(\mathbf{\tilde{Z}})\times\lg (\tilde{U}( \mathbf{\tilde{Z}}))],
\end{align*}
as desired, where (1) holds by \eqref{eq:CI_muK_defn}, and (2) by \eqref{eq:CI_muK8} and \eqref{eq:CI_muK10}.
\end{proof}

\begin{proof}[Proof of Theorem \ref{thm:CI_pi_K}]
This proof follows from the same arguments as in Theorem \ref{thm:CI_K}. The main difference between the proofs is that the vector $\{ (V_{(1),j}-P_{j}):j=1,\ldots ,N\} $ 
is replaced by $\{ (P_{j}-b_{K}):j=1,\ldots ,N\} $,
which then implies that $( \pi _{K}-b_{K}) $ and $Y_{\mu }$ are replaced by $\mu _{K}$ and $Y_{\pi }$, respectively. With these changes in place, the desired result follows from repeating the steps used to prove Theorem \ref{thm:CI_K}.
\end{proof}

\begin{proof}[Proof of Theorem \ref{thm:simpleTest}]
This proof follows from applying \citet[Theorem 1]{muller:2011} based on the convergence result in Lemma \ref{lem:densitySecond}. To apply these results in that paper, we specify the connection between our object and the relevant objects in that paper: the class of DGPs satisfying \eqref{eq:DomainOfAttraction} take the role of $\mathcal{M}$, $K\to \infty$ takes the role of $T\to \infty$, $j=1,\ldots ,n$ the role of $ i=1,\ldots ,n$, $\{ \xi _{0}\} $ the role of $\Theta _{0}$, $ \{ \xi _{1}\} $ the role of $\Theta _{1}$, $\mathbf{{P}}$ the role of $Y_{T}$, $\mathbf{\tilde{P}}$ the role of $X_{T}$, $\Sigma $ the role of $S$, the self-normalizing transformation $\{( P_{j}-P_{( 1) })/( P_{( n) }-P_{( 1) })\}_{j=1}^{N}:\mathbf{P} \to \Sigma$ the role of $h_{T}$, $\mathbf{\tilde{Z}}$ the role of $X$, $\tilde{\varphi}^{\ast }(\mathbf{\tilde{Z}}) \equiv {\varphi}^{\ast }(\mathbf{{Z}}) $ the role of $\varphi _{S}( Y) $, $\varphi _{K}^{\ast }( \mathbf{P}) $ the role of $\hat{\varphi}_{T}^{\ast }( \mathbf{P}) $, $\varphi _{K}( \mathbf{P}) $ the role of $\varphi _{T}(\mathbf{P}) $, and $f_{\mathbf{\tilde{Z}}|\xi _{0}}$ the role of $\mu _{P}$. Moreover, we set $\mathcal{F}_s$ equal to the set with the zero function and $\pi_0=0$, which effectively trivializes \citet[Equations (5)-(6) and (10)-(11)]{muller:2011}.

By definition, $\tilde{\varphi}^{\ast }(\mathbf{\tilde{Z}}):\Sigma \to [ 0,1]$ is a level-$\alpha $ test in the limiting problem, i.e., \eqref{eq:sizeAlpha} holds. According to the Neyman-Pearson Lemma, it maximizes power in the limiting problem. Finally, note that $\tilde{\varphi}^{\ast }$ is continuous except for a set of zero $f_{\mathbf{\tilde{Z}}|\xi _{0}}$-measure. These conditions align with the requirements in \citet[Theorem 1]{muller:2011}, and so our desired results follow immediately from its application.
\end{proof}

\begin{proof}[Proof of Theorem \ref{thm:compTest}]
This proof follows from the same arguments as in Theorem \ref{thm:simpleTest}.
\end{proof}

\begin{proof}[Proof of Lemma \ref{lem:psiOne}]
Fix $x>0$ arbitrarily. By $v_H<\infty $ and \citet[Theorem 1.2.1]{dehaan2006book}, it suffices show that 
\begin{equation*}
\lim_{t\downarrow 0}\tfrac{1-F_{V}(v_H-tx)}{1-F_{V}(v_H-t)}=x.
\end{equation*}
To this end, consider the following derivation.
\begin{equation*}
\lim_{t\downarrow 0}\tfrac{1-F_{V}(v_H-tx)}{1-F_{V}(v_H-t)}
~\overset{(1)}{=}~\lim_{t\downarrow 0}\tfrac{xf_{V}(v_H-tx)}{f_{V}(v_H-t)}
~\overset{(2)}{=}~x,
\end{equation*}
as desired, where (1) holds by L'H\^{o}pital rule and (2) by $f_{V}( v) \to f_{V}( v_H)>0 $ as $v\uparrow v_H$.
\end{proof}

\begin{proof}[Proof of Lemma \ref{lem:densityFirst}]
This result follows from Lemma \ref{lem:cts}, Lemma \ref{lem:joint_first}, and the continuous mapping theorem.
\end{proof}

\subsection{Auxiliary results}\label{sec:appendix-aux}

\begin{lemma}\label{lem:cts}
For any $n>1$ and for first-price or second-price auctions, 
\begin{equation*}
P(P_{(n)}=P_{(1)})~=~0.
\end{equation*}
\end{lemma}
\begin{proof}
Note that 
\begin{equation}
\{ P_{(n)}=P_{(1)}\} ~\overset{(1)}{\subseteq }~\underset{l,m=1,\ldots,n,l\neq  m}{\bigcap}~\underset{i_{l}=1,\ldots ,K_{l},i_{m}=1,\ldots ,K_{m}}{\bigcup}\{V_{i_{l},l}=V_{i_{m},m}\} ,  \label{eq:cts}
\end{equation}
where (1) holds by \eqref{eq:secondprice} and \eqref{eq:bid_first} for second-price and first-price auctions, respectively. Then, 
\begin{equation*}
P(P_{(n)}=P_{(1)})~\overset{(1)}{\leq }~\min_{l,m=1,\ldots,n,l\neq m}~\sum_{i_{l}=1,\ldots ,K_{l},i_{m}=1,\ldots ,K_{m}}P(V_{i_{l},l}=V_{i_{m},m}) ~\overset{(2)}{=}~0,
\end{equation*}
as desired, where (1) holds by \eqref{eq:cts} and (2) by the fact that valuations are continuously distributed and $l,m$ refer to different and, thus, independent, auctions.
\end{proof}

\begin{lemma}\label{lem:finite_muK2}
Assume \eqref{eq:DomainOfAttraction} holds. For some $\varepsilon>0$ with $(1+\varepsilon)\xi<1$, assume that $E[ |V_{i,j}|^{1+\varepsilon}] <\infty $ for all $i=1\dots,K_j$ in auction $j=1,\dots,N$. Then, $\exists \bar{K} \in \mathbb{N}$ such that for $\delta = \varepsilon/2$ and $K \geq \bar{K}$,
\begin{align}
E[ \vert \tfrac{V_{(1),j}-b_{K}}{a_{K}} \vert^{1+\delta} ] ~<~\infty~~~~\text{and}~~~~
E[ \vert \tfrac{V_{(2),j}-b_{K}}{a_{K}} \vert^{1+\delta} ] ~<~\infty .  \label{eq:integrab0}
\end{align}
\end{lemma}
\begin{proof}
The first result in \eqref{eq:integrab0} is a corollary of \citet[Theorem 5.3.1]{dehaan2006book}. Thus, we focus the rest of this proof on the second result in \eqref{eq:integrab0}.

It is convenient to introduce the following notation. For any $\tilde{K}\in \mathbb{N} $, let $U({\tilde{K}})$ distributed according to $\max\{V_{i,j}:j=1,\dots,\tilde{K}\}$. For any $x\in [ v_{L},v_{H}] $, the CDF and PDF of $U({\tilde{K}})$ are
\begin{equation}
F_{U{(\tilde{K})}}( x) ~=~F_{V}(x) ^{ \tilde{K}}~~~~\text{and}~~~~f_{U({\tilde{K}})}( x) ~=~ \tilde{K}( F_{V}(x)) ^{\tilde{K}-1}f_{V}( x) . \label{eq:integrab1}
\end{equation}

Conditional on $V_{(1),j}=x\in [ v_{L},v_{H}] $, $V_{(2),j}$ is the highest valuation among the remaining $K_{j}-1$ bidders. Thus, for any $ t\leq x$ with $x,t\in [ v_{L},v_{H}] $, $P( V_{(2),j}\leq t|V_{(1),j}=x) =( F_{V}(t)/F_{V}(x)) ^{K_{j}-1}$. By setting $\tilde{t}=( t-b_{K}) /a_{K}$ and $\tilde{x}=( x-b_{K}) /a_{K}$ for any $x,t\in [ v_{L},v_{H}] $ with $t\leq x$,
\begin{equation*}
P\Big( \tfrac{V_{(2),j}-b_{K}}{a_{K}}\leq \tilde{t}\Big|\tfrac{V_{(1),j}-b_{K}}{ a_{K}}=\tilde{x}\Big) ~=~\Big( \tfrac{F_{V}(\tilde{t}a_{K}+b_{K})}{F_{V}( \tilde{x}a_{K}+b_{K})}\Big) ^{K_{j}-1}
\end{equation*}
and, so, the conditional PDF is
\begin{equation}
f_{\Big(\tfrac{V_{(2),j}-b_{K}}{a_{K}}\Big|\tfrac{V_{(1),j}-b_{K}}{a_{K}}=\tilde{x}\Big)}( \tilde{t})~=~\tfrac{( K_{j}-1) ( F_{V}(\tilde{t} a_{K}+b_{K})) ^{K_{j}-2}f_{V}(\tilde{t}a_{K}+b_{K})a_{K}}{F_{V}( \tilde{x}a_{K}+b_{K}) ^{K_{j}-1}}~=~\tfrac{f_{U(K_{j}-1)}( \tilde{t}a_{K}+b_{K}) a_{K}}{F_{U(K_{j}-1)}( \tilde{x}a_{K}+b_{K}) }. \label{eq:integrab2}
\end{equation}
For any $x\in [ v_{L},v_{H}] $, we have $P( ({V_{(1),j}-b_{K}})/{a_{K}}\leq \tilde{x}) =( F_{V}( \tilde{x}a_{K}+b_{K})) ^{K_{j}}$ and, so,
\begin{equation}
f_{\tfrac{V_{(1),j}-b_{K}}{a_{K}}}( \tilde{x}) ~=~K_{j}( F_{V}( \tilde{x}a_{K}+b_{K})) ^{K_{j}-1}f_{V}(\tilde{x} a_{K}+b_{K})a_{K}~=~f_{U(K_{j})}( \tilde{x} a_{K}+b_{K}) a_{K}. \label{eq:integrab3}
\end{equation}

Consider the following derivation:
\begin{align}
&E\big[ \big\vert \tfrac{V_{(2),j}-b_{K}}{a_{K}}\big\vert ^{1+\delta } \big] \notag\\
&=~E\big[  E\big[ \big\vert \tfrac{V_{(2),j}-b_{K}}{a_{K}} \big\vert ^{1+\delta }\big\vert \tfrac{V_{(1),j}-b_{K}}{a_{K}}\big] \big] \notag\\
& \overset{(1)}{=}~\int_{\tilde{x}=\tfrac{v_{L}-b_{K}}{a_{K}}}^{\tilde{x}= \tfrac{v_{H}-b_{K}}{a_{K}}}\Big( \int_{\tilde{t}=\tfrac{v_{L}-b_{K}}{a_{K}}}^{ \tilde{t}=\tilde{x}}\vert \tilde{t}\vert ^{1+\delta }\tfrac{ f_{U(K_{j}-1)}( \tilde{t}a_{K}+b_{K}) a_{K}}{ F_{U(K_{j}-1)}( \tilde{x}a_{K}+b_{K}) }d\tilde{t }\Big) f_{U(K_{j})}( \tilde{x}a_{K}+b_{K}) a_{K}d\tilde{x} \nonumber \\
& \overset{(2)}{=}~
\int_{\tilde{x}=\tfrac{v_{L}-b_{K}}{a_{K}}}^{\tilde{x}=\tfrac{v_{H}-b_{K}}{ a_{K}}}\Big( E\big[ \big\vert \tfrac{U(K_{j}-1)-b_{K}}{ a_{K}}\big\vert ^{1+\delta }1[ U(K_{j}-1)\leq \tilde{x}a_{K}+b_{K}] \big] \Big) \tfrac{f_{U(K_{j})}( \tilde{x}a_{K}+b_{K}) }{F_{U(K_{j}-1)}( \tilde{x}a_{K}+b_{K}) }a_{K}d\tilde{x}, \label{eq:integrab4}
\end{align}
where (1) holds by \eqref{eq:integrab2} and \eqref{eq:integrab3} and (2) by the change of variables $t=\tilde{t}a_{K}+b_{K}$.

For any $p,q>1$ with $1/p+1/q=1$, we have
\begin{align}
&E\big[ \big\vert \tfrac{U(K_{j}-1)-b_{K}}{ a_{K}}\big\vert ^{1+\delta }1[ U(K_{j}-1)\leq \tilde{x}a_{K}+b_{K}] \big] \notag \\
&\overset{(1)}{\leq }~E\big[ \big\vert \tfrac{U(K_{j}-1)-b_{K_{j}-1}}{a_{K_{j}-1}}\tfrac{a_{K_{j}-1}}{a_{K}}+\tfrac{ b_{K_{j}-1}-b_{K}}{a_{K}}\big\vert ^{( 1+\delta ) p}\big] ^{1/p}F_{U(K_{j}-1)}( \tilde{x}a_{K}+b_{K}) ^{1/q} \nonumber \\
&\overset{(2)}{\leq }~
\Bigg\{\begin{array}{c}
2^{\tfrac{( 1+\delta ) p-1}{p}}\Big[ E\big[ \big\vert \tfrac{U(K_{j}-1)-b_{K_{j}-1}}{ a_{K_{j}-1}}\big\vert ^{( 1+\delta ) p}\big] \big\vert \tfrac{ a_{K_{j}-1}}{a_{K}}\big\vert ^{( 1+\delta ) p}+\big\vert \tfrac{ b_{K_{j}-1}-b_{K}}{a_{K}}\big\vert ^{( 1+\delta ) p}\Big] ^{1/p}\\
\times F_{U(K_{j}-1)}( \tilde{x}a_{K}+b_{K}) ^{1/q}
\end{array}\Bigg\}, \label{eq:integrab5}
\end{align}
where (1) holds by H\"{o}lder's inequality and (2) by Minkowski's inequality. Next, let $C>2^{{(( 1+\delta ) p-1)}/{p}}[ E( \vert Z_{\xi}\vert ^{( 1+\delta ) p}) ] ^{1/p}$ where $Z_{\xi}$ has CDF $G_{\xi }$, and $p= ( 1+\varepsilon ) /( 1+\delta ) =( 2+2\varepsilon ) /( 2+\varepsilon )>1 $ (since $\delta =\varepsilon /2$). Under these conditions, we have that
\begin{equation}
\lim_{K\to \infty }\Big[ E\big[ \big\vert \tfrac{U(K_{j}-1)-b_{K_{j}-1}}{ a_{K_{j}-1}}\big\vert ^{( 1+\delta ) p}\big] \big\vert \tfrac{ a_{K_{j}-1}}{a_{K}}\big\vert ^{( 1+\delta ) p}+\big\vert \tfrac{ b_{K_{j}-1}-b_{K}}{a_{K}}\big\vert ^{( 1+\delta ) p}\Big] ^{1/p}~<~C. \label{eq:integrab6}
\end{equation}
This result is a corollary of two observations. First, we note that \eqref{eq:joint_5} implies that $\lim_{K \to \infty}(a_{K_{j}-1}/a_{K},( b_{K_{j}-1}-b_{K}) /a_{K}) = (1,0)$. Second, under $p= ( 1+\varepsilon ) /( 1+\delta )$, $(1+\varepsilon )\xi <1$, and $E[ \vert V_{i,j}\vert ^{1+\varepsilon }] <\infty $,  \citet[Theorem 5.3.1]{dehaan2006book} implies that
\begin{equation*}
 \lim_{K \to \infty} E( \vert \tfrac{U(K_{j}-1)-b_{K_{j}-1}}{ a_{K_{j}-1}}\vert ^{( 1+\delta ) p}) ~=~ \lim_{K \to \infty} E( \vert \tfrac{U(K_{j}-1)-b_{K_{j}-1}}{ a_{K_{j}-1}}\vert ^{1+\varepsilon}) ~=~ E[ \vert Z_{\xi}\vert ^{1+\varepsilon }] ~<~\infty .
\end{equation*}
Note that \eqref{eq:integrab5} and \eqref{eq:integrab6} imply that $\exists \bar{K}\in \mathbb{N} $ such that for all $K>\bar{K}$,
\begin{equation}
E\big[ \big\vert \tfrac{U(K_{j}-1)-b_{K}}{ a_{K}}\big\vert ^{1+\delta }1[ U(K_{j}-1)\leq \tilde{x}a_{K}+b_{K}] \big] ~\leq~ C \times F_{U(K_{j}-1)}( \tilde{x}a_{K}+b_{K}) ^{1/q}. \label{eq:integrab7}
\end{equation}

To complete the proof, consider the following argument for all $K>\max\{\bar{K},2\}$,
\begin{align*}
&E[ \vert \tfrac{V_{(2),j}-b_{K}}{a_{K}}\vert ^{1+\delta } ] \notag\\
&\overset{(1)}{\leq }~C\int_{\tfrac{v_{L}-b_{K}}{a_{K}}}^{\tfrac{ v_{H}-b_{K}}{a_{K}}}F_{U(K_{j}-1)}( \tilde{x} a_{K}+b_{K}) ^{1/q-1}f_{U(K_{j})}( \tilde{x} a_{K}+b_{K}) a_{K}d\tilde{x} \\
&\overset{(2)}{=}~C\int_{\tfrac{v_{L}-b_{K}}{a_{K}}}^{\tfrac{v_{H}-b_{K}}{a_{K}}}F_{U(K_{j}-1)}( \tilde{x}a_{K}+b_{K})^{1/q-1}K_{j}( F_{V}(\tilde{x}a_{K}+b_{K})) ^{K_{j}-1}f_{V}( \tilde{x}a_{K}+b_{K}) a_{K}d\tilde{x} \\
&\overset{(3)}{\leq }~2C\int_{\tfrac{v_{L}-b_{K}}{a_{K}}}^{\tfrac{v_{H}-b_{K}}{ a_{K}}}F_{U(K_{j}-1)}( \tilde{x}a_{K}+b_{K}) ^{1/q-1}( K_{j}-1) ( F_{V}(\tilde{x}a_{K}+b_{K})) ^{K_{j}-2}f_{V}( \tilde{x}a_{K}+b_{K}) a_{K}d\tilde{x} \\
&\overset{(4)}{=}~2C\int_{\tfrac{v_{L}-b_{K}}{a_{K}}}^{\tfrac{v_{H}-b_{K}}{ a_{K}}}F_{U(K_{j}-1)}( \tilde{x}a_{K}+b_{K}) ^{1/q-1}f_{U(K_{j}-1)}( \tilde{x}a_{K}+b_{K}) a_{K}d\tilde{x} \\
&\overset{(5)}{=}~2C\int_{0}^{1}y^{1/q-1}dy ~=~2C ~<~ \infty,
\end{align*}
as desired, where (1) holds by \eqref{eq:integrab4} and \eqref{eq:integrab7}, (2) and (4) by \eqref{eq:integrab1}, (3) by $F_{V}(\tilde{x}a_{K}+b_{K})\leq 1$ and $K_{j}\geq K \geq 2$, and (5) by the change of variables $y=F_{U(K_{j}-1)}( \tilde{x}a_{K}+b_{K}) $.
\end{proof}

\begin{lemma}\label{lem:interval1}
For any $h\in \Sigma $, assume that
\begin{equation}
\Big\{ y\in \mathbb{R}:\int_{\Xi }\kappa _{\xi }(h)f_{\tilde{\mathbf{Z}} |\xi }(h)dW(\xi )\leq \int_{\Xi }f_{(Y_{\mu },\tilde{\mathbf{Z}})|\xi }(y,h)d\Lambda (\xi )\Big\}~ =~[ A( h) ,B( h)] ,
\label{eq:interval1}
\end{equation}
where $B( h) -A( h) :\Sigma \to \mathbb{R} _{+}$ is continuous. Then,
\begin{equation*}
\tilde{U}( h) ~=~\Big\{ y\in \mathbb{R}:\int_{\Xi }\kappa _{\xi }(h)f_{\tilde{\mathbf{Z}}|\xi }(h)dW(\xi )\leq \int_{\Xi }f_{(Y_{\mu }, \tilde{\mathbf{Z}})|\xi }(y,h)d\Lambda (\xi )\Big\} .
\end{equation*}
\end{lemma}
\begin{proof}
It suffices to show that $[ A( h) ,B( h) ] $ satisfies conditions (a)-(c) of Theorem \ref{thm:CI_K}.

\underline{Part (a)}. Consider the following derivation.
\begin{align*}
P( \partial \{(y,h)\in \mathbb{R}\times \Sigma :y\in \tilde{U} (h)\}) &~\overset{(1)}{\leq} ~
\int  P_{\xi }( Y_{\mu } \in \{A( h),B(h)\} \vert \mathbf{\tilde{Z}}=h) f_{\tilde{\mathbf{Z}}|\xi }(h)dh ~\overset{(2)}{=}~0,
\end{align*}
where (1) holds by \eqref{eq:interval1}, which implies that $\partial \{(y,h) \in \mathbb{R}\times \Sigma :y\in \tilde{U}(h)\}
~\subseteq~ \{ (A( h) ,h),(B( h) ,h):h\in \Sigma\}$, and (2) because $\{ Y_{\mu }|\mathbf{ \tilde{Z}}=h;\xi \} $ is continuously distributed.

\underline{Part (b)}. Fix $h\in \Sigma $ arbitrarily. Since 
\begin{equation*}
\lim_{\vert y\vert \to \infty }\int_{\Xi }f_{(Y_{\mu },\tilde{\mathbf{Z}} )|\xi }(y,h)d\Lambda (\xi )=0<\int_{\Xi }\kappa _{\xi }(h)f_{\tilde{\mathbf{Z }}|\xi }(h)dW(\xi ), 
\end{equation*}
$\exists y( h) \in ( 0,\infty ) $ such that  $\max \{ \vert B( h) \vert ,\vert A( h) \vert \} \leq y( h) $ for all $ \vert y\vert >y( h) $. Then, $[ A( h) ,B( h) ] \subseteq [ -y( h) ,y( h) ] ,$ and so $\lg ([ A( h) ,B( h) ] )\leq 2y( h) <\infty $.

\underline{Part (c)}. Under \eqref{eq:interval1}, $\lg ([ A( h) ,B( h) ] )=B( h) -A( h) $ which is assumed continuous.
\end{proof}

\begin{lemma}\label{lem:sp-surplus}
Assume that $\xi<1$. Let $Z_1 = H_{\xi} (E_{1}) $ and $Z_2=H_{\xi }( E_{1}+E_{2}) $ with $E_{1},E_{2}$ i.i.d.\ standard exponential random variables and $H_{\xi }$ as in \eqref{eq:H_function}. Then,
\begin{equation*}
    E[Z_1 - Z_2] ~=~\Gamma (1-\xi ),
\end{equation*}
where $\Gamma$ is the standard Gamma function.
\end{lemma}
\begin{proof}
The support of $(Z_{1},Z_{2})$ is $S_{\xi }=\{(x_{1},x_{2}):x_{1}\geq x_{2},~\xi x_{1}\geq -1,~\xi x_{2}\geq -1\}$. For any $(x_{1},x_{2}) \in S_{\xi}$, the PDF of $(Z_{1},Z_{2})$ is 
\begin{equation*}
f_{Z_{1},Z_{2}|\xi }(x_{1},x_{2})~=~
\left\{ 
\begin{array}{lc}
(1+\xi x_{1})^{-1/\xi -1}(1+\xi x_{2})^{-1/\xi -1}\exp ( -( 1+\xi x_{2}) ^{-1/\xi })  & \text{if }\xi \neq 0, \\ 
\exp({-x_{1}})\exp({-x_{2}})\exp({-\exp({-x_{2})})} & \text{if }\xi =0,%
\end{array}%
\right. 
\end{equation*}%
The desired result follows from this formula. We show only the result for $\xi < 0$, as the results for the other two cases are analogous. 
\begin{align*}
&E[ Z_{1}-Z_{2}] \\
&=~\int_{-\infty }^{-1/\xi}\int_{x_{2}}^{-1/\xi }( x_{1}-x_{2}) (1+\xi x_{1})^{-1/\xi-1}(1+\xi x_{2})^{-1/\xi -1}\exp ( -( 1+\xi x_{2}) ^{-1/\xi}) dx_{1}dx_{2} \\
&\overset{(1)}{=}~\frac{1}{\xi ^{3}}\int_{\infty }^{0}\int_{t_{2}}^{0}(t_{1}-t_{2}) t_{1}^{-1/\xi -1}t_{2}^{-1/\xi -1}\exp (-t_{2}^{-1/\xi }) dt_{1}dt_{2} \\
&\overset{(2)}{=}~\int_{0}^{\infty }( \frac{1}{( 1-\xi ) }) \exp ( -v) v^{-\xi +1}dv \\
&\overset{(3)}{=}~\Gamma ( 1-\xi ) ,
\end{align*}%
as desired, where (1) holds by the change of variables $t_{1}=1+\xi x_{1}$ and $t_{2}=1+\xi x_{2}$, (2) by the change of variables $v=t_{2}^{-1/\xi }$, and (3) by $( 1-\xi ) \Gamma ( 1-\xi ) =\Gamma ( 2-\xi ) $.
\end{proof}

\begin{lemma}\label{lem:sp-revenue}
Assume that $\xi<1$. Let $Z=H_{\xi }( E_{1}+E_{2}) $ with $E_{1},E_{2}$ i.i.d.\ standard exponential random variables and $H_{\xi }$ as in \eqref{eq:H_function}. Then, 
\begin{equation*}
    E[ Z] ~=~\bigg\{ 
\begin{array}{ll}
(\Gamma (2-\xi )-1)/\xi  & \text{if }\xi \not=0 \\ 
\bar{\gamma}-1 & \text{if }\xi =0
\end{array}
\end{equation*}
where $\Gamma $ is the Gamma function and $\bar{\gamma}\approx 0.577$ is the Euler constant.
\end{lemma}
\begin{proof}
Note that the PDF of $Z$ is
\begin{equation*}
f_{Z|\xi}( z)~=~\bigg\{ 
\begin{array}{ll}
( 1+\xi z) ^{-{(2+\xi) }/{\xi }}\exp ( -( 1+\xi z) ^{-1/\xi }) 1[ 1+\xi z\geq 0]  & \text{if }\xi \neq 0 \\ 
\exp ( -2z) \exp ( -\exp ( -z) )  & \text{if }\xi =0.
\end{array}
\end{equation*}
We only show the result for $\xi  \not=0$, as the result for $\xi =0$ is analogous.
\begin{align*}
E[ Z]  &~=~\int_{-\infty }^{\infty }z( 1+\xi z) ^{-\tfrac{2+\xi }{\xi }}\exp ( -( 1+\xi z) ^{-1/\xi }) 1[ 1+\xi z\geq 0] dz \\
&~\overset{(1)}{=}~\big( \int_{0}^{\infty }( t^{1-\xi }-t) \exp ( -t) dt\big) /\xi  \\
&~\overset{(2)}{=}~( \Gamma ( 2-\xi ) -1) /\xi ,
\end{align*}
where (1) holds by the change of variables $u=( 1+\xi z) ^{-1/\xi }$ and (2) by $\Gamma(0)=1$. 
\end{proof}

 \begin{lemma}\label{lem:integralConv}
 Assume \eqref{eq:DomainOfAttraction} holds. Also, assume that $\xi<1$ and $E[|V|^{1+\varepsilon}]<\infty$ for some $\varepsilon>0$. For any sequence $\{ x_{K}:K\in \mathbb{N} \} $ with $x_{K}\to x\in \bar{S}_{\xi} \equiv \{ s:G_{\xi }(s) >0\} $, $\lim_{K \to \infty} L_{K}( x_{K}) = L( x)$, where
 \begin{align}
 L_{K}( x) ~\equiv~ x - \tfrac{\int_{{(v_L-b_K)/a_K} }^{x}{F_{V}(ha_{K}+b_{K})^{K}}dh}{F_{V}( xa_{K}+b_{K}) ^{K}}~~~\text{and}~~~L( x) ~\equiv~ x -\tfrac{\int_{-\infty }^{x}G_{\xi }(h) dh}{G_{\xi }( x) }.\label{eq:LKdefn}
\end{align}
\end{lemma}
\begin{proof}
Throughout this proof, it is relevant to note that ${F_{V}(ha_{K}+b_{K})^{K}}$ is the CDF of $(V_{(1)}-b_{K})/a_{K}$ and $G_{\xi }$ is the CDF of $H_{\xi }(E_{1})$, where $V_{(1)}$ denotes the sample maximum of $K$ random draws from $F_{V}$, $H_{\xi }$ and $E_{1}$ are as in Lemma \ref{lem:joint}. 

As a preliminary step, we show that for any $x\in \mathbb{R}$, 
\begin{equation}
\lim_{K \to \infty} \int_{-\infty }^{x}({F_{V}(ha_{K}+b_{K})^{K}}-G_{\xi }(h))dh~= ~0 .  \label{eq:L0}
\end{equation}
To this end, consider the following argument for any $K$.
\begin{align}
&\int_{-\infty }^{x}({F_{V}(ha_{K}+b_{K})^{K}}-G_{\xi }(h))dh~  \notag \\
&\overset{(1)}{=} 
\bigg[\begin{array}{c}
x({F_{V}(xa_{K}+b_{K})^{K}}-G_{\xi }(x))+\\
{{E[H_{\xi }(E_{1})1}[ {H_{\xi }(E_{1})\leq x}] ]}-{{E[\tfrac{V_{(1)}-b_{K}}{a_{K}}1}[\tfrac{V_{(1)}-b_{K}}{a_{K}}\leq x] ]}
\end{array}\bigg].  \label{eq:L2}
\end{align}
where (1) holds by $\lim_{h \to -\infty}{F_{V}(ha_{K}+b_{K})^{K}}=\lim_{h \to -\infty}G_{\xi }(h)=0$ and integration by parts.

Lemma \ref{lem:joint} implies that as $K\to \infty $, 
\begin{equation}
\tfrac{{V_{(1)}-b_{K}}}{{a_{K}}}~\overset{d}{\to }~H_{\xi }(E_{1}).
\label{eq:L2p2}
\end{equation}
Since $H_{\xi }(E_{1})$ is continuously distributed, \eqref{eq:L2p2} implies that
\begin{equation}
\lim_{K\to \infty }x({F_{V}(xa_{K}+b_{K})^{K}}-G_{\xi }(x))=0
\label{eq:L2p3}
\end{equation}
Also, note that \eqref{eq:L2p2} and the continuous mapping theorem imply that as $K\to \infty $, 
\begin{equation}
{\tfrac{V_{(1)}-b_{K}}{a_{K}}1}[ \tfrac{V_{(1)}-b_{K}}{a_{K}}\leq x] ~\overset{d}{\to }~{H_{\xi }(E_{1})1}[ {H_{\xi
}(E_{1})\leq x}] .  \label{eq:L2p5}
\end{equation}
Let $\delta >0$ and $\bar{K}\in \mathbb{N}$ be as in Lemma \ref{lem:finite_muK2}. Then, for all $K\geq \bar{K}$, 
\begin{equation}
E[|{\tfrac{V_{(1)}-b_{K}}{a_{K}}1}[ \tfrac{V_{(1)}-b_{K}}{a_{K}}\leq x] |^{1+\delta }]~\leq ~E[|\tfrac{{V_{(1)}-b_{K}}}{{a_{K}}}|^{1+\delta }]~<~\infty .  \label{eq:L2p8}
\end{equation}
Then, \eqref{eq:L2p5} and \eqref{eq:L2p8} imply that 
\begin{equation}
\lim_{K\to \infty }{E[\tfrac{V_{(1)}-b_{K}}{a_{K}}1}[ \tfrac{V_{(1)}-b_{K}}{a_{K}}\leq x]] ~=~{E[H_{\xi }(E_{1})1}[ {H_{\xi}(E_{1})\leq x}]].  \label{eq:L3}
\end{equation}
Finally, note that \eqref{eq:L0} follows from \eqref{eq:L2}, \eqref{eq:L2p3}, and \eqref{eq:L3}.

As a second preliminary step, we note that
\begin{equation}
\int_{-\infty }^{x}G_{\xi }(h) dh~ \overset{(1)}{=}~
xG_{\xi }( x) -\int_{-\infty }^{x}hg_{\xi }(h) dh~\leq~ xG_{\xi }( x) +E[ \vert H_{\xi }(E_{1})\vert ] ~ \overset{(2)}{<}~\infty ,
\label{eq:L4}
\end{equation}
where (1) holds by integration by parts and $\lim_{x\to -\infty }xG_{\xi }( x) =0$, and (2) by $\xi <1$.

We are now ready to show the desired results. For any $x\in \bar{S}_{\xi}$, consider the following argument:
\begin{align}
\vert L_{K}( x_{K}) -L( x) \vert 
&\overset{(1)}{=}~\bigg\vert \tfrac{\int_{-\infty }^{x_{K}}{F_{V}(ha_{K}+b_{K})^{K}}dh}{ F_{V}( x_{K}a_{K}+b_{K}) ^{K}}-\tfrac{\int_{-\infty }^{x}G_{\xi }(h) dh}{G_{\xi }( x) }\bigg\vert \notag \\
&\overset{(2)}{\leq}~\tfrac{\Bigg\vert 
\begin{array}{c}
 \vert x_{K}-x\vert +\vert \int_{-\infty }^{x}( {F_{V}(ha_{K}+b_{K})^{K}}-G_{\xi }(h) ) dh\vert  \\ 
+( \int_{-\infty }^{x}G_{\xi }(h) dh) \sup_{y\in \mathbb{R}}\vert F_{V}( ya_{K}+b_{K}) ^{K}-G_{\xi }( y)\vert  \\ 
+( \int_{-\infty }^{x}G_{\xi }(h) dh) \vert G_{\xi }( x_{K})-G_{\xi }( x) \vert 
\end{array}
\Bigg\vert }{
G_{\xi }( x)[ G_{\xi }( x)-\sup_{y\in \mathbb{R}}\vert F_{V}( ya_{K}+b_{K}) ^{K}-G_{\xi }( y)\vert -\vert G_{\xi }( x_{K}) -G_{\xi }( x)\vert] },\label{eq:L5}
\end{align}
where (1) holds because $F_V(v)=0$ for $v<v_L$, 
and (2) by $\sup_{y\in \mathbb{R} }{F_{V}(ya_{K}+b_{K})^{K}\leq 1}$ and $G_{\xi }( x) \leq 1$. As $K\to \infty$, we can show that the numerator and the denominator of the right-hand side of \eqref{eq:L5} converge to zero and $G_{\xi }( x)^2 >0$, respectively. This conclusion relies on $x \in \bar{S}_{\xi}$ (and so $G_{\xi }( x)^2 >0$), $x_{K}\to x\in \bar{S}_{\xi}$ as $K\to \infty$, the continuity of $G_{\xi }$, \eqref{eq:L0}, \eqref{eq:L4}, and \citet[Lemma 2.11]{vandervaart:1998}. 
From this and \eqref{eq:L5}, the desired result follows.
\end{proof}

\begin{lemma}\label{lem:joint_first}
Assume \eqref{eq:DomainOfAttraction} holds. Assume that $\xi<1$ and that $E[ |V_{i,j}|^{1+\varepsilon}] <\infty $ for all $i=1\dots,K_j$ in auction $j=1,\dots,N$ for some $\varepsilon >0$. For any $n\in \mathbb{N}$ and as $K\to \infty$,
\begin{equation*}
\big\{\tfrac{P_{j}-b_{K}}{a_{K}}:j=1,\dots,n\big\}~\overset{d}{\to }~\big\{H_{\xi }( E_{1,j}) -\tfrac{ \int_{-\infty}^{H_{\xi }( E_{1,j})} G_{\xi }(h) dh }{G_{\xi }( H_{\xi }( E_{1,j}))} :j=1,\dots,n\big\} ,
\end{equation*}
with $\{(a_{K},b_{K}) \in \mathbb{R}_{++} \times \mathbb{R} :K \in \mathbb{N}\}$, $\{E_{1,j}:~j=1,\dots,n\}$, and $H_{\xi}$ as in Lemma \ref{lem:joint}.
\end{lemma}
\begin{proof}
Since auctions are independent, it suffices to prove that for any auction $j=1,\ldots n$, as $K\to \infty$,
\begin{equation}
\tfrac{P_{j}-b_{K}}{a_{K}}~\overset{d}{\to }~H_{\xi }( E_{1,j}) -\tfrac{\int_{-\infty}^{H_{\xi }( E_{1,j}) }G_{\xi }(h) dh}{G_{\xi }( H_{\xi }( E_{1,j}) ) }, \label{eq:aux1}
\end{equation}%
where $E_{1,j}$ is a standard exponential random variable. 
To this end, consider the following argument.
\begin{align}
\tfrac{P_{j}-b_{K}}{a_{K}}
&=~\big( \tfrac{V_{(1),j}-b_{K_{j}-1}}{a_{K_{j}-1}}-\tfrac{( \int_{{v_L}}^{V_{(1),j}}{F_{V}(u)^{K_{j}-1}}du) /a_{K_{j}-1}}{ F_{V}(V_{(1),j})^{K_{j}-1}}\big) \tfrac{a_{K_{j}-1}}{a_{K}}+\tfrac{ b_{K_{j}-1}-b_{K}}{a_{K}} \nonumber \\
&\overset{(1)}{=}~\big( L_{K_{j}-1}\big( \tfrac{V_{(1),j}-b_{K_{j}-1}}{ a_{K_{j}-1}}\big) \big) \tfrac{a_{K_{j}-1}}{a_{K}}+\tfrac{b_{K_{j}-1}-b_{K} }{a_{K}}, \label{eq:first1}
\end{align}
where (1) holds by the change of variables $u=ha_{K_{j}-1}+b_{K_{j}-1} $ and $L_{K}$ as in \eqref{eq:LKdefn}.

Since $( K_{j}-1) /K\to 1$ as $K\to \infty $, the same argument as Lemma \ref{lem:joint} implies that
\begin{equation}
\lim_{K \to \infty}\big( \tfrac{a_{K_{j}-1}}{a_{K}},\tfrac{b_{K_{j}-1}-b_{K}}{a_{K}}\big) ~=~ ( 1,0) . \label{eq:first2}
\end{equation}

Next, notice that Lemma \ref{lem:joint}
and \eqref{eq:first2} imply that as $K\to \infty$,
\begin{equation}
\tfrac{V_{(1),j}-b_{K_{j}-1}}{a_{K_{j}-1}}~=~\tfrac{V_{(1),j}-b_{K}}{a_{K}}\tfrac{ a_{K}}{a_{K_{j}-1}}+\tfrac{b_{K}-b_{K_{j}-1}}{a_{K}}~\overset{d}{\to }~ H_{\xi }( E_{1}). \label{eq:first3}
\end{equation}
In addition, $H_{\xi }( E_{1}) \in \bar{S}_{\xi} \equiv \{x \in \mathbb{R} : G_\xi(x)>0 \}$. By \eqref{eq:first3}, Lemma \ref{lem:integralConv}, and the extended continuous mapping theorem \citep[e.g.][Theorem 1.11.1]{vandervaart:1998}, as $K\to \infty$,
\begin{equation}
L_{K_{j}-1}\big( \tfrac{V_{(1),j}-b_{K_{j}-1}}{a_{K_{j}-1}}\big)~ \overset{d}{\to }~L( H_{\xi }( E_{1}) ) ~=~O_{p}( 1), \label{eq:first4}
\end{equation}
where $L$ is as in \eqref{eq:LKdefn}. Then, \eqref{eq:aux1} follows from \eqref{eq:first1}, \eqref{eq:first2}, and \eqref{eq:first4}.
\end{proof}

\begin{theorem} \label{thm:CI_K_fp} 
Assume \eqref{eq:DomainOfAttraction} holds, and that for some $\varepsilon >0$ with $(1+\varepsilon)\xi<1$, $E[|V_{i,j}|^{1+\varepsilon}]<\infty $ for all $i=1\dots,K_j$ in auction $j=1,\dots,N$. Finally, assume that the CI for $\mu _{K}$, $U(\mathbf{P})$, is as in \eqref{eq:CI_muK_defn2_fp} with $\tilde{U}:\Sigma\to \mathcal{P}(\mathbb{R}) $ that satisfies the following conditions:
\begin{enumerate}[(a)]
\item $P_{\xi}(\{Y_{\mu },\mathbf{\tilde{X}}\}\in \partial \{(y,h)\in \mathbb{R} \times \Sigma:y\in \tilde{U}(h)\})=0$, where $\partial A$ denotes the boundary of $A$.
\item $\lg (\tilde{U}(h))<\infty $ for any $h\in \Sigma$, where $ \lg (A)$ denotes the length of $A$ (i.e., $\lg (A)\equiv \int \mathbf{1} [y\in A]dy$).
\item For any sequence $\{h_{\ell}\in \Sigma\}_{\ell\in \mathbb{N}}$ with $ h_{\ell}\to h\in \Sigma$, $\lg (\tilde{U}(h_{\ell}))\to \lg (\tilde{U}(h))$.
\end{enumerate}
Then, as $K\to \infty $,
\begin{enumerate}
\item $P(\mu _{K}\in U(\mathbf{P}))~\to ~P_{\xi }(Y_{\mu }\in \tilde{ U}(\mathbf{\tilde{X}}))$,
\item $E[\lg (U(\mathbf{P}))]/a_{K}~\to ~E_{\xi }[\kappa _{\xi }( \mathbf{\tilde{X}})\lg (\tilde{U}(\mathbf{\tilde{X}}))]$.
\end{enumerate}
\end{theorem}
\begin{proof}
This proof follows from the same arguments used to prove Theorem \ref{thm:CI_K}. The main difference between the proofs is that we replace $\{(P_j-b_K)/a_K:j=1,\dots,N\}$ with $\{L_{K_j}((V_{(1),j}-b_K)/a_K):j=1,\dots,N\}$ with $L_{K_j}$ as in \eqref{eq:LKdefn} (instead of $\{(V_{(2),j}-b_K)/a_K:j=1,\dots,N\}$ as in second-price auctions). Then, the vectors $\{Z_j:j=1,\dots,n\}$ and $\tilde{\bf Z}=\{\tilde{Z}_{j}:j=1,\dots ,N\}$ are replaced by $\{X_j:j=1,\dots,N\}$ and $\mathbf{\tilde{X}}=\{\tilde{X}_{j}:j=1,\dots ,N\}$, as in \eqref{eq:aux_RV} and \eqref{eq:Xtilde}, respectively. With these changes in place, the results follow from repeating the steps used to prove Theorem \ref{thm:CI_K}.
\end{proof}

\begin{theorem}\label{thm:CI_pi_K_fp}
Assume \eqref{eq:DomainOfAttraction} holds, and that for some $\varepsilon >0$ with $(1+\varepsilon)\xi<1$, $E[|V_{i,j}|^{1+\varepsilon}]<\infty $ for all $i=1\dots,K_j$ in auction $j=1,\dots,N$. Finally, assume that the CI for $\pi _{K}$, $U(\mathbf{P})$, is as in \eqref{eq:CI_piK_defn2_fp} with $\tilde{U}:\Sigma\to \mathcal{P}(\mathbb{R}) $ that satisfies the following conditions:
\begin{enumerate}[(a)]
\item $P_{\xi}(\{Y_{\pi},\mathbf{\tilde{X}}\}\in \partial \{(y,h)\in \mathbb{R} \times \Sigma:y\in \tilde{U}(h)\})=0$, where $\partial A$ denotes the boundary of $A$.
\item $\lg (\tilde{U}(h))<\infty $ for any $h\in \Sigma$, where $\lg (A)$ denotes the length of $A$.
\item For any sequence $\{h_{\ell}\in \Sigma\}_{\ell\in \mathbb{N}}$ with $ h_{\ell}\to h\in \Sigma$, $\lg (\tilde{U} (h_{\ell}))\to \lg (\tilde{U}(h))$.
\end{enumerate}
Then, as $K\to \infty $,
\begin{enumerate}
\item $P(\pi _{K}\in U(\mathbf{P}))~\to ~P_{\xi }(Y_{\pi }\in \tilde{ U}(\mathbf{\tilde{X}}))$,
\item $E[\lg (U(\mathbf{P}))]/a_{K}~\to ~E_{\xi }[\kappa _{\xi }( \mathbf{\tilde{X}})\lg (\tilde{U}(\mathbf{\tilde{X}}))]$.
\end{enumerate}
\end{theorem}
\begin{proof}
This proof follows from the same arguments as in Theorems \ref{thm:CI_K} and \ref{thm:CI_K_fp}.
\end{proof}


\begin{theorem}\label{thm:compTest_fp}
Assume \eqref{eq:DomainOfAttraction} holds. In the hypothesis testing problem in \eqref{eq:HT0_fp}, the test defined by \eqref{eq:compTest_fp} satisfies the following properties:
\begin{enumerate}
\item It is asymptotically valid and level $\alpha $, i.e., $\lim_{K\to \infty }E_{\xi_0 }[ \varphi _{K}^{\ast }( \mathbf{P}) ] ~= ~\alpha .$
\item It is asymptotically efficient.
\end{enumerate}
\end{theorem}
\begin{proof}
This proof follows from the same arguments as in Theorems \ref{thm:simpleTest} and \ref{thm:compTest}.
\end{proof}

\begin{proof}[Proof of Lemma \ref{lem:sp-rp}] 
The statement imposes $\xi<1$. We divide the remainder of the argument according to the sign of $\xi$. For convenience, define $H_{1} = H_{\xi}(E_{1,j})$ and $H_{2} = H_{\xi}(E_{1,j}+E_{2,j})$.

 \underline{Case 1:} $\xi \in (0,1)$. In that case,
 \begin{align}
 &1+\xi \pi( \gamma)\notag \\
 &\overset{(1)}{=}~( 1+\xi \gamma) P( H_{2}\leq \gamma\leq H_{1}) +E[ 1+\xi H_{2}|\gamma\leq H_{2}] P(\gamma\leq H_{2})+ (1+\xi v_0)P( H_{1}<\gamma) \notag \\
& \overset{(2)}{=}~\tilde{\gamma}P( \tilde{H}_{2}\leq \tilde{\gamma}\leq \tilde{H} _{1}) +E[ \tilde{H}_{2}|\tilde{\gamma}\leq \tilde{H}_{2}] P( \tilde{\gamma}\leq \tilde{H}_{2})+ \tilde{v}_0 P( \tilde{H}_{1}<\tilde{\gamma} ) \notag \\
& \overset{(3)}{=}~\exp ( -\tilde{r}^{-1/\xi }) \tilde{r}^{1-1/\xi }+\Gamma ( 2-\xi ) -\Gamma ( 2-\xi ,\tilde{r}^{-1/\xi }) +\tilde{v}_{0}\exp ( -\tilde{r}^{-1/\xi }) \notag \\
& \overset{(4)}{=}~
\left\{\begin{array}{c}
\exp ( -( 1+\xi r) ^{-1/\xi }) ( 1+\xi r) ^{1-1/\xi }+\Gamma ( 2-\xi ) -\Gamma ( 2-\xi ,( 1+\xi r) ^{-1/\xi })\\
+( 1+\xi v_{0}) \exp ( -( 1+\xi r) ^{-1/\xi })  
\end{array}\right\}
,\label{eq:L32_1}
\end{align}
where (1) holds by definition of $\pi(\gamma) $, (2) by using $ \tilde{\gamma}=1+\xi \gamma$, $\tilde{v}_0 = (1+\xi v_0)$ and $\tilde{H}_{j}=1+\xi H_{j}$ for $j=1,2$, (3) by computation from \eqref{eq:pdf_e1e2}, and uses $ \Gamma ( a,x) \equiv \int_{x}^{\infty }t^{a-1}\exp ( -t) dt$ to denote the incomplete Gamma function and $\Gamma ( a) =\Gamma ( a,0) $ to denote the complete Gamma function, and (4) holds by replacing back $\tilde{\gamma}=1+\xi \gamma$. The computation that delivers (3) is greatly simplified by the transformation of the random variables $\tilde{H}_{j}=1+\xi H_{j}$ for $j=1,2$, whose joint and marginal PDFs are
\begin{align*}
f_{\tilde{H}_{1},\tilde{H}_{2}}( x_{1},x_{2}) &=x_{1}^{-1/\xi -1}x_{2}^{-1/\xi -1}\exp ( -x_{2}^{-1/\xi }) \xi ^{-2}1[ x_{1}\geq x_{2}\geq 0] , \\
f_{\tilde{H}_{j}}( x) &=x^{-j/\xi -1}\exp ( -x^{-1/\xi }) \xi ^{-1}1[ x\geq 0] ~~~~~\text{for}~ j=1,2.
\end{align*}

Since $\xi>0$, maximizing $\pi(\gamma)$ is equivalent to maximizing $1+\xi \pi(\gamma)$. The corresponding first and second order conditions imply Lemma \ref{lem:sp-rp}.

\underline{Case 2:} $\xi <0$. Despite the change in sign, an analogous derivation shows that \eqref{eq:L32_1} also holds. Since $\xi<0$, maximizing $\pi(\gamma)$ is equivalent to minimizing $1+\xi \pi(\gamma)$. The corresponding first and second order conditions then imply Lemma \ref{lem:sp-rp}.

\underline{Case 3:} $\xi =0$. In that case, 
\begin{align*}
\pi( \gamma) & \overset{(1)}{=}\gamma P( H_{2}\leq \gamma\leq H_{1}) +E [ H_{2}|\gamma\leq H_{2}] P(\gamma\leq H_{2}) \\
& \overset{(2)}{=}\gamma\exp ( -\gamma) \exp ( -\exp ( -\gamma) ) +\int_{\gamma}^{\infty }t\exp ( -2t) \exp ( -\exp ( -t) ) dt ,
\end{align*}
where (1) holds by definition of $\pi( \gamma) $ and (2) by computation from \eqref{eq:pdf_e1e2}. From here, the first and second order conditions imply Lemma \ref{lem:sp-rp}. 
\end{proof}

\subsection{Computational details}\label{sec:appendix-computation}

\subsubsection{Second-price auctions}\label{sec:appendix-computation1}


This section provides computational details for objects introduced in Section \ref{sec:sp}. Throughout this section, we use $\mathbf{\tilde{z}}=(z_{1},\dots ,z_{N})\in \Sigma $ and $N=n-2$.

First, recall that $f_{\mathbf{\tilde{Z}}|\xi }(\mathbf{\tilde{z}})$ is as computed in \eqref{eq:densitySecond}. 
For $\mathbf{\tilde{z}} \in \Sigma$, we can compute:
\begin{equation}
    \kappa _{\xi }(\tilde{\mathbf{z}})f_{\tilde{\mathbf{Z}}|\xi }(\tilde{\mathbf{
z}})=n!(\Gamma (2n-\xi ))\int_{0}^{b(\xi )}s^{n-1}\exp \bigg(
\begin{array}{c}
(\xi -2n)\log \big(\sum_{j=1}^{n}(1+z_{j}\xi s)^{-1/\xi }\big)\\ -(1+\tfrac{2}{ \xi })\big(\sum_{j=1}^{n}\log (1+z_{j}\xi s)\big)
\end{array}
\bigg)ds,
\label{eq:f_aux1}
\end{equation}
where $\kappa _{\xi }(\mathbf{\tilde{z}})=E[Z_{(n)}-Z_{(1)}|\mathbf{\tilde{Z}}=\mathbf{\tilde{z}}]$ and $b(\xi )=-1/\xi $ for $\xi <0$, and $b(\xi )=\infty $ otherwise. For $(y,\mathbf{\tilde{z}}) \in \mathbb{R}_{+} \times \Sigma$, calculations yield:
\begin{align}
     &f_{Y_{\mu },\tilde{\mathbf{Z}}|\xi }(y,\tilde{\mathbf{z}})=\notag\\
     &\tfrac{n!(\Gamma (1-\xi))^{n-1}}{y^{n}}\int_{a_{\mu }(\xi )}^{c_{\mu }(\xi )}\exp \bigg(
\begin{array}{c}
-\sum_{j=1}^{n}(1+\xi s+(\Gamma (1-\xi ))\xi {z_{j}}/{y})^{-1/\xi } \\ -(1+\tfrac{2}{\xi })\sum_{j=1}^{n}\log (1+\xi s+(\Gamma (1-\xi ))\xi z_{j}/{y} )
\end{array}
\bigg)ds,
\label{eq:f_aux2}   
\end{align}
where 
$a_{\mu }(\xi )$ and $c_{\mu }(\xi )$ are defined such that for all $s\in (a_{\mu }(\xi ),c_{\mu }(\xi ))$, we have $1+\xi s+(\Gamma (1-\xi ))\xi/y>0$ and $1+\xi s>0$. For $(y,\mathbf{\tilde{z}}) \in \mathbb{R} \times \Sigma$, more calculations yield:
\begin{align}
    &f_{Y_{\pi },\tilde{\mathbf{Z}}|\xi }(y,\tilde{\mathbf{z}})=\notag\\
    &\tfrac{n!}{|y|^{n}}\int_{a_{\pi }(\xi )}^{c_{\pi }(\xi )}{|\pi-s|}^{n-1}\exp \bigg(
\begin{array}{c}
-\sum_{j=1}^{n}(1+\xi (s+{z_{j}}(\pi -s)/y)^{-1/\xi } \\ 
-(1+\tfrac{2}{\xi })\sum_{j=1}^{n}\log (1+\xi (s+z_{j}(\pi -s)/y)
\end{array}
\bigg)ds,
\label{eq:f_aux3}
\end{align}
where 
$a_{\mu }(\xi )$ and $c_{\mu }(\xi )$ are such that for all $s\in
(a_{\mu }(\xi ),c_{\mu }(\xi ))$, we have $1+\xi (s+(\pi -s)/y>0$, $1+\xi s>0
$, and $(\pi -s)/y>0$. The integrals in \eqref{eq:f_aux1}, \eqref{eq:f_aux2}, and \eqref{eq:f_aux3} can be approximated by Gaussian quadrature.

The CIs in Section \ref{sec:sp} require Lagrange multipliers $\Lambda$ that we compute using the algorithm developed by \citet{elliott2015}. See \citet{mullerwang2017} for another application of this algorithm. We now provide a detailed description of how we implemented this algorithm:
\begin{enumerate}[1)] 
\item Discretize $\Xi$ into a fine grid $\Xi _{M}\equiv \{\xi _{1},\xi _{2},\ldots ,\xi _{M}\}$ between $\xi _{1} = \inf\{\Xi\}$ and $\xi _{M+1} = \sup\{\Xi\}$ (we use $M = 50$ uniformly located points between $\xi _{1}$ and $\xi _{M+1}$). Set $s=1$, and define an arbitrary set of initial positive weights $\lambda^{(s)} = \{\lambda^{(s)}_{m}:m=1,\dots,M\}$ over $ \Xi _{M}$ (we use a uniform weights, i.e., $\lambda^{(1)} = \{1/M,\ldots ,1/M\}$).

\item For each $m=1,\dots ,M$, simulate a large number $B$ of $n$ i.i.d.\ draws the EV distribution with parameter $\xi _{m}$ (we use $B=10,000$). For each $m=1,\dots ,M$ and $b=1,\dots ,B$, the samples is denoted by $\mathbf{Z} _{\xi _{m}}(b) =\{Z_{\xi _{m},1}(b) ,\dots ,Z_{\xi _{m},n}(b) \}$. By applying \eqref{eq:ZsortedScaled} to each sample, we get $N=n-2$ sorted and normalized draws, denoted by $\tilde{\mathbf{Z}}_{\xi _{m}}(b) =\{\tilde{Z}_{\xi _{m},1}(b) ,\dots ,\tilde{Z}_{\xi _{m},N}(b) \}\in \Sigma $.

\item For each $m=1,\dots ,M$, use the random draws in step 2) to approximate the limiting coverage probabilities for parameter $\xi _{m}$ in the following manner:
\begin{enumerate}
\item For the winner's expected utility, we approximate ${P}_{\xi _{m}}({ \Gamma (1-\xi _{m})}/({Z_{(n)}-Z_{(1)}})\in U(\tilde{\mathbf{Z}}))$ with  
\[
\mathbb{\hat{P}}_{m}~\equiv ~\tfrac{1}{B}\sum_{b=1}^{B}1\big(\tfrac{\Gamma (1-\xi _{m})}{Z_{\xi _{m},(1)}(b) -Z_{\xi _{m},(n)}( b) }\in U(\tilde{\mathbf{Z}}_{\xi _{m}}(b) )\big),
\]
where $\Gamma $ is the standard Gamma function and $U(\tilde{\mathbf{Z}} _{\xi _{m}}(b) )$ is as in \eqref{eq:sol_program2}, involving integrals of \eqref{eq:f_aux1} and \eqref{eq:f_aux2}.

\item For the seller's expected revenue, we approximate ${P}_{\xi _{m}}(({ \pi (\xi _{m})-Z_{(1)}})/({Z_{(n)}-Z_{(1)}})\in U(\tilde{\mathbf{Z}})))$, with 
\[
\mathbb{\hat{P}}_{m}~\equiv ~\tfrac{1}{B}\sum_{b=1}^{B}1\big(\tfrac{\pi (\xi_{m})-Z_{\xi _{m},(n)}(b) }{Z_{\xi _{m},(1)}(b) -Z_{\xi _{m},(n)}(b) }\in U(\tilde{\mathbf{Z}}_{\xi _{m}}(b))
\big),
\]
where $\pi (\xi _{m})=(\Gamma (2-\xi _{m})-1)/\xi _{m}$ with $\Gamma $ is the standard Gamma function if $\xi _{m}\not=0$, and $\pi (\xi _{m})=-1+\bar{ \gamma}$ if $\xi _{m}=0$ with $\bar{\gamma}\approx 0.577$ equal to the Euler's constant, and $U(\tilde{\mathbf{Z}}_{\xi _{m}}(b) )$ is as in \eqref{eq:Lambda-rp}, involving integrals of \eqref{eq:f_aux1} and \eqref{eq:f_aux2}.
\end{enumerate}

\item Update the weights by setting $\lambda _{m}^{(s+1)}=\lambda _{m}^{(s)}+\varepsilon ((1-\mathbb{\hat{P}}_{m})-\alpha )$ for all $ m=1,\dots ,M$, where $\alpha $ is the significance level and $\varepsilon >0$ is a small step length (we use $\varepsilon =0.05$). Intuitively, the weight on $\xi _{m}$ is decreased or increased if the CI has overcoverage and undercoverage, respectively.

\item Repeat steps 3)-4) a large number of times $S$ (we use $S=2,000$), to get $\lambda^{(S)}=\{\lambda _{m}^{(S)}:m=1,\dots ,M\}$. The Lagrange multipliers  $\Lambda$ is obtained by interpolating $\lambda^{(S)}$ from $\Xi _{M}$ to $\Xi$.


\end{enumerate}
If the algorithm's tuning parameters are appropriately chosen, the Lagrange multipliers generated by it yield a CI that (a) approximately minimizes asymptotic weighted length (due to the reliance on \eqref{eq:sol_program2} or \eqref{eq:Lambda-rp}), and (b) achieves approximately size control for all $\xi \in \Xi _{M}$, which should extend to all $\xi \in \Xi$ by continuity. Finally, it is worth noting that these Lagrange multipliers need to be determined only once for a given $n$ and set $\Xi$.
The tables of the Lagrange multipliers and the corresponding MATLAB code are available on our \href{https://drive.google.com/file/d/1XGZjG-FaTKMqRK_aQ08EQ0TlG7iAbsHK/view}{website}.

\subsubsection{First-price auctions}

We now consider analogous expressions for first-price auctions. Let $j=1,\dots, N$ denote an arbitrary first-price auction. Relative to the previous section, the main difficulty here is that the PDF of $X_{j}$ in \eqref{eq:aux_RV} does not generally have a closed-form expression and, thus, we do not have an analog of \eqref{eq:densitySecond} for first-price auctions.

To explain this issue, we now provide an implicit formula for the PDF of $X_{j}$ for any arbitrary auction $j=1,\dots, N$. Consider the following argument for $\xi <0$: 
\begin{align}
X_{j}& ~\overset{(1)}{=}~H_{1}-\tfrac{1}{G_{\xi }(H_{1})}\int_{-\infty }^{H_{1}}G_{\xi }(h)dh \nonumber \\
&~\overset{(2)}{=}~H_{1}-\exp ( ( 1+\xi H_{1}) ^{-1/\xi }) \int_{-\infty }^{H_{1}}\exp ( -( 1+\xi h) ^{-1/\xi }) dh \nonumber \\
& ~\overset{(3)}{=}~
H_{1}-\exp ((1+\xi H_{1})^{-1/\xi })\Gamma (-\xi ,(1+\xi H_{1})^{-1/\xi }) \nonumber \\
& ~\overset{(4)}{=}~[{( E_{1,j}) ^{-\xi }-1}]/{\xi }-\exp (E_{1,j})\Gamma (-\xi ,E_{1,j}) \nonumber \\
& ~\overset{(5)}{=}~[{\exp (E_{1,j})\Gamma (1-\xi ,E_{1,j})-1}]/{\xi },\label{eq:Xformula}
\end{align}
where (1) holds by \eqref{eq:aux_RV} and denoting $H_{1}=H_{\xi }(E_{1,j})$ with $\{E_{1,j}:~j=1,\dots,n\}$ and $H_{\xi }$ as specified in Lemma \ref{lem:joint}, (2) by \eqref{eq:G}, (3) by the change of variables $u=( 1+\xi h) ^{-1/\xi }$ and using $\Gamma (a,x)=\int_{x}^{\infty }u^{a-1}\exp ( -u) du$ to denote the upper incomplete Gamma function, (4) by $E_{1,j}=H_{\xi }^{-1}( H_{1}) $, and (5) by $\Gamma (1+a,x)=a\Gamma (a,x)+x^{a}\exp ( -x) $ 
applied to $a=-\xi $ and $x=E_{1,j}$. By a similar derivation for $\xi>0$ and $\xi=0$, we get $X_{j}~=~e_{\xi }(E_{1,j})$ with
\begin{equation}
    e_{\xi }(x)~\equiv~ \bigg\{
\begin{array}{ll}
[ \exp (x)\Gamma (1-\xi ,x)-1]/\xi& \text{ if }\xi \neq 0, \\
-\ln ( x) -\exp ( x) \Gamma (0,x)& \text{ if }\xi =0.
\end{array}
\label{eq:XdensityAppendix}
\end{equation}
We numerically verified that $e_{\xi }( x) $ is decreasing in $x$ for all $\xi $, and so the PDF of $X_j$ can be expressed as follows: 
\begin{equation}
f_{X_j|\xi }(x)~=~-\tfrac{\partial e_{\xi }^{-1}( x) }{\partial x}\exp ( -e_{\xi }^{-1}( x) ) . \label{eq:fx}
\end{equation}
We can use \eqref{eq:fx} to derive an implicit expression for the joint PDF of $\mathbf{\tilde{X}}=\{\tilde{X}_{j}:j=1,\dots ,N\}\in \Sigma $ with $\tilde{X}_{j}$ as in \eqref{eq:Xtilde}. The main difficulty relative to the second-price auction is that neither $ e_{\xi }$ in \eqref{eq:XdensityAppendix} nor its inverse has a closed-form expression, and so evaluating them repeatedly is computationally challenging. For this reason, we do not have a closed-form analog of \eqref{eq:densitySecond} for first-price auctions.

To deal with the aforementioned computational issues, we propose a numerical approximation based on a series expansion of the incomplete Gamma function in \citet[Page 263]{abramowitz/stegun:1964}: for $x$ sufficiently large,
\begin{equation}
\Gamma (1- \xi,x)~\approx~ x^{- \xi}\exp (-x)[ 1+\tfrac{(- \xi) }{x}+ \tfrac{(- \xi)(- \xi-1)}{x^{2}}+ \tfrac{(- \xi)(- \xi-1)(- \xi-2)}{x^{3}}+\dots ] . \label{eq:taylorAS64}
\end{equation}
A first-order approximation based on \eqref{eq:taylorAS64} gives $\ln (\exp (x)\Gamma (1-\xi ,x))\approx -\xi \ln x$. To allow for approximation errors, we propose the following equation:
\begin{equation}
\ln (\exp (x)\Gamma (1-\xi ,x))~\approx~ r_{2}( \xi ) -r_{1}( \xi ) \ln x, \label{eq:XdensityAppendix2}
\end{equation}
where $r_{1}( \xi )$ and $r_{2}( \xi ) $ are functions of $\xi $.\footnote{If the first-order approximation was exact, we would have $r_{1}( \xi ) = \xi$ and $r_{2}( \xi )=0$.} To find suitable values for these functions we fit an OLS regression of $\ln (\Gamma (1-\xi,x)\exp (x))$ on $\ln x$ and a constant for a grid of 50,000 equally-spaced values of $x$ between $10^{-6} $ and $1-10^{-6}$ quantiles of the exponential distribution. We use the OLS coefficients as our values for $ (-r_{1}( \xi ) ,r_{2}( \xi ) )$. By combining \eqref{eq:XdensityAppendix} and \eqref{eq:XdensityAppendix2}, we get
\begin{equation}
e_{\xi }(x)\approx \bigg\{ 
\begin{array}{ll}
\lbrack x^{-r_{1}( \xi ) }\exp (r_{2}( \xi ) )-1]/\xi & \text{ if }\xi \neq 0, \\
( r_{1}( 0) -1) \ln x-r_{2}( 0) & \text{ if }\xi =0.
\end{array}
\label{eq:e_xiAppendix1}
\end{equation}
By \eqref{eq:e_xiAppendix1} and using $r_{3}( \xi ) =\exp (r_{2}( \xi ) /r_{1}( \xi ) )$, we get
\begin{align}
e_{\xi }^{-1}(x)&\approx \bigg\{ 
\begin{array}{ll}
r_{3}( \xi ) (1+\xi x)^{-1/r_{1}( \xi ) } & \text{ if } \xi \neq 0, \\
\exp [(x+r_{2}( 0) )/( r_{1}( 0) -1) ] & \text{ if }\xi =0,
\end{array} \notag\\
\tfrac{\partial e_{\xi }^{-1}(x)}{\partial x}&\approx \bigg\{ 
\begin{array}{ll}
-(1+\xi x)^{-( 1/r_{1}( \xi ) +1) }\xi r_{3}( \xi ) /r_{1}( \xi ) & \text{ if }\xi \neq 0, \\
\exp [(x+r_{2}( 0) )/( r_{1}( 0) -1) ]/( r_{1}( 0) -1) & \text{ if }\xi =0.
\end{array} 
\label{eq:e_xiAppendix2}
\end{align}
By combining \eqref{eq:fx} and \eqref{eq:e_xiAppendix2}, we obtain an approximation of the joint distribution of $\tilde{\mathbf{X}}$ and related functions. Ignoring the case $\xi =0$, then for $\mathbf{\tilde{x}}\in \Sigma $ and $N=n-2$ we get
\begin{align}
    f_{\tilde{\mathbf{X}}|\xi }(\tilde{\mathbf{x}})=\tfrac{n!\Gamma (n)|r_{1}( \xi ) |}{|\xi |}\int_{0}^{b(\xi )}s^{n-2}\exp \bigg(
\begin{array}{c}
-n\ln (r_{3}( \xi ) \sum_{j=1}^{n}(1+\xi \tilde{x} _{j}s)^{-1/r_{1}( \xi ) }) \\
-( 1/r_{1}( \xi ) +1) \sum_{j=1}^{n}\ln (1+\xi \tilde{x }_{j}s)\\
+n\ln (|\xi r_{3}( \xi ) /r_{1}( \xi ) |)
\end{array}
\bigg) ds,
\label{eq:e_xiAppendix3}
\end{align}
where $b(\xi )=-1/\xi $ for $\xi <0$, and $b(\xi )=\infty $ otherwise. This expression is the analog of \eqref{eq:densitySecond} for first-price auctions. Analogously, we have
\begin{align}
    &\kappa _{\xi }(\tilde{\mathbf{x}})f_{\tilde{\mathbf{X}}|\xi }(\tilde{\mathbf{ x}})=\notag\\
    &\tfrac{n!|r_{1}( \xi ) |\Gamma (n-r_{1}( \xi ) )}{ |\xi |}\int_{0}^{b(\xi )}s^{n-1}\exp \Bigg(
\begin{array}{c}
(r_{1}( \xi ) -n)\ln (r_{3}( \xi ) \sum_{j=1}^{n}(1+\xi \tilde{x}_{j}s)^{-1/r_{1}( \xi ) }) \\
-( 1/r_{1}( \xi ) +1) \sum_{j=1}^{n}\ln (1+\xi \tilde{x }_{j}s)\\
+n\ln (|\xi r_{3}( \xi ) /r_{1}( \xi ) |)
\end{array}
\Bigg) ds,
\label{eq:e_xiAppendix4}
\end{align}
where $\kappa _{\xi }(\mathbf{x})=E[X_{( n) }-X_{( 1) }|\mathbf{\tilde{X}}=\mathbf{x}]$. For $(y,\mathbf{\tilde{x}})\in \mathbb{R} _{+}\times \Sigma $, we also have
\begin{align}
    &f_{Y_{\pi },\tilde{\mathbf{X}}|\xi }(y,\tilde{\mathbf{x}})=\notag\\
    &\tfrac{1}{|y|^{n}}\int_{a_{ \pi }(\xi )}^{c_{\pi }(\xi )}|\pi -s|^{n-1}\exp \Bigg(
\begin{array}{c}
-( 1/r_{1}( \xi ) +1) \sum_{j=1}^{n}\ln (1+\xi (s+ \tilde{x}_{j}(\pi -s)/y)) \\
-r_{3}( \xi ) \sum_{j=1}^{n}(1+\xi (s+\tilde{x}_{j}(\pi -s)/y))^{-1/r_{1}( \xi ) }\\
+n\ln (|\xi r_{3}( \xi ) /r_{1}( \xi ) |)
\end{array}
\Bigg) ds,
\label{eq:e_xiAppendix5}
\end{align}
where $a_{\pi }(\xi )$ and $c_{\pi }(\xi )$ are such that for all $s\in (a_{\pi }(\xi ),c_{\pi }(\xi ))$, we have $1+\xi (s+(\pi -s)/y)>0$, $1+\xi s>0$, and $(\pi -s)/y>0$. The integrals in \eqref{eq:e_xiAppendix3}, \eqref{eq:e_xiAppendix4}, and \eqref{eq:e_xiAppendix5} can be approximated by Gaussian quadrature. Given these expressions, we can use the algorithm in Section \ref{sec:appendix-computation1} to compute the CI.

\subsection{Additional Monte Carlo simulations}\label{sec:simulation-fp}

This section provides Monte Carlo simulations for first-price auctions. We first consider the problem of inference on the winner's expected utility using only transaction prices. The parameters of the simulations are as in Section \ref{sec:simulation}.

We consider three CIs for $\mu _{K}$ that are analogous to those used in Section \ref{sec:simulation}:
\begin{enumerate}[(i)]
    \item This is our proposed CI in Section \ref{sec:fp-supplus}. As in Section \ref{sec:simulation}, this method is implemented by constructing $U( \mathbf{P})$ in \eqref{eq:CI_muK_defn2_fp} with $\Xi =[ -1,0.5] $ and $W$ equal to the uniform distribution over this interval. As before, the validity of this CI is based on asymptotics as $K\to\infty$.
    \item A CI based on observing the transaction price and the highest valuation for each auction. These data are infeasible in empirical applications of first-price auctions (valuations are unobserved), but we take it as a benchmark for any method that relies on the traditional asymptotics with $n\to \infty$. These data allows us to compute $D_{j}\equiv V_{(1),j}-P_{j}$ for each auction $j=1,\dots,n$. In turn, this information enables us to test $H_0:\mu_{K} = b$ using the t-statistic $\sqrt{n}(\bar{D}_{n} - b)/s$, where $\bar{D}_{n}=\sum_{j=1}^n D_{j}/n$ and $s^2 = {\sum^{n}_{j=1} ( D_{j}-\bar{D}_{n}) ^{2}/(n-1)}$. Under $H_0$, standard asymptotic arguments imply that the t-statistic converges to a standard normal distribution as $n\to\infty$. One can then construct a CI by inverting the aforementioned hypothesis tests. In contrast to the first CI, the validity of this method relies on standard asymptotic arguments with $n\to\infty$.
    \item A bootstrap-based CI based on observing the highest valuations for each auction. This approach is obviously infeasible in using data from first-price auctions, but it could be feasible with bid data from second-price auctions; cf.\ \cite{menzel2013}. Observing the highest valuations across auctions allows us to identify the distribution of valuations, which, in turn, allows us to identify $\mu_K$ (further details are provided in Section \ref{sec:appendix-MC}). By replacing identified objects with suitable estimators, one can consistently estimate $\mu_K$. In addition, one can repeat this process using bootstrap samples to construct a CI for $\mu _{K}$. The method is implemented with 500 bootstrap samples. The validity of this method relies on standard asymptotic arguments with $n\to\infty$.
\end{enumerate}

Table \ref{tbl:fp-CI} presents the coverage and length of the three CIs. Our proposed CI suffers from a slight undercoverage when valuations are $U(0,3)$ or when $n$ is significantly larger than $K$. As our theory predicts, this undercoverage should diminish as $K$ increases. On the other hand, our proposed CI demonstrates excellent size control for the non-$U(0,3)$ valuations and $K$ is significantly larger than $n$. The other two methods perform analogously to Table \ref{tbl:sp-CI}. The CI based on observing the highest valuation and the transaction price and the asymptotics $n\to \infty$ are relatively shorter and tend to suffer from undercoverage, especially when $n$ is small or valuations are Pareto distributed. In turn, the bootstrap-based CI tends to suffer from significant undercoverage. 

\begin{table}[ht]
\centering
\renewcommand{\arraystretch}{1.5}
\begin{tabular}{lllllllllllll}
\hline\hline
\# Bidders & \multicolumn{4}{c}{$K=10$} & \multicolumn{4}{c}{$K=100$} & \multicolumn{4}{c}{$K\sim U\{90,91,\dots,110\}$} \\ 
\# Auctions & \multicolumn{2}{c}{$n=10$} & \multicolumn{2}{c}{$n=100$} & \multicolumn{2}{c}{$n=10$} & \multicolumn{2}{c}{$n=100$} & \multicolumn{2}{c}{$n=10$} & \multicolumn{2}{c}{$n=100$} \\ \hline
Dist.\ & Cov & Lgth & Cov & Lgth & Cov & Lgth & Cov & Lgth & Cov & Lgth & Cov & Lgth\\ \hline
\multicolumn{13}{c}{Method (i): Our proposed CI, asy.\ w.\ $K\to\infty$} \\ 
$U(0,3)$               & 0.89 & 1.20 & 0.89 & 0.29 & 0.91 & 0.22 & 0.98 & 0.12 & 0.91 & 0.22 & 0.98 & 0.13 \\ 
{$\vert N(0,1) \vert$} & 0.94 & 1.28 & 0.84 & 0.40 & 0.97 & 1.46 & 0.98 & 0.33 & 0.98 & 1.47 & 0.98 & 0.33 \\ 
{$\vert t(20) \vert$}  & 0.95 & 1.29 & 0.86 & 0.46 & 0.98 & 1.30 & 0.97 & 0.42 & 0.98 & 1.34 & 0.98 & 0.43 \\ 
$Pa(0.25) $            & 0.97 & 1.47 & 0.87 & 0.54 & 0.97 & 2.12 & 0.95 & 0.99 & 0.97 & 2.12 & 0.96 & 1.00 \\ 
\hline
\multicolumn{13}{c}{Method (ii): CI based on highest valuations, asy.\ w.\ $n\to\infty$} \\  
$U(0,3)$               & 0.89 & 0.03 & 0.95 & 0.01 & 0.85 & 0.00 & 0.95 & 0.00 & 0.90 & 0.00 & 0.96 & 0.00 \\
{$\vert N(0,1) \vert$} & 0.86 & 0.31 & 0.94 & 0.11 & 0.86 & 0.24 & 0.93 & 0.09 & 0.86 & 0.25 & 0.93	& 0.09\\ 
{$\vert t(20) \vert$}  & 0.88 & 0.40 & 0.93 & 0.14 & 0.83 & 0.37 & 0.92 & 0.13 & 0.85 & 0.36 & 0.94 & 0.13 \\ 
$Pa(0.25) $            & 0.81 & 0.69 & 0.88 & 0.28 & 0.76 & 1.22 & 0.91 & 0.47 & 0.75 & 1.21 & 0.90 & 0.47 \\  
\hline
\multicolumn{13}{c}{Method (iii): CI based on bootstrap \& highest valuations, asy.\ w.\ $n\to\infty$} \\  
$U(0,3)$               & 0.35 & 0.19 & 0.89 & 0.11 & 0.36 & 0.02 & 0.89 & 0.01 & 0.36 & 0.02 & 0.88 & 0.01 \\ 
{$\vert N(0,1) \vert$} & 0.28 & 0.29 & 0.84 & 0.12 & 0.28 & 0.21 & 0.85 & 0.09 & 0.27 & 0.21 & 0.86 & 0.09 \\ 
{$\vert t(20) \vert$}  & 0.29 & 0.34 & 0.85 & 0.14 & 0.27 & 0.28 & 0.82 & 0.12 & 0.26 & 0.28 & 0.86 & 0.12 \\ 
{$Pa(0.25) $}          & 0.20 & 0.34 & 0.69 & 0.20 & 0.23 & 0.61 & 0.69 & 0.35 & 0.23 & 0.61 & 0.69 & 0.35 \\ 
\hline\hline
\end{tabular}
\caption{Empirical coverage frequency (Cov) and length (Lgth) of various CIs for the winner's expected utility $\mu_K$ in first-price auctions. The results are the average of 500 simulation draws and a nominal coverage level of 95\%.}
\label{tbl:fp-CI}
\end{table}

Finally, we consider the problem of hypothesis testing about the tail index using only transaction prices. We focus on testing $H_0:\xi = -1$ vs.\ $H_1:\xi \in \Xi/\{-1\}$, with $\Xi =[-1,0.5] $ and $W$ equal to the uniform distribution over this interval. Table \ref{tbl:fp-test} shows the average rejection rate of the test proposed in Section \ref{sec:HypApplication-fp} over 500 simulations using a desired nominal size of $5\%$. Under $H_0$ (i.e., when valuations are $U(0,3)$), our proposed test controls size as long as the number of bidders is not too small relative to the sample size. Under $H_1$ (i.e., when valuations are $\vert N(0,1) \vert$, $\vert t(20) \vert$, and $Pa(0.25)$), our proposed test exhibits adequate power.

\begin{table}[ht]
\centering
\renewcommand{\arraystretch}{1.5}
\begin{tabular}{lcccccccc}
\hline\hline
\# Bidders & \multicolumn{2}{c}{$K=10$}   & \multicolumn{2}{c}{$K=20$}   & \multicolumn{2}{c}{$K=100$} & \multicolumn{2}{c}{ $K\sim U\{90,91,\dots,110 \}$  } \\ 
\# Auctions & $n=10$ & $n  = 100$   & $n=10$ & $n  = 100$   & $n=10$ & $n  = 100$  & $n=10$ & $n  = 100$\\ \hline
$U(0,3)$               & 0.06 & 0.09 &  0.04 & 0.06 &  0.04 & 0.05  &  0.03 & 0.09 \\ 
{$\vert N(0,1) \vert$} & 0.25 & 1.00 &   0.22 & 1.00 &  0.28 & 1.00  &  0.29	& 1.00 \\ 
{$\vert t(20) \vert$}  & 0.27 & 1.00 &   0.29 & 1.00 &  0.34 & 1.00  &  0.36 & 1.00 \\ 
$Pa(0.25) $            & 0.54 & 1.00 &   0.55 & 1.00 &  0.55 & 1.00  &  0.62 & 1.00 \\ 
\hline\hline
\end{tabular}
\caption{Empirical rejection rate of the test proposed in Section \ref{sec:HypApplication-fp} for $H_0:\xi = -1$ in first-price auctions. The results are the average of 500 simulation draws and a nominal size of 5\%.}\label{tbl:fp-test}
\end{table}

\subsection{Auxiliary derivations related to the Monte Carlo simulations}\label{sec:appendix-MC}

We first provide auxiliary derivations related to the third CI considered in Section \ref{sec:simulation}. In a second-price auction $j=1,\ldots ,n$ with $K$ bidders, the CDF of the transaction price satisfies:
\begin{equation}
F_{P_{j}}(x)~\overset{(1)}{=}~F_{V_{(2),j}}( x) ~\overset{(2)}{=}~F_{V}(x)^{K}+KF_{V}(x)^{K-1}( 1-F_{V}( x) ) , \label{eq:functional1}
\end{equation}
where (1) holds by \eqref{eq:secondprice} and (2) by the fact that bidders' valuations are i.i.d. The previous equation expresses the CDF of the transaction price in terms of the CDF of valuations. By inverting this mapping, we can derive a relationship between the distributions of the transaction prices and the valuations. Furthermore, we have the following connection between the distribution of valuations and the winner's expected utility:
\begin{align}
\mu_K &~\overset{(1)}{=}~E[ V_{( 1) ,j}-P_{j}] \notag\\
&~\overset{(2)}{=}~E[ V_{( 1) ,j}-V_{( 2) ,j}] \notag\\
&~\overset{(3)}{=}~K\int_{0}^{\infty } xF_{V}( x) ^{K-1}f_{V}( x) dx-( K-1) K\int_{0}^{\infty } x( f_{V}( x) ( 1-F_{V}( x) ) F_{V}( x) ^{K-2}) dx \nonumber \\
&~\overset{(4)}{=}~K\bigg( \int_{0}^{\infty }( 1-F_{V}^{K}(x)) dx-\int_{0}^{\infty }( 1-F_{V}^{K-1}(x)) dx\bigg) ,\notag
\end{align}
where (1) holds because the auctions are i.i.d.\ (which holds in all of our Monte Carlo designs), (2) holds by \eqref{eq:secondprice}, (3) by the fact that bidders are i.i.d.\ and valuations are nonnegative (which holds in all of our Monte Carlo designs), and (4) by integration by parts.

To conclude, we provide auxiliary derivations related to the third CI considered in Section \ref{sec:simulation-fp}. Recall that this method presumes that we observe the highest valuation $V_{(1),j}$ for each auction $j=1,\dots,n$. Given that there are $K$ bidders, the CDF of the highest valuation satisfies:
\begin{equation}
F_{V_{( 1) }}( x )~\overset{(1)}{=}~F_{V}(x)^{K}, \label{eq:functional2}
\end{equation}
where (1) holds by the fact that bidders' valuations are i.i.d.\ and $V_{(1)} = \max\{V_1, V_2,\dots, V_K\}$. Furthermore, we can connect the distribution of valuations and the winner's expected utility in the following fashion:
\begin{align*}
\mu_K ~&\overset{(1)}{=}~E\bigg[\tfrac{ \int_{0}^{V_{(1)}}{F_{V}(u)^{K-1}}du}{F_{V}(V_{(1)})^{K-1}}\bigg]\\
~&\overset{(2)}{=}~
\int^{\infty}_{0} \bigg( \int_{0}^{v}K F_{V}^{K-1}( x) dx\bigg) f_{V}(v)dv \\
~&\overset{(3)}{=}~K\int^{\infty}_{0}  ( 1-F_{V}(v) ) F_{V}(v)^{K-1} dv, 
\end{align*}
where (1) holds by \eqref{eq:fp_muK_as_fn_ofV} and by the fact that valuations are nonnegative and auctions are i.i.d.\ (which holds in all of our Monte Carlo designs), (2) by \eqref{eq:functional2}, and (3) by integration by parts.

\newpage 
\bibliographystyle{ecta}
\bibliography{bib_auction}

\end{document}